%% file: nonplanar_jhep.tex
\documentclass[12pt]{article}
\usepackage{jheppub}
\usepackage{graphicx}
\usepackage{epsfig}
\usepackage{amsfonts}
\usepackage{amssymb}
\usepackage{amsmath}
\usepackage{bm}
\usepackage{mathrsfs}
\usepackage{bbm}
\usepackage{cancel}

\newcommand{\nn}{\nonumber \\}
\newcommand{\nb}{\nonumber}
\newcommand{\LP}{\left(}
\newcommand{\RP}{\right)}

\newcommand{\Tr}{\mathop{\rm Tr}\nolimits}

\newcommand{\brac}[2]{\langle#1\ #2 \rangle}
\newcommand{\bracf}[1]{\langle#1 \rangle}

\newtheorem{lemma}{Lemma}
\newtheorem{theorem}{Theorem}

\newenvironment{proof}{Proof:}

\def\endofproof{\hfill{$\Box$}\\}

\title{Permutation Relations of Yangian Invariants, Unitarity Cuts, and  Scattering Amplitudes}
\author{Peizhi Du, Gang Chen\footnote{Corresponding author}
and
Yeuk-Kwan E. Cheung\footnote{Corresponding author}
}
\affiliation{%
Department of Physics, Nanjing University\\
22 Hankou Road, Nanjing 210093, P. R. China
}%

\emailAdd{%
<gang.chern@gmail.com>  <cheung@nju.edu.cn>}

\date{\today}

\abstract{%
We find a permutation relation among Yangian Invariants -- two Yangian Invariants with adjacent external lines exchanged are related by a simple kinematic factor--%
which is shown to be equivalent to U(1) decoupling and Bern-Carrasco-Johansson (BCJ) relation at the level of maximal helicity violating (MHV) amplitudes.
We propose using unitarity cuts to  study nonplanar amplitudes and to systematically reconstruct the integrands of nonplanar MHV amplitudes, up to a rational function which vanishes under all possible  unitarity cuts. This is made possible with the newly found permutation relations by converting nonplanar on-shell diagrams into planar ones.
As explicit examples the construction of one-loop double-trace MHV amplitudes of 4- and 5-point interactions are presented using on-shell diagrams.
The kinematic factors and the resultant planar diagrams are carefully dealt with using the unitarity cut conditions.
The first next-to-MHV amplitudes are addressed using generalized unitarity cuts.  Their leading singularities can be identified as residues of the Grassmanian integral.
These  examples  also serve to demonstrate the power of the newly found relation of Yangian Invariants.
}


\preprint{NITS-PHY-2014002}
\arxivnumber{1401.6610}

\begin{document}

\maketitle
\section{Introduction}

The recent progress in the computation of Yang-Mills scattering amplitudes  has been  exciting.
At  tree level, BCFW recursion
relation~\cite{BCFW04-01, BCFW04-02, BCFW05, 2012FrPhy...7..533F}
can be used to calculate  n-point amplitudes efficiently.
Unitarity cuts~\cite{Bern94, Bern95, Bern05} and generalized unitarity cuts%
~\cite{GUT04-01, GUT05-02, GUT05-06, Drummond:2008bq, Mastrolia:2012wf, Mastrolia:2012du, Mastrolia:2013kca, Multiloop2013}
combined  with  BCFW  for the rational terms work well at loop
level~\cite{Bern06, Carrasco11, Bern:2008ap, Eden1009, Eden1103, Eden11}.
All loop integrands~\cite{NimaAllLoop, NimaLocalInt, 2012JHEP...07..174C, CaronHuot:2011ky} for $\mathcal{N}=4$ super Yang-Mills (SYM) planar amplitudes can be obtained recursively in principle.
On the other hand there is  much  progress on gluon amplitude computation at strong coupling~\cite{Alday07, Alday07-2} via  the celebrated AdS/CFT correspondence.

Besides the progress on calculations interesting and useful relations among color-ordered partial amplitudes have been uncovered.
A relation of such kind was proposed by Bern, Carrasco and Johansson~\cite{BCJ}, the BCJ relation. Together with the KK relation proposed earlier by Kleiss and Kuijf~\cite{KK},  these two relations have since then been widely used to simplify calculations at tree level~\cite{BCJ09, 2013NuPhB.873...65S, 2010PhRvD..82h7702T, BCJ11, BCJ12}.

Lately Arkani-Hamed et al~\cite{NimaGrass} proposed using positive  Grassmannian  to study $\mathcal{N} = 4$ super Yang-Mills
along  with  the  constructions of the bipartite  ribbon on-shell
diagrams~\cite{GeometryOnshell}--in which all internal legs are {\it on shell}--for  planar Yang-Mills interactions.
In such a construction each on-shell bipartite diagram is automatically gauge invariant;  and a direct relationship  between planar  amplitudes in
$\mathcal{N}=4$ super Yang-Mills (SYM) and the  positive Grassmannian structures is presented.
Furthermore they also prescribe a permutation rule for characterizing on-shell diagrams of tree level amplitudes as well as the leading singularities~\cite{Cachazo:2004dr, Cachazo2008,
Cachazo:2008hp, Spradlin08} in planar loop-level amplitudes.

Each on-shell diagram corresponds  to a Yangian invariant, as shown in~\cite{Drummond:2009fd} at tree level and~\cite{Brandhuber:2009kh, NimaSimpleField, Elvang:2009ya} at loop level.
(See~\cite{2003JHEP...10..017D, 2004qts..conf..300D} for earlier works
and~\cite{2014arXiv1401.7274B, Frassek:2013xza,
Amariti:2013ija, Bourjaily2013, CaronHuot:2011ky, Drummond:2010uq, Beisert2010, Drummond:2010qh,  Feng2010, Alday:2010vh, Mason:2009qx, Beisert:2009cs, 2009JHEP...06..045A, Bargheer:2009qu, 2009JHEP...04..120A} for a sample of  interesting  developments thereafter,
and~\cite{Elvang:2013cua, Benincasa:2013faa,
2012LMaPh..99....3B, Drummond:2011ic, Dixon:2011xs,  Beisert:2010jq, Bartels:2011nz, 2011JPhA...44S4011H, 2011JPhA...44S4012B, Roiban:2010kk, Drummond:2010km, 2005IJMPA..20.7189M, 1993IJMPB...7.3517B} for a sample of  reviews and a new book~\cite{Henn:2014yza}.)
Hence the tree-level amplitudes as well as leading singularities~\cite{Cachazo2008} of loop-level amplitudes are invariant under Yangian symmetry, which is a symmetry combining conformal symmetry and dual conformal symmetry~\cite{Drummond:2008vq, Drummond:2008cr,  
Drummond:2009fd, Brandhuber:2009kh, NimaSimpleField}.
And the scattering amplitudes can be obtained by summing  over the underlying Yangian Invariants.
All can be done in  either the  momentum space or the  momentum twistor space~\cite{NimaSMatrix, NimaDuality, NimaGrassDuality}.

On-shell bipartite diagrams fall into  equivalence classes  under square moves and mergers. Such equivalence operations leave the corresponding Yangian Invariants unchanged. However if we only require the corresponding Grassmannian geometry (and hence the $C$ matrix) unchanged under a certain definition, by intuition, there should  exist  new  generators of a new kind of equivalence operations. 
In this paper, we will discuss a class of  such operations  generated  by  a black \& white (B\&W) box--which for the rest of the paper will be called a ``basic box''--leading to  permutation relations in the on-shell diagrams.
These, in turn, induce new relations  among  Yangian invariants.

Another motivation for this work  is to present  a systematic method to construct the local integrands~\cite{NimaLocalInt, LocalInt2013} of Yang-Mills scattering  amplitudes from unitarity cuts for {\it nonplanar}  diagrams. Under each unitarity cut, the integrand of the amplitude is well-defined and can be obtained by  gluing  tree-level amplitudes.  According to the on-shell diagrams of the tree-level amplitudes, we can directly remove the unitarity cut constraints in the frame of on-shell diagram. Then for each unitarity cut we obtain a simple form of the integrand up to a rational function which will vanish on the unitary cut. After introducing a proper operation to combine the integrands for all kinds of the unitarity cuts,  we can get an integrand for an general amplitude up to a rational function which will vanish under all the unitarity cuts.  According to the unitarity  constructible condition for the amplitudes in super Yang-Mills theory~\cite{Bern94, Bern95}, the final ambiguity of the integrand can be fixed by setting  the rational function to zero.  Then we obtain the final form of the integrand for the amplitude. This  method enjoys a  direct  extension  to nonplanar diagrams when combined with the newly found on-shell permutation relations.
In nonplanar diagrams\footnote{By ``nonplanar diagrams'' we mean either the loop line twisted  nonplanar diagrams or the higher loop multi-trace diagrams. This is because both cases are of the same form in on-shell diagrams.}
it is still possible to define an integrand up to rational functions which will vanish under all the unitarity cuts~\cite{BernNP98, BernNP08, BernNP10, BernNP12, Bianchi:2013pfa}.

Unitarity cuts are deployed, in nonplanar diagrams, to help remove the ambiguity in loop momentum definition due to the nonplanar leg(s), as opposed to single cuts used by Arkani-Hamed et al~\cite{NimaGrass} in the construction of planar amplitudes.
Definitions of loop momenta  in the nonplanar loop diagrams under unitarity cuts  will be presented.
The resulted diagrams after a  unitarity  cut of a given one loop nonplanar diagram (by which we mean one-loop double-trace amplitudes)  will be transformed into the corresponding on-shell diagrams.
Using our newly found on-shell permutation relation,
a nonplanar on-shell diagram could be subsequently converted to (a linear combinations of)  planar diagrams with  kinematic functions as coefficients.
One then needs to sum up  the resultant planar diagrams, from all possible unitarity cuts of a given nonplanar diagram, in a proper procedure which we call ``union'' prescribed in Sect.~\ref{Sec:MHV} for MHV amplitudes and in Sect.~\ref{sec:NMHV} for NMHV amplitudes,   to obtain the total nonplanar amplitudes.

The final step of our construction is to use appropriate BCFW bridges to
re-construct the total on-shell diagrams for a given nonplanar diagram.
All possible but inequivalent connections by BCFW bridges need to be taken into account.
The most crucial  step is the discovery of the permutation  relation for bipartite on-shell diagrams  that enable us to  convert nonplanar on-shell  sub-diagrams into planar ones,  which, in turn, enable the straight forward application of the existing techniques developed for on-shell planar diagrams.
This method works is well-adapted for higher loops; and we believe that it  can be generalized to higher-loop nonplanar diagrams (work in progress).

We shall show by explicit computations in Section~\ref{Sec:MHV} that
the total on-shell diagrams constructed by unitarity cuts for MHV nonplanar one-loop amplitudes in ${\mathcal N} =4$ super Yang-Mills give the correct local integrands.
The total on-shell diagrams constructed for NMHV nonplanar one-loop amplitudes by generalized unitarity
cuts~\cite{2011JPhA...44S4003B}
reproduce the correct integrals~\cite{OneloopNMHV2013}, as presented in Section~\ref{sec:NMHV}.

\section{A permutation relation among Yangian Invariants in on-shell diagrams}
\label{sec:permutation}

In dealing with nonplanar amplitudes it is crucial that there be a relation to enable the transformation of nonplanar elements  into planar ones.
To our pleasant surprise there exists such a simple relation, represented pictorially in
Fig.~\ref{fig:YI_4point} below.
\begin{figure}[ht!]
  \centering
 \includegraphics[width=0.60\textwidth]{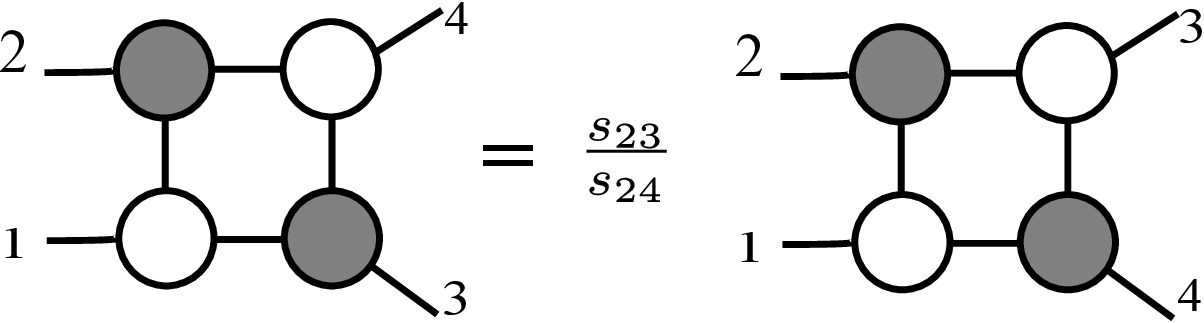}
 \caption{Yangian invariant relation in 4-point tree amplitude.}
 \label{fig:YI_4point}
\end{figure}
A permutation relation relating different Yangian invariants is completely analogous to the  BCJ relation~\cite{BCJ} for 4-point amplitude. We shall henceforth call such transformation    \textit{the permutation relation among Yangian invariants}.

In  $\mathcal N=4$ SYM the amplitudes can be constructed by unitary cuts or generalized unitary cuts~\cite{Bern94, Bern95, Bern05, GUT04-01, GUT05-02, GUT05-06, Drummond:2008bq, Mastrolia:2012wf, Mastrolia:2012du, Mastrolia:2013kca, Multiloop2013}.
After cuts the loop amplitude is a combination of tree level amplitudes.
As we want to transform the nonplanar  amplitudes, after unitarity cut, into planar ones we need to know how the constituent tree amplitudes change under the permutations of legs.
Since a permutation of legs is generated by a pairwise exchange of  two consecutive legs  we only need to know the transformations of the tree amplitudes under  an exchange of  two consecutive legs.

\subsection{A permutation relation of two bipartite boxes}
\label{sec:PRBox}

The  set of rules governing the permutations of external legs  for the  bipartite on-shell diagrams have been introduced in \cite{NimaGrass}.
Let us  take, again, the box  as an example: the 4-point tree amplitude  has only one Yangian invariant.
The corresponding permutation is
$$\left(
\begin{array}{cccc}
  1& 2  &3&4   \\
 3& 4 &5&6
\end{array}
\right),
$$
In a on-shell diagram of a tree-level amplitude a permutation is in one to one correspondence  to a Yangian invariant,  we can therefore use permutations to characterize  Yangian invariants.
Without loss of generality, we take the permuted  external legs to be 3 and 4. It is then easy to see
\begin{eqnarray}
\label{eq:4point}
Y_4^{(2)} (1,2,4,3)&=&\frac{s_{23}}{s_{24}}Y_4^{(2)} (1,2,3,4)
\end{eqnarray}
where $s_{ij}=(p_i+p_j)^2=\left \langle i\ j \right \rangle[i\ j]$ and $Y_4^{(2)} (1,2,3,4)$ is
\begin{eqnarray}
Y_4^{(2)} (1,2,3,4)= \frac{\delta^{2\times 4}(\lambda\cdot\tilde{\eta})\ \delta^{2\times 2} (\lambda\cdot\tilde{\lambda})  }{\brac{1}{2} \brac{2}{3}\brac{3}{4} \brac{4}{1} }.\nn
\end{eqnarray}
This is exactly what has been  depicted in Fig.~\ref{fig:YI_4point} above.

We can now generalize to n-point MHV amplitudes.
The corresponding permutation  is
\begin{eqnarray}
\label{eq:MHVPer}
\left(
\begin{array}{ccccccc}
  1& 2  &\cdots&i&\cdots&n-1&n   \\
 3& 4  &\cdots&i+2&\cdots&n+1&n+2
\end{array}
\right).
\end{eqnarray}
For MHV ($k=2$) amplitudes, $Y_n^{(2)} (1,2,\ldots,n)$ can be written explicitly,
\begin{eqnarray}
\label{eq:MHV}
Y_n^{(2)} (1,2,3,\cdots,n-2,n-1, n)
= \frac{\delta^{2\times 4}(\lambda\cdot\tilde{\eta})\
\delta^{2\times 2} (\lambda\cdot\tilde{\lambda}) }
{\brac{1}{2} \brac{2}{3}\cdots \brac{n-2}{n-1} \brac{n-1}{n}
\brac{n}{1} }\nn
\end{eqnarray}

We, for concreteness, take the permuting  legs to be $n-1$ and $n$. As shown in Fig.~\ref{fig:MHVGBox}, for MHV on-shell diagrams, it is always possible to connect a box directly to the pair of  the permuting legs~\cite{NimaGrass}, 
evident  from  the expression
$Y_n^{(2)}=Y_4^{(2)} \underbrace{\odot Y_3^{(1)}\odot\ldots\odot Y_3^{(1)}}_{n-4}$
together with the cyclic symmetry of the external legs.
This box is nothing but  the 4-point on-shell amplitude.
Using~\ref{eq:4point},
we  obtain a permutation relation for any MHV amplitude
\begin{eqnarray}
\label{eq:npoint}
Y_n^{(2)} (1,2,3,\cdots,n-2,n, n-1)&=&\frac{s_{\widehat{n-2},n-1}}{s_{\widehat{n-2},n}}Y_n^{(2)} (1,2,3,\cdots,n-2,n-1, n).
\end{eqnarray}
The coefficient
$\frac{s_{\widehat{n-2},n-1}}{s_{\widehat{n-2},n}}$
is  obtained as followed (See Fig.~\ref{fig:MHVGBox}.).
\begin{figure}[ht!]
  \centering
 \includegraphics[width=1.0\textwidth]{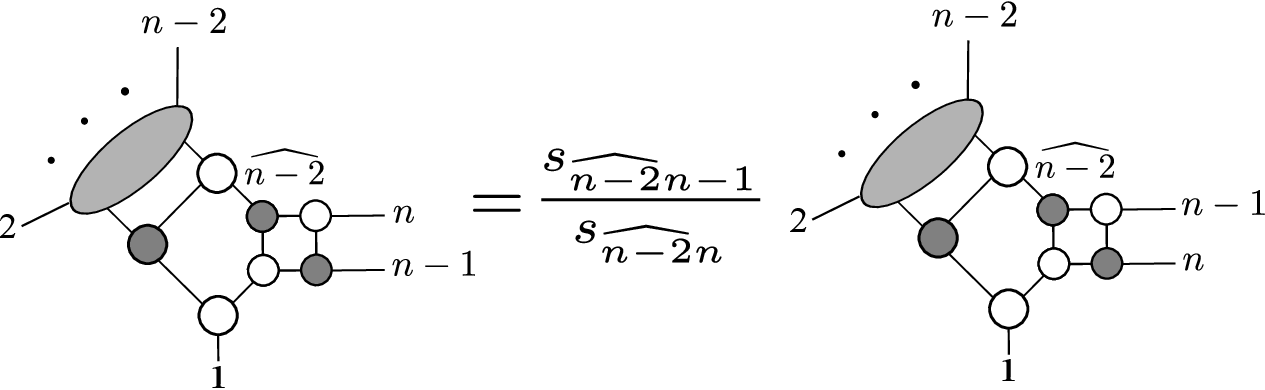}
 \caption{A Yangian invariant relation in 4-point tree amplitude,  where ellipsis  represent the process of  adding ``k-preserving inverse soft factors.''}
 \label{fig:MHVGBox}
\end{figure}
Components other than the box, in  a  bipartite diagram, are just ``k-preserving inverse soft factor $\odot Y_3^{(1)}$''~\cite{NimaAllLoop}. Adding a factor $\odot Y_3^{(1)}$ does not change the spinors $\lambda_{\widehat{n-2}}$ and $\lambda_{\hat 1}$; and    $\lambda_{\widehat{n-2}}=\lambda_{n-2}$
and $\lambda_{\hat 1}=\lambda_{1}$.
Altogether  we  get
\begin{eqnarray}
\label{eq:coeN}
\frac{s_{\widehat{n-2}n-1}}{s_{\widehat{n-2}n}}={\brac{n-2}{n-1}[\widehat{n-2}\ n-1]\over \brac{n-2}{n}[\widehat{n-2}\ n]}=-{\brac{n-2}{n-1}\brac{1}{n}\over \brac{n-2}{n}\brac{1}{n-1}}.
\end{eqnarray}
According to~(\ref{eq:npoint}) and~(\ref{eq:coeN})
the permutation relation we found is consistent with results for MHV amplitudes in the Parke-Taylor~(\ref{eq:MHV})
form~\cite{Parke:1986gb}.

This new permutation relation  holds, furthermore, for an analogous class of Yangian invariants in non-MHV amplitudes.
To aid in the  discovery we first  establish  a  criterion suitable for this class of amplitudes.
Firstly, we should define a modified BCFW-decomposition~\cite{NimaGrass}.
The on-shell diagram can be decomposed by taking  a BCFW bridge away from the diagram leaving only a sub-diagram.
The permutation of the diagram $\sigma$ can  then be decomposed as $(i j)\circ \sigma'$, where $(ij)$ is the permutation of the BCFW bridge on $i, j$ and $\sigma'$ is the permutation of the left sub-diagram.

\paragraph{A BCFW-Bridge decomposition to a Box:}

Starting with a given  permutation
$\sigma$ and picking two consecutive legs $i$ and $i+1$,
if  $\sigma(i)~\neq i ~\text{mod}~n$ and
$\sigma(i+1)~\neq i+1~\text{mod}~n$ and $\sigma$ for the
other legs is not identical to the identity modulus $n$
(a ``dressed'' identity\footnote{As an example, for $n=6$ and $k=3$, a dressed identity is $\{7,8,9,4,5,6\}$.}),
one can decompose $\sigma$ as $(j_1 j_2) \circ \sigma'$,
where $1 \leqslant j_1<j_2 \leqslant n$  and
$\sigma(j_1) < \sigma(j_2)$, with
$j_1\neq \{i, i+1\}$, and  $j_2\neq \{i,  i+1\}$.
The legs $j_1$ and  $j_2$ are being separated only by the unpermuted  legs or leg $i$, or $ i+1$, keeping   the order of $\sigma^{-1}(i), \sigma^{-1}(i+1)$ invariant.
One repeats the process  until $\sigma$ is an  identity for all the legs except the legs $i, i+1$ and
$\sigma(i)~mod~n, \sigma(i+1)~mod~n$.
We denote the final permutation as $\bar\sigma$. \\
If, on the other hand,
\begin{eqnarray}\label{eq:BoxCondition0}
\bar\sigma(i) &<& \bar\sigma(i+1)\nb\\
\bar\sigma^{-1}(i) &<& \bar\sigma^{-1}(i+1),
\end{eqnarray}
it is easy to see that the $\bar\sigma$  corresponds  to a four point amplitude for legs
$(i, i+1, \sigma(i)~mod~n, \sigma(i+1)~mod~n)$.
Furthermore $\sigma$ is obtained by putting BCFW bridges on
$\bar\sigma$.  Hence we conclude that the on-shell diagram corresponding to $\sigma$ can have ``a box'' connecting  directly to these two marked legs.

Moreover, in a  BCFW-Bridge decomposition, we always keep the order of $\sigma^{-1}(i)$ and  $\sigma^{-1}(i+1)$.
Then the condition~(\ref{eq:BoxCondition0}) is equivalent to
\begin{eqnarray}
\label{eq:BoxCondition1}
\sigma(i) &<& \sigma(i+1)\nb\\
\sigma^{-1}(i) &<& \sigma^{-1}(i+1),
\end{eqnarray}
which is a convenient criterion on permutations  to check  whether  a ``box'' can enjoy direct connection  to a pair of adjacent legs in a on-shell  bipartite diagram.

\paragraph{An application:}  In all MHV amplitudes, any two consecutive legs $\{i,i+1\}$ are in ``a box'' due to
$\sigma(i)<\sigma(i+1)$ for the permutation $\sigma$ of MHV amplitude~\ref{eq:MHVPer}.
This observation agrees with~\cite{NimaGrass}.
For general Yangian invariants,  with two consecutive legs $n-1$ and $n$,  if the condition~\ref{eq:BoxCondition1} holds  we  have
\begin{eqnarray}
\label{eq:Ynk}
{Y_{\sigma}}_n^{(k)} (1,2,3,\cdots,n-2,n, n-1)&=&\frac{s_{\widehat{n-2},n-1}}{s_{\widehat{n-2},n}}{Y_{\sigma}}_n^{(k)} (1,2,3,\cdots,n-2,n-1, n)
\end{eqnarray}
when permuting the pair of legs $n-1$ and $n$,  as shown in Fig.~\ref{fig:NMHVGBox}.
\begin{figure}[ht!]
  \centering
 \includegraphics[width=1.0\textwidth]{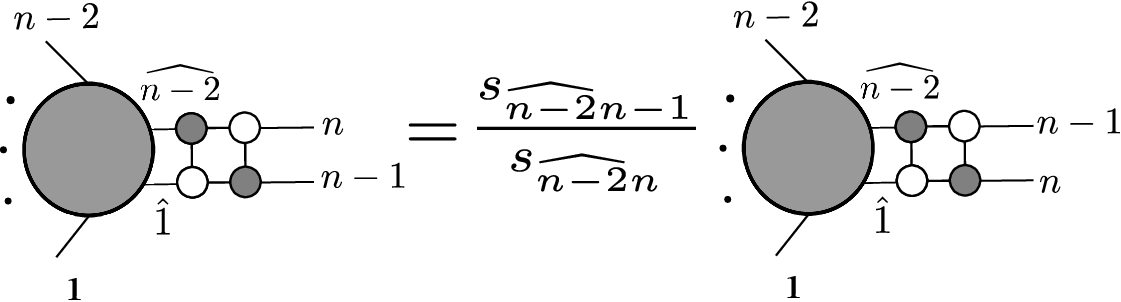}
 \caption{A Yangian invariant relation of  a  4-point tree amplitude, with the dark circle  denoting a general  on-shell diagram.}
 \label{fig:NMHVGBox}
\end{figure}

\paragraph{Conclusion:} For planar  bipartite on-shell diagrams, we can use the criterion~\ref{eq:BoxCondition1} to justify if any two adjacent legs fall  into  a box.

\subsection{Kinematic factors}
Let us now turn our attention to the kinematic factors--the remaining obstacle in the construction of total integrands for nonplanar amplitudes.
How to  deal with the resultant  planar diagrams as well as the concrete steps of reconstruction will be presented in Section~\ref{Sec:MHV} and~\ref{sec:NMHV}.
In this subsection we also study the behavior  of  Yangian invariants $Y_n^k$ and the corresponding Grassmannian cells--the $(k\times n)$-matrices, $C$--under permutations of two external legs.
Each Grassmannian cell $C$ is a point in the Grassmannian,
$G(k,n)$, characterizing the $k$-plane in the $n$-dimensional space.

To compute the kinematic factor
$\frac{s_{\widehat{n-2}n-1}}{s_{\widehat{n-2}n}}$
we only need to determine the  momenta of the two internal lines connecting to  $Y_4^{(2)}$,  which  can be done recursively by BCFW method.
\begin{figure}[ht!]
  \centering
 \includegraphics[width=0.30\textwidth]{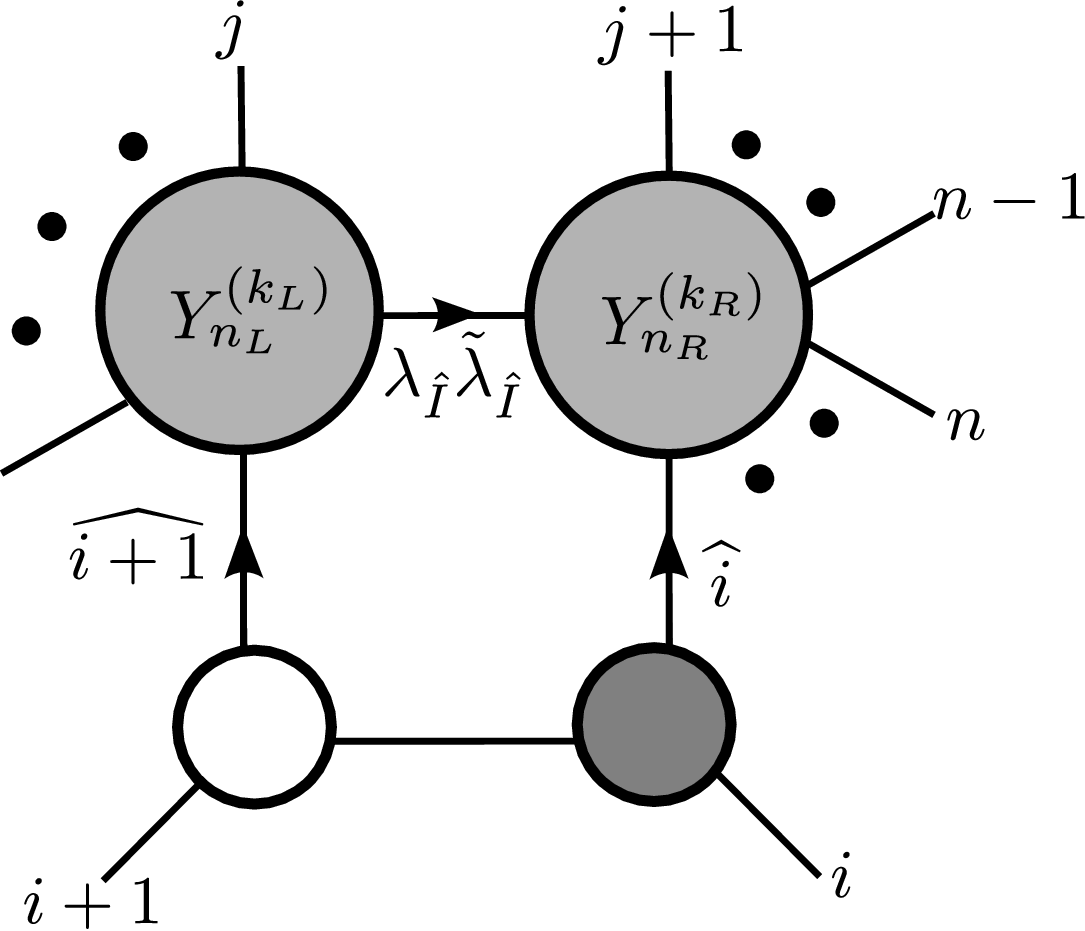}
\end{figure}
We set a variable $\alpha$ to exhibit the  momentum shift between $i+1$ and $i$,
$\alpha\,\lambda_{i+1}\,\tilde{\lambda}_{i}$.
Then the shifted momenta  are
$$\lambda_{\hat{i+1}}\tilde{\lambda}_{\hat{i+1}}
=\lambda_{i+1}(\tilde{\lambda}_{i+1}-\alpha \tilde{\lambda}_ {i})$$
and
$$\lambda_{\hat{i}}\tilde{\lambda}_{\hat{i}}
=({\lambda}_{i}+\alpha{\lambda}_{i+1})\tilde{\lambda}_{i}.$$
The internal momentum connecting  these two legs, $\lambda_I\tilde{\lambda}_{\hat I}$,  is determined by  momentum conservation from the left (or right),
$$\lambda_I\tilde{\lambda}_{I}
=\sum\limits_k \lambda_k\tilde{\lambda}_k+\lambda_{\hat{i+1}}\tilde{\lambda}_{\hat{i+1}}$$
where the index $k$ in the sum runs through all the external momenta on the left.
The variable $\alpha$ can thus be solved by the condition
$\lambda_I\tilde{\lambda}_{I}$  being on-shell
\begin{equation}
\label{eqn:onshellCD}
(\sum\limits_k \lambda_k\tilde{\lambda}_k
+\lambda_{j}\tilde{\lambda}_{j}+\alpha{\lambda}_ i\tilde{\lambda}_j)^2=0.
\end{equation}
The momenta $\lambda_I\tilde{\lambda}_{I}$ and
$\lambda_{\hat j}\tilde{\lambda}_{\hat j}$
are fully  determined by the spinors of the external momenta.
And the  Yangian invariant on the right is, in turn,  determined.

We  repeat the above operation until only a $Y_m^{(2)}$ (a MHV amplitude with $m< n$) is left.
Using~\ref{eq:npoint} and~\ref{eq:coeN}  we    arrive at  the desired kinematic factor
$\frac{s_{\widehat{n-2},n-1}}{s_{\widehat{n-2},n}}$.

\paragraph{An example:}
An example is warranted here.
In ${Y_{\sigma_0}}_6^{(2)}(1,2,3,4,5,6)$,  where $\sigma_0$ is taken to be $\{4,5,6,8,7,9\}$,  we take $5$ and $6$ to be the permuting legs.
According to~\ref{eq:BoxCondition1}, such a Yangian invariant can have a ``box'' connecting  to legs $5$ and $6$ directly as shown in Fig.~\ref{fig:6pointNMHV},
\begin{figure}[ht!]
  \centering
 \includegraphics[width=0.50\textwidth]{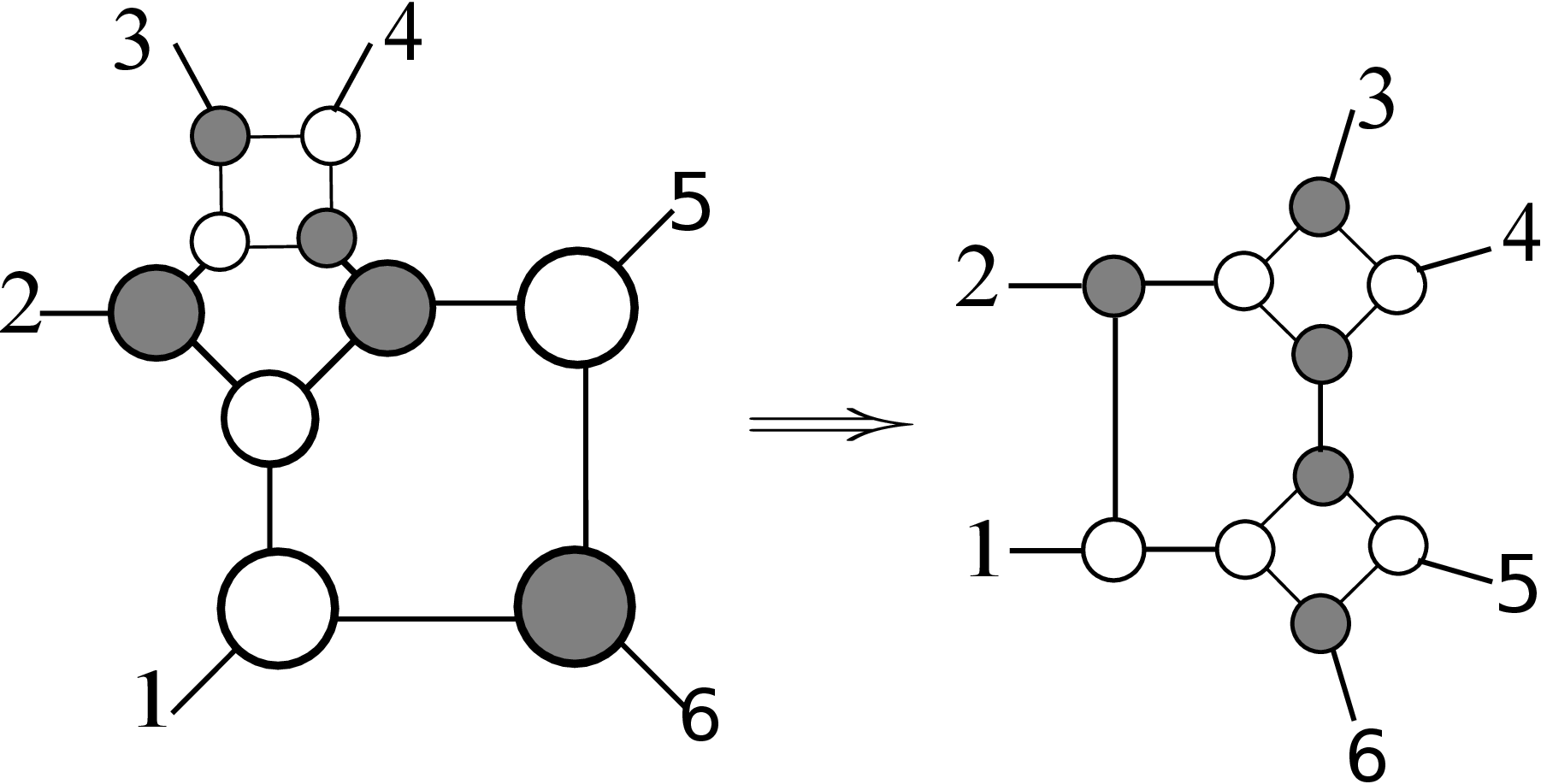}
 \caption{Transforming  to an on-shell diagram with box connecting legs $5,6$ directly.}
 \label{fig:6pointNMHV}
\end{figure}
According to~\ref{eq:Ynk} we  get
\begin{eqnarray}
\label{eq:YB63}
{Y_{\sigma_0}}_6^{(3)} (1,2,3,4,6,5)
&=&\frac{s_{\widehat{4}5}}{s_{\widehat{4}6}}{Y_{\sigma_0}}_6^{(3)}
(1,2,3,4,5,6)~.
\end{eqnarray}
The kinematic factor can hence be read off directly
\begin{equation}
\frac{s_{\widehat{4}5}}{s_{\widehat{4}6}}=-{\brac{\hat{4}}{5}\brac{1}{6}\over \brac{\hat 4}{6}\brac{1}{5}}~
\end{equation}
with $\lambda_{\hat 4}$ being solved by~\ref{eqn:onshellCD}
$$\lambda_{\hat 4}=(p_2+p_3+p_4)|\tilde\lambda_2].$$

According to the arguments in~\cite{NimaGrass} each on-shell diagram or Yangian invariant is associated with a differential form
\begin{eqnarray}
\label{eq:IntForm}
d\Omega \delta(C\cdot\widetilde\eta)\delta(C\cdot\widetilde\lambda)\delta(\lambda\cdot C^{\bot}),
\end{eqnarray}
where $d\Omega$ is the Grassmannian integration measure and $C^{\bot}$  is orthogonal to $C$.
The $C$ can be taken as the matrix associated with the linear constraints
$\delta(C\cdot\widetilde\eta)$, $\delta(C\cdot\widetilde\lambda)$,
$\delta(\lambda\cdot C)$
for the  external spinors $\lambda$,
$\widetilde\lambda$, $\widetilde\eta$.
The Grassmannian cell for a MHV amplitude is always
$$C=\left(
\begin{array}{cccccccc}
  \lambda^1_1& \lambda^1_2  &\cdots&\lambda^1_i&\lambda^1_{i+1}&\cdots&\lambda^1_{n-1}&\lambda^1_n   \\
 \lambda^2_1& \lambda^2_2  &\cdots&\lambda^2_i&\lambda^2_{i+1}&\cdots&\lambda^2_{n-1}&\lambda^2_n
\end{array}
\right).
$$

A permutation of two external lines does not change any of the linear constraints in~\ref{eq:IntForm}.
The Grassmannian cell is thus not affected by  permutations when we fix the vector of external spinors.
In fact this  rule  can be generalized to any on-shell diagrams in tree-level amplitudes.  Any permutation of two external legs attached to  a box does not affect the Grassmannian cell for a given vector of external spinors.

It should be emphasized that, in the sense of positroid stratification, the previous Grassmannian $C$ matrix and the matrix $C^{\prime}$ obtained after  a permutation  is not exactly the same. So this kind of transformations is distinct from  square moves and merges with the latter two  leave the $C$ matrix exactly the same as before. 
However, if we look at  the linear constraints $\delta(C\cdot\widetilde\lambda)$, we can see that  (we will prove it later)
$$
C\cdot\widetilde\lambda=C^{\prime}\cdot\widetilde\lambda^{\prime},
$$
implying  these two matrices capture  the same  set of linear constraints.
Since $\widetilde\lambda^{\prime}$ can be simply related to 
$\widetilde\lambda$ by a  matrix transformation, 
if we fix the order of external spinors, setting 
$\widetilde\lambda^{\prime}\to\widetilde\lambda$, 
$C^{\prime}$ has a natural map to $C$. 
At this level we take the two $C$ matrices  to be equivalent.

We can proceed to  evaluate the final results of these two diagrams. According to~\cite{NimaGrass}, the final result of the tree level diagram is 
\begin{equation}
f_\sigma^{(k)}\!=\!\oint\limits_{C\subset\Gamma_\sigma}\!\!\!\frac{d^{k \times n} C}{\mathrm{vol}(GL(k))}\;\frac{\delta^{k\times4}\big(C\!\cdot\!\widetilde{\eta}\big)}{(1\cdots k)\cdots(n\cdots k 1)}\delta^{k\times2}\big(C\!\cdot\!\widetilde{\lambda}\big)\delta^{2\times(n-k)}\big(\lambda\!\cdot\!C^\perp\!\big).
\label{eqn:general_on_shell}
\end{equation}
Since the $C \cdot\widetilde\lambda$ and $\lambda\cdot C^\perp$ 
are the same in these  two  cases, the only difference between the results of these two diagrams comes from the minors 
in (\ref{eqn:general_on_shell}). 
The original diagram can result from the minors of consecutive chains of columns, which is the property from positroid stratification. However, the permuted results can have some minors of the \textit{inconsecutive columns}. So, in this sense, we can classify the diagrams with box permutations of a given  kind,  which, in turn,  can be used to evaluate non-planar diagrams. Examples will be shown in Section~\ref{sec:NMHV}. 

This is not obvious that $C^{\prime}$ can be transformed to $C$ by rearranging  the columns. However it is not hard to prove.
Without loss of generality we  take the permuted external legs $n-1$ and $n$.
According to a BCFW decomposition to a box (Sometimes we cannot reduce to a box by the canonical BCFW decomposition introduced in~\cite{NimaGrass}, but we can always obtain a box by remove BCFW bridges  in a certain way.),  
the $C$ matrix of an on-shell diagrams can be generalized by performing  BCFW  operations
on the  $C_0$ matrix  corresponding to $\bar\sigma$.
The  rows of $C_0$ are denoted by the $i_w$'s  which satisfy
$\bar\sigma(i_w)=i_w+n$, $\bar\sigma(n-1)-n$ and $\bar\sigma(n)-n$.
In the tree level on-shell diagrams  the $\delta$-functions are just enough to fix the parameters $\alpha_I$'s.
Hence the total number of BCFW bridges acting on a box is $2n-8$;  we  obtain
\begin{eqnarray}
\label{eq:CBCFWBOX}
C=C_0 \mathcal{B}(i_5,j_5;\alpha_5)\cdot \mathcal{B}(i_6,j_6;\alpha_6)\dots \mathcal{B}(i_{I},j_{I};\alpha_{I})\dots \mathcal{B}(i_{2n-4},j_{2n-4};\alpha_{2n-4}),
\end{eqnarray}
where
\begin{equation}
\mathcal{B}(i_I,j_I;\alpha_{I})=
\bordermatrix{
& & &  &  & & j_I &   & \cr
 &1 & 0 & \cdots  & 0 & \cdots  & 0 & \cdots  & 0\cr
 &0 & 1 & \cdots  & 0 & \cdots  & 0 & \cdots  & 0 \cr
 &\vdots  & \vdots  & \ddots & 0 & \cdots  & 0 & \cdots  & 0 \cr
i_I &0 & 0 & 0 & 1 & \cdots  & \alpha_I & \cdots  & 0 \cr
 &\vdots  & \vdots  & \vdots  & \vdots  & \ddots & 0 & \cdots  & 0 \cr
 &0 & 0 & 0 & 0 & 0 & 1 & \cdots  & 0 \cr
 &\vdots  & \vdots  & \vdots  & \vdots  & \vdots  & \vdots  & \ddots & 0 \cr
& 0 & 0 & 0 & 0 & 0 & 0 & \cdots  & 1
},\nb
\end{equation}
and
\begin{equation}
C_0=\bordermatrix{
  &   &  & \bar\sigma(n-1)-n &   & \bar\sigma(n)- n & & n-1 & n \cr
& \cdots  & \cdots  & 0 & \cdots  & 0 & \cdots  & 0 & 0 \cr
\bar\sigma(n-1)- n & 0 & 0 & \lambda^1_{\bar\sigma(n-1)-n} & \cdots  & \lambda^1_{\bar\sigma(n)-n} & \cdots & \lambda^1_{n-1}  & \lambda^1_{n} \cr
  & \vdots  & \vdots  & \vdots  & \ddots  & \vdots & \ddots  & \vdots  & \vdots  \cr
  \bar\sigma(n)- n & 0 & 0 & \lambda^2_{\bar\sigma(n-1)-n} & \cdots & \lambda^2_{\bar\sigma(n)-n} & \cdots & \lambda^2_{n-1} & \lambda^2_{n} \cr
  & \vdots  & \vdots  & 0 & \ddots  & 0 & \ddots & 0 & 0
}, \nb
\end{equation}
The elements $C_0[i_w,i_w]=1$, where $C_0[i_w,i_w]$ are the element in $i_w$ row and $i_w$ columns of $C_0$.  Other elements in $C_0$ are zero. It is obvious that the permutation on the box will not affect $C_0$ and all the delta functions  associate with $C_0$ if we fix the order of $\lambda$ .
Furthermore all the parameters $\alpha_{I}$ are fixed  by the delta functions in the box. And none of vertices and internal lines outside the box in the on-shell diagram can be affected  by  permutations of the external legs.
Hence all BCFW bridges are invariant under the external leg permutations, which in turn implies that the Grassmannian cell $C$ is invariant under the permutations.
We can  make a stronger generalization: any permutations of two legs attached to  a bipartite box  will not affect the Grassmannian cell.
The proof is completely analogous; an example will be presented in Section~\ref{sec:NMHV}.

\subsection{A permutation relation for NMHV
amplitudes--a twin-box connected by a BCFW bridge}
\label{sec:BridgeBiBox}

In general, however, not any two consecutive legs can enjoy a  direct connection to a box.
From NMHV amplitudes onward to more general amplitudes,
the next basic object arisen in a permutation of two adjacent legs is a twin-box with the two permuting legs connected by a  BCFW bridge.  Fig.~\ref{fig:BiBox56} shows
the  first NMHV example of a six-point Yangian invariant corresponding  to the permutation $\sigma_1=\{4,5,6,8,7,9\}$
with $5$ and $6$ permuted.
\begin{figure}[ht!]
  \centering
 \includegraphics[width=0.30\textwidth]{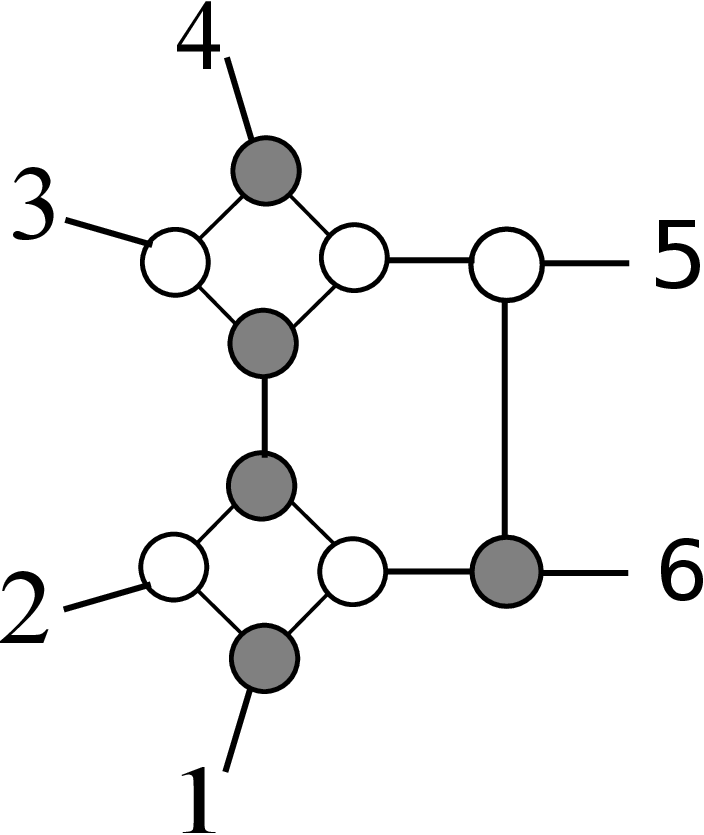}
 \caption{In a typical NMHV diagram, a pair of external legs, marked "5" and "6,"  can be made to connect to a basic twin-box by  a BCFW bridge with judicious applications of the permutation relations of the Yangian Invariants}
 \label{fig:BiBox56}
\end{figure}
According to
\begin{eqnarray}
\label{eq:YBiBox63}
&&Y_{\sigma_1}^{(3)} (o2)\equiv A^{MHV}(o2) {\bar Y_{\sigma_1}}^{(3)} (o2)\nb\\
&=&{A^{MHV}(o2)\over \bracf{2365}^{o2}\bracf{3651}^{o2}\bracf{6512}^{o2}\bracf{1236}^{o2}\bracf{5123}^{o2}}\nb\\
&\times& \delta(\bracf{2365}^{o2}\tilde{\eta}_1+\bracf{3651}^{o2}\tilde{\eta}_2+\bracf{6512}^{o2}\tilde{\eta}_3+\bracf{5123}^{o2}\tilde{\eta}_6+\bracf{1236}^{o2}\tilde{\eta}_5)
\end{eqnarray}
and a similar equation for $Y_{\sigma_1}^{(3)} (o1)$~%
\footnote{This is the form of a Yangian invariant in the momentum twistor space, and we will mainly discuss amplitudes in the momentum twistor space in this paper.
A brief introduction of the momentum twistor space is included
in Appendix~\ref{app:twistor}.},
the permutation relation is easy  to obtain,
\begin{eqnarray}\label{eq:YBiBox63}
\bar Y_{\sigma_1}^{(3)} (o2)&=&{\bracf{2356}^{o1}\bracf{3561}^{o1}\bracf{5612}^{o1}\bracf{6123}^{o1}\bracf{1235}^{o1}\over \bracf{2365}^{o2}\bracf{3651}^{o2}\bracf{6512}^{o2}\bracf{5123}^{o2}\bracf{1236}^{o2}}\nb\\
&\times&{1\over (\bracf{1235}^{o2})^4} \int d^4\bar{\tilde{\eta}}_6 \delta(\bracf{1235}^{o2}\bar{\tilde{\eta}}_6-\sum_{i=1}^{6}c_i\tilde{\eta}_i) \bar Y_{\sigma_1}^{(3)} (o1),
\end{eqnarray}
where $o2=(1,2,3,4,6,5), o1=(1,2,3,4,5,6)$ and $c_1=\bracf{2365}^{o2}-\bracf{2356}^{o1}, c_2=\bracf{3651}^{o2}-\bracf{3561}^{o1}, c_3=\bracf{6512}^{o2}-\bracf{5612}^{o1}, c_4=0, c_5=\bracf{1236}^{o2}-\bracf{6123}^{o1}, c_6=\bracf{5123}^{o2}$.

In fact for a Yangian Invariant in NMHV amplitudes, a bipartite diagram is composed of BCFW-bridged box glued  with k-preserving  inverse soft  factor $Y_n^{(3)}=Y_6^{(3)}\underbrace{\odot Y_3^{(1)}\odot\ldots\odot Y_3^{(1)}}_{n-6}$. This is because we can choose a BCFW bridge for at least one pair of consecutive legs such that $Y_{n_L}^{k_L}$ and $Y_{n_R}^{k_R}$ with
$n_L>3, n_R>3$ and $k_L=k_R=2$.
And according to the analysis in Section~\ref{sec:PRBox}  the general structure and the permutation relation of an  on-shell diagram in NMHV is as shown in Fig.~\ref{fig:NMHVBiBox}.
\begin{figure}[ht!]
  \centering
 \includegraphics[width=0.50\textwidth]{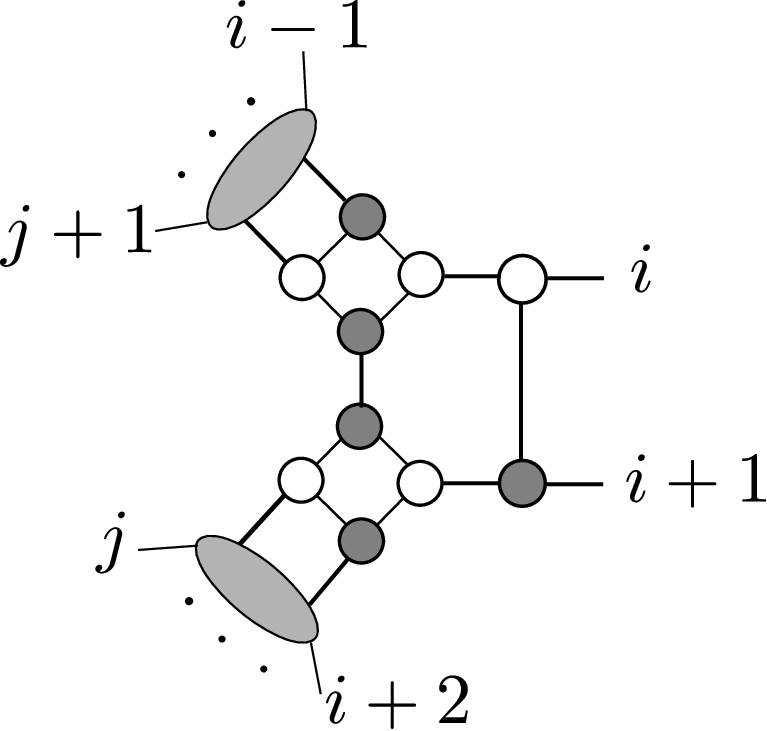}
 \caption{Bridged bi-box with permutation legs $i,i+1$.}\label{fig:NMHVBiBox}
\end{figure}

For general amplitudes beyond NMHV this permutation relation for the BCFW-bridged twin-box can be easily  shown to  exist for a class of Yangian invariants.
We shall establish a convenient criterion for them.
To this end, we define a revised BCFW-Bridge decomposition~\cite{NimaGrass}.
\paragraph{BCFW-Bridge decomposition of a bridged twin-box:}
Staring with a given permutation $\sigma$, we mark two consecutive legs $i$ and $i+1$ and other two legs $\sigma(i)~mod~n, \sigma(i+1)~mod~n$. Now one chooses another $n-6$ legs, other than the four chosen ones.
The  left 6 legs are  left fixed  if  $\sigma(i)~ \text{mod}~ n\neq i$ and $\sigma(i+1)~ \text{mod}~ n \neq i+1$ and the the box contact condition (\ref{eq:BoxCondition1})  does not hold.
If $\sigma$ for other legs is not a ``dressed'' identity~\footnote{For example, for $n=6, k=3$, a ``dressed" identity is $\{7,8,3, 10,5,6\}$.}
we decompose $\sigma$ as $(j_1 j_2)\circ \sigma'$,
where $1\leqslant j_1<j_2\leqslant n$, $\sigma(j_1)<\sigma(j_2)$,
$j_1\neq i, i+1$, $j_2\neq i, i+1$ and $j_1, j_2$
are separated only by marked legs or legs $i, i+1$,
keeping  the order of  the 6 fixed legs. This procedure is  repeated until $\sigma$ becomes  the identity for all the mobile  legs and  the  resultant permutation is denoted by $\bar\sigma$.

If  $\bar\sigma$ is  a permutation of $Y_6^3$ then the bipartite diagram is of the form  shown in Fig.~\ref{fig:GAmpBiBox}.
When we permute the legs $``i"$ and $``i+1"$, the effect of the permutation on the Grassmannian matrix  $C$ will  be partially blocked   by the box.
In fact  the total number of BCFW bridges acting on a bridged twin-box is $2n-12$; and we  obtain
\begin{eqnarray}
\label{eq:CBCFWBOX}
C=C_0 \mathcal{B}(i_8,j_8;\alpha_8)\cdot \mathcal{B}(i_9,j_9;\alpha_9)\dots \mathcal{B}(i_{I},j_{I};\alpha_{I})\dots \mathcal{B}(i_{2n-4},j_{2n-4};\alpha_{2n-4}).
\end{eqnarray}

Similar to a permutation of  the bipartite box all the
parameters $\alpha_{I}$ are fixed by  delta functions of
momentum  conservation in the bridged twin-box.
None of the  vertices  and internal lines outside of the bridged twin-box are  affected  by the  permutations of the external legs.
Hence all  BCFW bridges  are invariant under these permutations.
Nevertheless one row in $C_0$ does change while the other rows of are invariant under a leg permutation on the twin-box.
Such transformation relations on $C$ is therefore useful for classifying Yangian invariants related by a given  permutation of legs.
\begin{figure}[ht!]
  \centering
 \includegraphics[width=0.50\textwidth]{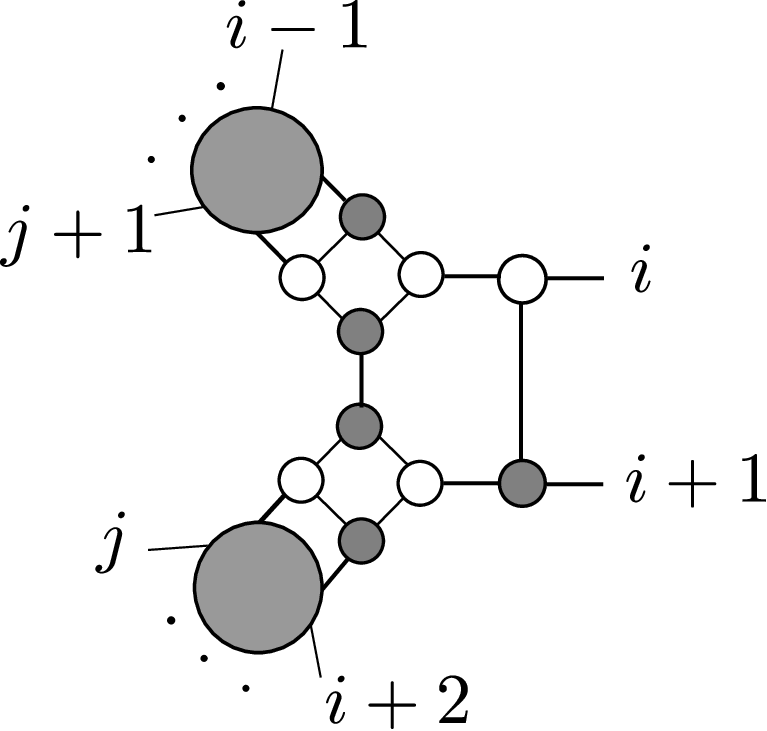}
 \caption{Bridged bi-box with permutation legs $i,i+1$.}\label{fig:GAmpBiBox}
\end{figure}

\section{Unitarity cuts and generalized unitarity cuts}
\label{sec:unitaritycuts}

In $\mathcal N=4$ supersymmetry Yang-Mills theory ``single cut''  is an  efficient way of  constructing   all loop integrands for planar loop amplitudes.
However in the nonplanar  case  the  resultant diagram after a  single cut  is often not a  well-defined Feynman diagram.
One can also view this problem as a difficulty to
endow the loop momentum with a canonical definition because the nonplanar leg(s) can fall between any two planar legs inside a loop. In fact all such possibilities should be taken into account.

The general loop amplitudes after the unitarity cuts and generalized unitarity cuts can be regarded as tree level amplitudes being glued together. Hence the unitarity and generalized unitarity cut loop amplitudes are well-defined and can be taken as the foundation to construct the  integral of the amplitudes.  From this point of view, the major difference for planar and nonplanar diagrams under unitarity cuts is that  all the gluing lines in each tree-level amplitude are adjacent for planar diagrams while in nonplanar diagrams at least a pair of gluing lines is nonadjacent. Furthermore, for the unitarity cuts, together with the on-shell diagrams for the tree-level amplitudes, it is also possible to construct the integrand of the general loop amplitudes systemically, which we will discuss in Section \ref{Sec:MHV}.

\subsection{Unitarity cut}\label{subset:UCut}

Given a nonplanar diagram one should consider all possible diagrams resulted from the nonplanar leg(s)  taking  all probable positions when traversing around the loop.
The simplest example is the four-point one-loop with one nonplanar leg--which we shall call the ``(3+1)'' case for short in the rest of the article--%
as shown in Fig.~\ref{fig:3plus1-all}.
 \begin{figure}[ht!]
  \centering
 \includegraphics[width=0.55\textwidth]{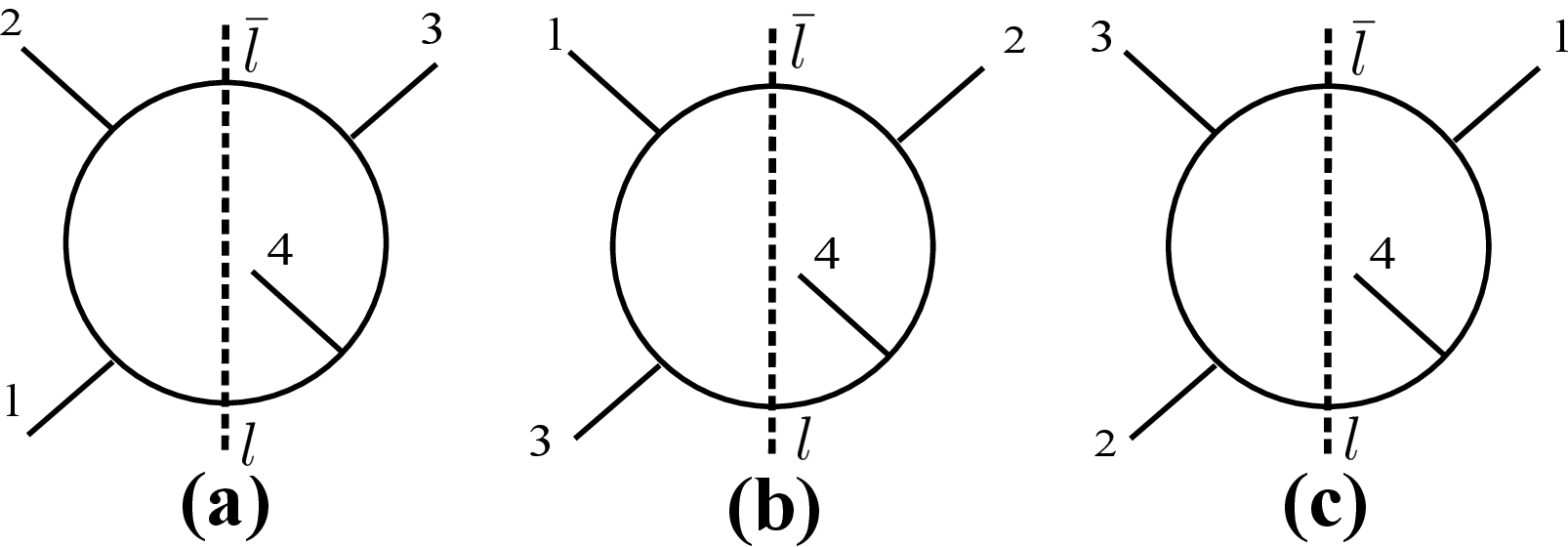}
 \caption{Three possible positions of the nonplanar leg and a possible unitarity cut in each case.}
 \label{fig:3plus1-all}
 \end{figure}
The nonplanar leg can take up three different positions; and there are two possible unitarity cuts.

The  ambiguity in defining the loop momentum is resolved
as follows: if we start with the external line marked ``${\bf 4}$'' we can call the momentum on the first cut loop line $l$, and $\bar{l}$  the loop momentum on the other cut line.
In the clockwise order for the color-ordered amplitudes one easily checks that in each resultant tree diagram the momentum is well defined for  each of the  external legs.

We can keep  track of the order of unitarity cuts as well: under each  cut the diagram will be divided into a diagram with one fewer loops in addition to a tree-level diagram.
A typical higher loop case is shown in Fig.~\ref{fig:loop-momentum}.
\begin{figure}[ht!]
  \centering
 \includegraphics[width=0.70\textwidth]{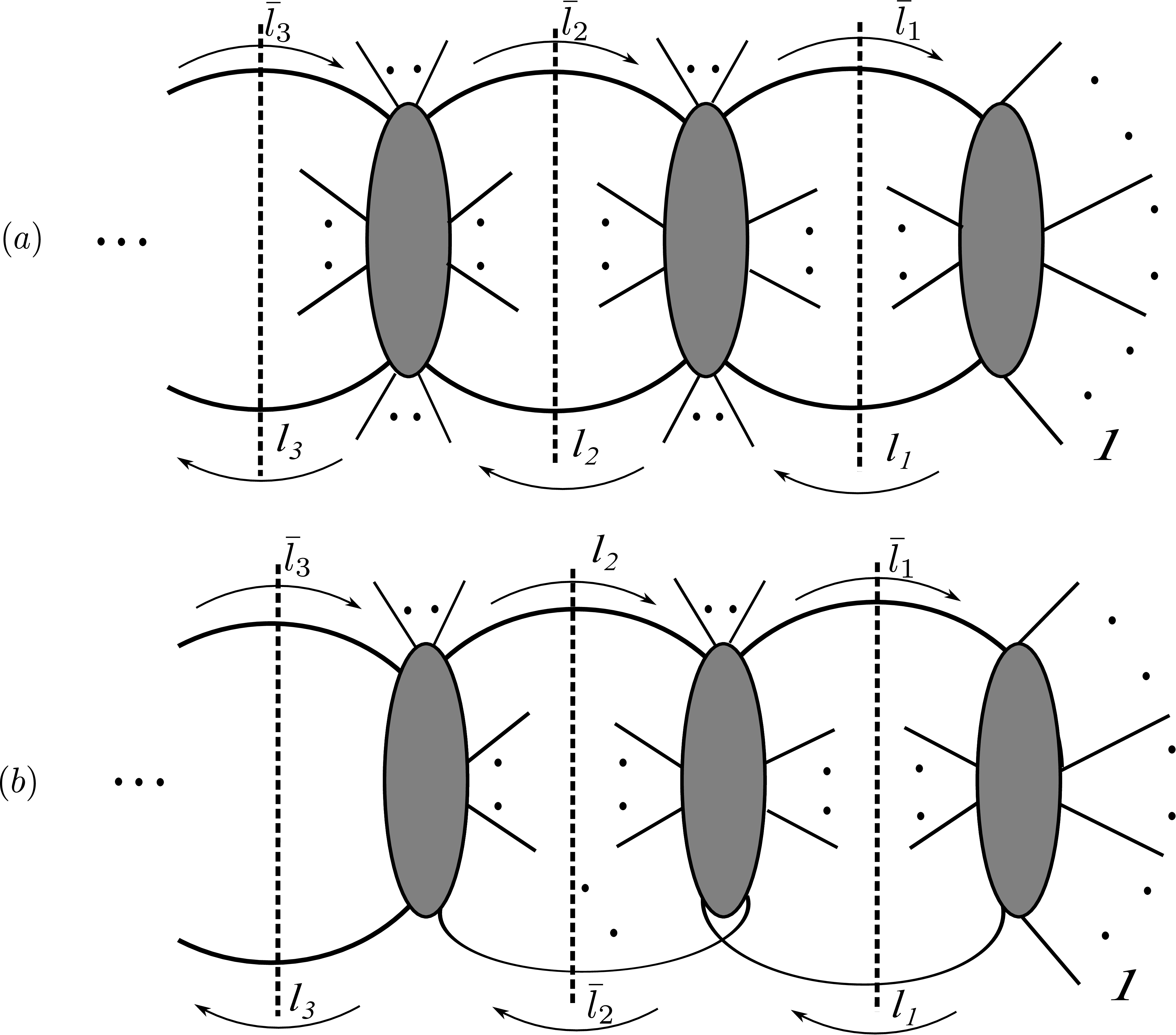}
 \caption{Unitarity cuts for a higher loop nonplanar diagram.}
 \label{fig:loop-momentum}
 \end{figure}

The ambiguity in the definitions of the loop momenta
is resolved by a series of unitarity cuts.
In fact we can define each loop momentum clockwise from a reference external line.
For the typical example in $(a)$ of
Fig.~\ref{fig:loop-momentum}, if we start with the external line marked ``${\bf 1}$''--in the clockwise order for the color-ordered amplitudes--we can call the momentum on the first cut loop-line $l_1$ and $\bar{l}_{1}$ the loop momentum on the other cut line.
After setting $l_1$ as the reference line in the first loop the momentum on the other cut line becomes
$\bar{l}_{1}=l_1-P_R$  where $P_R$ being the sum of all  external momenta  between these two cut lines.
The momenta on the second loop can thus be fixed to be $l_2$ and $\bar{l}_{2}=l_2-P'_R$ with $P'_R$ denoting the sum of all the external momenta  to the right of $l_2$, and so forth.

Obvious in this construction,  topological  information of the
non-planarity is preserved: one unitarity cut can only fix two components of the 4-momentum  integrals leaving the other two integrations unconstrained.
Furthermore, the integrand under the unitarity cut is also well-defined.
This means that the integrand  contains enough information for the characterization of  the loop topology of  nonplanar diagrams.
And the loop topology and geometric  properties of the  underlying Grassmannian arisen in nonplanar amplitudes will be manifest once we construct nonplanar amplitudes
in on-shell bipartite diagrams~\cite{NimaGrass}.

The  permutation relations
of  bipartite on-shell diagrams--each of them corresponding to a Yangian with its gauge invariance--are instrumental
in  constructing the whole amplitude from unitarity cut or generalized unitarity cut diagrams.
To construct the bipartite  on-shell diagrams after a  unitarity cut we  need to convert each resultant  tree amplitudes in Fig.~\ref{fig:loop-momentum}
into  the corresponding bipartite diagrams.
The loop lines connecting the tree amplitudes now denote the same integration as the internal lines in tree-level bipartite on-shell diagrams.  Since the construction of bipartite on-shell diagram for each tree level diagram is well established
in~\cite{NimaGrass}, we only need to verify the gluing lines in bipartite diagram are equivalent to a unitarity cut of loop amplitudes.

In the language of  bipartite on-shell diagram, the gluing line represents an extra integral:
\begin{equation}
\int \frac{d^2\lambda_{l_1} d^2\tilde\lambda_{l_1}}{   vol(GL(1))} d^4\tilde{\eta_{l_1}}  \int \frac{d^2\lambda_{l_2} d^2\tilde\lambda_{l_2}}{  vol(GL(1))} d^4\tilde{\eta_{l_2}} A_L^S ~A_R^S
\end{equation}
where
$A_L^S=A_L \delta^{2\times 2}(l_1+P_L-l_2)\delta^{2\times 4}(\lambda_L\cdot\tilde{\eta}_L)$,
$A_R^S=A_R\delta^{2\times 2}(-l_1+P_R+l_2)\delta^{2\times 4}(\lambda_R\cdot\tilde{\eta}_R)$.
This  can be further simplified,
\begin{equation}
\int \langle \lambda d\lambda\rangle [ \tilde\lambda d\tilde\lambda ] d^4\tilde{\eta_{l_1}}   \frac{P_L^2}{\langle \lambda|P_L|\tilde\lambda]^2} d^4\tilde{\eta_{l_2}} A_L ~A_R \delta^{2\times 2} (P_R+P_L) \delta^{2\times 4}(\lambda_L\cdot\tilde{\eta}_L) \delta^{2\times 4}(\lambda_R\cdot \tilde{\eta}_R)
\end{equation}
which is exactly the expression of a loop level amplitude after a unitarity cut.

Let us  turn, again,  to our lovely ``3+1''  example 
(Fig.~\ref{fig:3plus1-all}),
the bipartite on-shell diagrams correspond to the tree amplitudes resulted from a unitary cut are shown in
Fig.~\ref{fig:stu-cuts}.
\begin{figure}[ht!]
  \centering
 \includegraphics[width=0.9\textwidth]{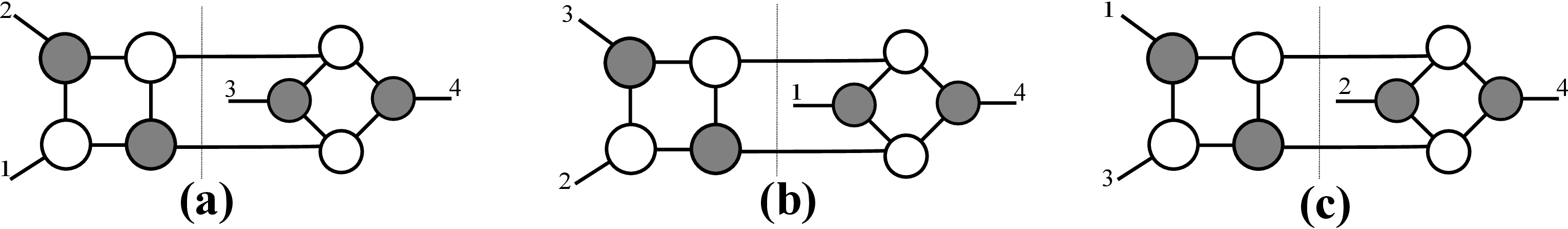}
 \caption{The  bipartite  on-shell diagrams of an s-channel cut~(a), and a t-channel cut~(b), and a u-channel cut~(c).  {\bf Note}:  In our convention a horizontal square denotes the planar tree amplitude while a rhombus (at  45 degrees) denotes a  nonplanar amplitude with  the plane of rhombus being perpendicular to the plane of the paper where the points marked ``3'' and ``4''  are at equal distance from the vertical edges of the (planar) square.}
 \label{fig:stu-cuts}
\end{figure}
After the cuts, only two four-point tree amplitudes are left. Each tree amplitude is a box. Now we can add two lines to connect these two boxes to represent the cut amplitude.  There are three different cuts--the s-channel cut, the t-channel cut, and the u-channel cut. Each of them can be  represented in on-shell diagrams in the ways shown in Fig.~\ref{fig:stu-cuts}.

We now present our strategy for  constructing the full scattering amplitudes in the bipartite on-shell language
for the corresponding  the nonplanar Feynman diagrams.
In this work we only focus on the one-loop diagrams.
We would like to stress that  our  strategy can be  extended to the   higher loops cases, with generalized unitarity cuts, in a straightforward way.
Detailed descriptions,  together with carefully worked
out examples, of the one-loop amplitudes will be presented in Section~\ref{Sec:MHV}.

\begin{itemize}
\item Perform all possible unitarity cuts on  a given nonplanar Feynman diagram. \\
We convert the resultant diagrams from each possible (series of) unitarity cuts into on-shell bipartite diagrams.
Each bipartite diagram corresponds to a Yangian invariant.
\item Remove all unphysical poles in loops, the structures  of which depend on the loop momenta. \\
They occur because Yangian invariants  in general contain unphysical poles. However the unphysical poles will cancel each other upon summing over all Yangian invariants
of a given amplitude--only physical poles remain.
Furthermore the unphysical poles in loops are not allowed in the total amplitudes.  We need to ensure that no unphysical poles of the loop momenta appear in the final expressions.
\item Sum over all the inequivalent terms from each series of unitarity cuts. \\
After removing the unitarity cut conditions, we will get an integral with respect to all the loop momenta.
If there appears the same integral when reconstructing from a different unitarity cut then it suffices to count it once.

\end{itemize}

\subsection{Generalized unitarity cuts}
Generalized unitarity  cuts~\cite{Bern:2011qt} can also be used to construct
 the full loop level amplitudes.
For $\mathcal N=4$ SYM after a quadruple cuts  on each loop only the leading singularity of the loop level amplitudes remains.
All  loop momenta are fixed by the cut constraints.
Absent is the possibility  of having a  rational function in  loop momenta.
However, from the bipartite on-shell diagrams,  lots of geometric information of the  Grassmannian can be read off from the leading singularity, as shown in
Fig.~\ref{fig:LeadSingularity}.
\begin{figure}[ht!]
  \centering
 \includegraphics[width=0.65\textwidth]{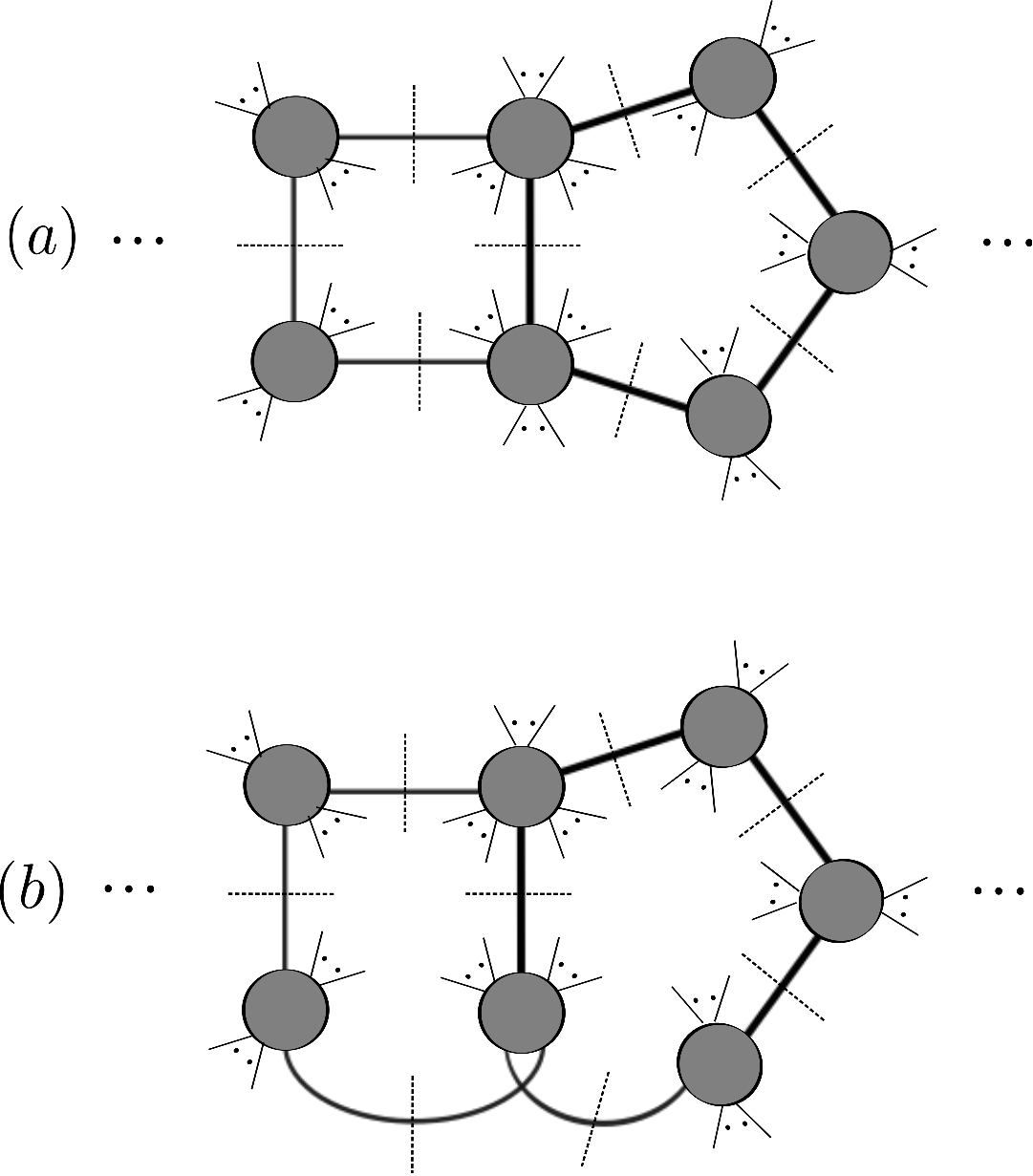}
 \caption{Leading singularity in general loop amplitudes under generalized unitarity cuts.}
 \label{fig:LeadSingularity}
 \end{figure}

In a one loop nonplanar diagram, Fig.~\ref{fig:LeadSingularityOneL},
the leading singularity is of form
\begin{equation}
\int\prod_{i=1}^4 \frac{d^2\lambda_{i} d^2\tilde\lambda_{i}}{   vol(GL(1))} d^4\tilde\eta_{i}\mathcal{A}_1^S(1\cdots 2\cdots)~\mathcal{A}_2^S(2\cdots 3\cdots)~\mathcal{A}_3^S(3\cdots 4\cdots)~\mathcal{A}_4^S(4\cdots 1\cdots).
\end{equation}
Similar to the case  of  planar diagrams the leading singularities of nonplanar diagrams can also be identified as residues of the Grassmannian integral, a specific example of which will be given  in Section~\ref{sec:NMHV}.
\begin{figure}[ht!]
  \centering
 \includegraphics[width=0.65\textwidth]{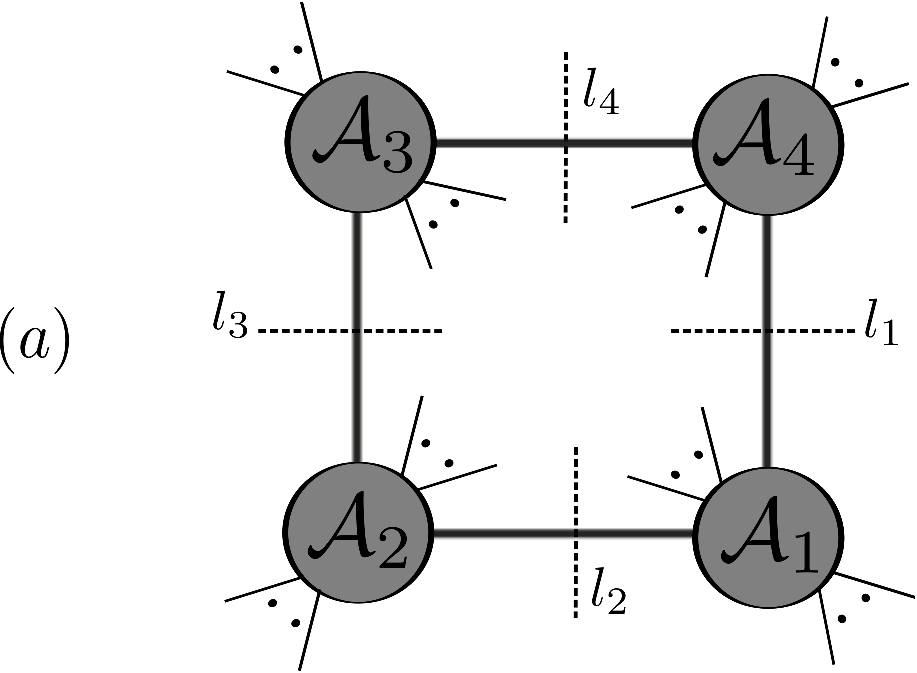}
 \caption{Leading singularity in one loop amplitudes under generalized unitarity cuts.}
 \label{fig:LeadSingularityOneL}
 \end{figure}

In the case of $\mathcal N=4$ SYM  generalized  unitarity cuts (quadruple cuts) can fully determine the full amplitudes.
Compared  to a unitarity cut  generalized unitarity cuts are  more convenient  to  the full amplitudes'  reconstruction because all the poles in loop momenta are automatically physical upon such cuts.
Furthermore different quadruple  cuts lead to   different scalar integrals: we do not need to consider equivalent integrals as we do with unitarity cuts.
The general procedures of reconstructing the full amplitudes by double cuts are as follows.
\begin{itemize}
\item Perform all possible quadruple  cuts on the nonplanar Feynman diagrams with loops. For each possible series of quadruple  cuts we convert the resultant tree-level amplitudes to the bipartite diagrams.
We then glue the cut loop lines according to the Feynman diagram. Similar to  the planar case  these  reproduce  the leading singularities of  the  nonplanar amplitudes.
\item Transform the nonplanar leading singularities into planar ones by the permutation relations of the Yangian invariants.
\item Multiply  the leading singularities by a standard integral.
\item Sum over all contributions from each series of quadruple  cuts.
\item To elucidate  the geometric properties of the Grassmannian we group terms according to their underlying Grassmannian geometry.
\end{itemize}

\section{MHV Loop Amplitudes}
\label{Sec:MHV}
In $U(N)$ Yang-Mills theory, the one loop amplitudes can be decomposed as~\cite{Dixon:1996wi}
\begin{eqnarray}
&&{A}^{\textrm{1-loop}}_{n}( \{a_i\})
=\sum_{\sigma \in S_n/Z_n}
    N_c\,\Tr\LP T^{a_{\sigma(1)}}\cdots T^{a_{\sigma(n)}}\RP\
     \mathcal{A}_{n;1}(\sigma(1),\ldots,\sigma(n))  \\
&&+\ \sum_{c=2}^{\lfloor{n/2}\rfloor+1}
      \sum_{\sigma \in S_n/S_{n;c}}
    \Tr\LP T^{a_{\sigma(1)}}\cdots T^{a_{\sigma(c-1)}}\RP\
    \Tr\LP T^{a_{\sigma(c)}}\cdots T^{a_{\sigma(n)}}\RP\
  \mathcal{A}_{n;c}(\sigma(1),\ldots,\sigma(n))\, , \nonumber
\end{eqnarray}
where $\mathcal{A}_{n;c}$ are the partial amplitudes,
$Z_n$ and $S_{n;c}$ are the subsets of $S_n$
that leave the corresponding single and double trace structures
invariant, and $\lfloor x \rfloor$ is the greatest integer less than or
equal to $x$.  In this paper the diagrams corresponding to single trace and  double trace partial amplitudes are regarded as planar and non-planar diagrams respectively. We will focus on the partial amplitudes of planar diagram,  $\mathcal{A}_{P}\equiv\mathcal{A}_{n;1}$, as well as nonplanar diagrams, $\mathcal{A}_{NP}\equiv\mathcal{A}_{n;2}$, of  4-point, 5-point interactions for $U(N)$ Yang-Mills gauge theory.
\subsection{MHV planar amplitudes and unitarity cuts}
\label{sec:Planar}

Using single cuts techniques, Arkani-Hamed et al has thoroughly studied planar amplitudes of all loops in momentum twistor space~\cite{NimaAllLoop}.  On the other hand Bern et al introduced unitarity cuts as a way to reconstruct  planar MHV amplitudes in momentum space~\cite{Bern94}, which has been instrumental as well as inspiring to our current project.
At this point, however,  no systematic method of MHV amplitudes  reconstruction  from unitarity cuts in momentum twistor space exists.

In this section, we present a detailed method of constructing  MHV one-loop amplitudes from unitarity cuts.
We build relations of bipartite on-shell diagrams and express them in momentum twistor space.
This method  leads  us naturally to  simple results  without unphysical poles, in addition to the final integrands being the same as those from single cuts~\cite{Cachazo:2004dr}.
This is to be contrasted with the way proposed by Bern et al~\cite{Bern94} in dealing with the box integrals.
Given these advantages it is thus a meaningful exercise to study MHV amplitudes in momentum twistor space together with unitarity cuts.  The steps of reconstruction of MHV one-loop amplitudes from unitary cuts are:
\begin{itemize}
\item[I] {Draw the on-shell diagrams of each amplitude under unitary cuts.}
\item[II] {Add BCFW bridges to remove the unitary cut constraints and directly write down the integrand form in momentum twistor space.}
\item[III]{Convert un-physical poles in the previous form to physical ones. }
\item[IV]{Combine results from different cuts to get the final integrands.}
\end{itemize}

\paragraph{Example: Integrands of five-point one loop amplitudes} ~~

Now we give an example to explicit the above procedure.  First non-trivial example is five-point planar amplitude.

\textbf{Step I}: In five-point situation, there are five different unitarity cuts ${A}_c(i,i+1|i+2,i+3,i+4),i=1,2,3,4,5$.
$A_c(12|345)$, for instance, can be constructed as gluing two tree level amplitudes $A_L(12l\bar l)$ and $A_R(\bar l l 345)$ as shown in Fig.~\ref{fig:loop-momentum} for the general case. The corresponding on-shell diagrams is shown in (a) of Fig.~\ref{fig:5point}.
\begin{figure}[ht!]
  \centering
 \includegraphics[width=0.83\textwidth]{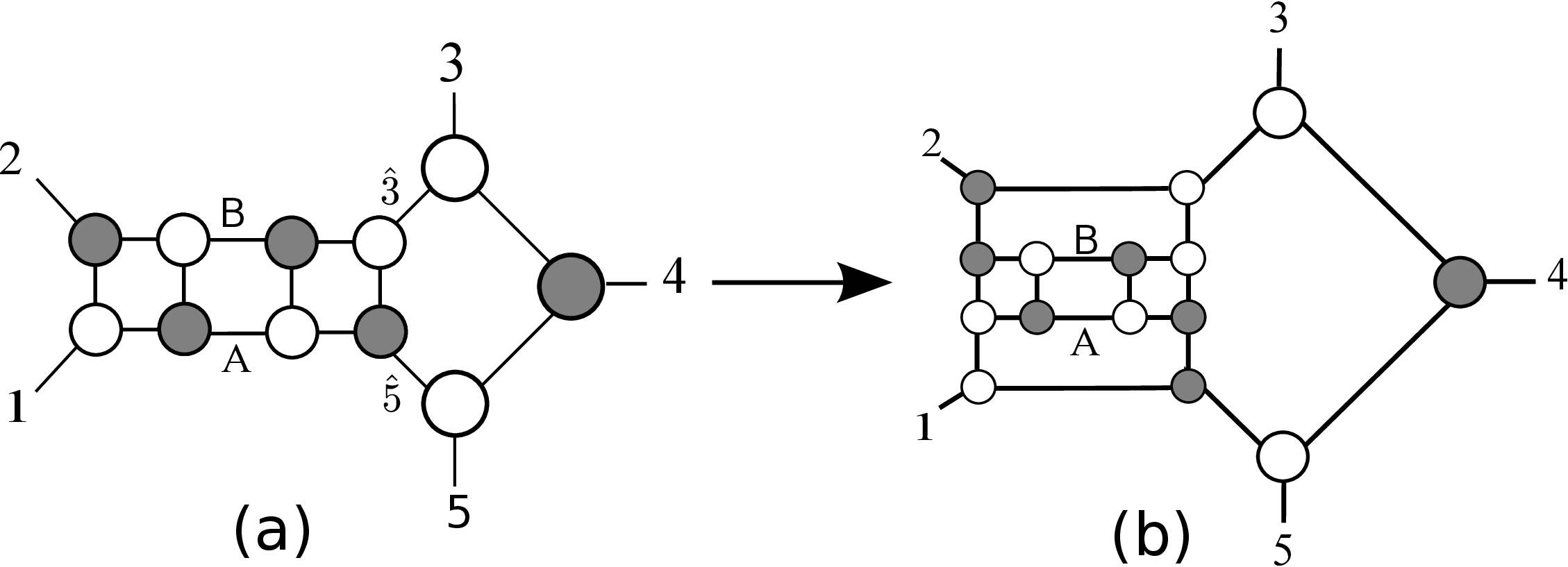}
 \caption{(a) shows the on-shell diagram of $A_c(12|345)$, and (a) transforming to (b) indicates a new way of adding BCFW bridges to remove the cut constraints (the step II). A and B denote two cut lines. }
 \label{fig:5point}
\end{figure}

\textbf{Step II}: Add BCFW bridges to (2 $\hat 3$) and (1 $\hat 5$) to remove the cut constraints (shown in Fig.~\ref{fig:5point}). We can simply write the integrand in momentum twistor space, based on the four-point one-loop situation, as
\begin{eqnarray}
\mathcal{A}_0(1,2|3,4,5)=\frac{-\left\langle1235\right\rangle^2}{\left\langle AB12\right\rangle\left\langle AB23\right\rangle\left\langle AB35\right\rangle\left\langle AB51\right\rangle},
\end{eqnarray}
where A and B denote the points of the cut lines  in momentum twistor space (Appendix \ref{app:twistor}).  
We can simply write down previous equation since the white vertices
on leg 3 and 5 dictates that $\hat\lambda_3 $ and $\hat \lambda_5$ is proportional to $\lambda_3$ and
$\lambda_5$ respectively, and the proportionality constant is irrelevant since the integrand is defined
projectively.

\textbf{Step III}: There are unphysical poles in the denominator of the previous equation, such as $\left\langle AB35\right\rangle$. These poles can be converted to physical ones using unitarity condition, $\left\langle AB23\right\rangle=0$, $\left\langle AB51\right\rangle=0$.
We could build a relation between $\frac{1}{\left\langle AB35\right\rangle}$ and $\frac{1}{\left\langle AB34\right\rangle\left\langle AB45\right\rangle}$ as
\begin{eqnarray}
{\mathcal{A}_1(1,2|3,4,5)}=\frac{\left\langle AB24\right\rangle\left\langle 3512\right\rangle\left\langle 1345\right\rangle+\left\langle AB34\right\rangle\left\langle 5123\right\rangle\left\langle 1245\right\rangle}{\left\langle AB12\right\rangle\left\langle AB23\right\rangle\left\langle AB34\right\rangle\left\langle AB45\right\rangle\left\langle AB51\right\rangle}
\end{eqnarray}

\textbf{Step IV}: Repeat the above three steps on another unitarity cut, we obtain the result of another cut $\mathcal{A}_1(2,3|4,5,1)$ as
\begin{eqnarray}
{\mathcal{A}_1(2,3|4,5,1)}=\frac{\left\langle AB25\right\rangle\left\langle 3451\right\rangle\left\langle 4123\right\rangle+\left\langle AB51\right\rangle\left\langle 3452\right\rangle\left\langle 4123\right\rangle}{\left\langle AB12\right\rangle\left\langle AB23\right\rangle\left\langle AB34\right\rangle\left\langle AB45\right\rangle\left\langle AB51\right\rangle}.
\end{eqnarray}
Now we need to combine these two results from different cuts.
$Z_5$ could be expanded based on $Z_1$, $Z_2$, $ Z_3$, $Z_4$, to get
$\left\langle AB25\right\rangle\left\langle 4123\right\rangle=\left\langle AB23\right\rangle\left\langle 4125\right\rangle+\left\langle AB24\right\rangle\left\langle 5123\right\rangle$
where $\left\langle AB12\right\rangle$ vanishes due to unitarity cut condition.
Obviously the term $\left\langle AB24\right\rangle$ is the same in these two terms, so we only need to count its contribution once. Other terms in these two equations could be directly added together  not affecting the results under both unitarity cuts. Combining other three cuts with the same method we get the final integrand of planar one loop five-point MHV amplitude as
\begin{eqnarray}
{\mathcal{A}_P(1,2,3,4,5)}&=&\frac{-\left\langle AB24\right\rangle\left\langle 2351\right\rangle\left\langle 4351\right\rangle}{\left\langle AB12\right\rangle\left\langle AB23\right\rangle\left\langle AB34\right\rangle\left\langle AB45\right\rangle\left\langle AB51\right\rangle}\\
&-&\frac{\left\langle AB23\right\rangle\left\langle 2451\right\rangle\left\langle 3451\right\rangle+\left\langle AB34\right\rangle\left\langle 2451\right\rangle\left\langle 2351\right\rangle+\left\langle AB51\right\rangle\left\langle 2345\right\rangle\left\langle 2341\right\rangle}{\left\langle AB12\right\rangle\left\langle AB23\right\rangle\left\langle AB34\right\rangle\left\langle AB45\right\rangle\left\langle AB51\right\rangle}\nb,
\end{eqnarray}
which is the same result  as obtained from single cuts in~\cite{NimaAllLoop}.

This simple example serves to illustrate our strategy to compute Yang-Mills integrands using unitarity cuts. This is to be contrasted with the single cuts method proposed by Arkani-Hamed et al~\cite{NimaAllLoop, NimaLocalInt} as well as constructing the scattering amplitudes integral after unitarity-cutting the Feynman diagrams as done by Bern et al at a much earlier attempt \cite{Bern94,Bern95}. We shall proceed with a general discussion of the higher-point results  in the rest of the section.

\paragraph{Property of MHV planar amplitudes under unitarity cuts}~

To study  MHV loop amplitudes under  unitarity cuts
we first tackle MHV tree amplitudes, of which Yangian Invariant is  $Y_n^{(2)}=Y_4^{(2)}\underbrace{\odot Y_3^{(1)}
\odot\ldots\odot Y_3^{(1)}}_{n-4}$.
The relationship between n-point Yangian Invariant (not necessarily MHV) and (n-1)-point Yangian Invariant after stripping off one $\odot Y_3^{(1)}$ is simply\footnote{We have omitted the MHV tree amplitude factor from the full amplitude in momentum space.}
\begin{eqnarray}
{Y}_{m-1}^{\prime(k)}(Z_1,\dots,Z_{m-1})=Y_{m}^{(k)}(Z_1,\dots,Z_{m-1},Z_m)
\end{eqnarray}
in momentum twistor space, which can, in turn, be represented in  on-shell bipartite  diagrams as
\begin{figure}[ht!]
  \centering
 \includegraphics[width=0.53\textwidth]{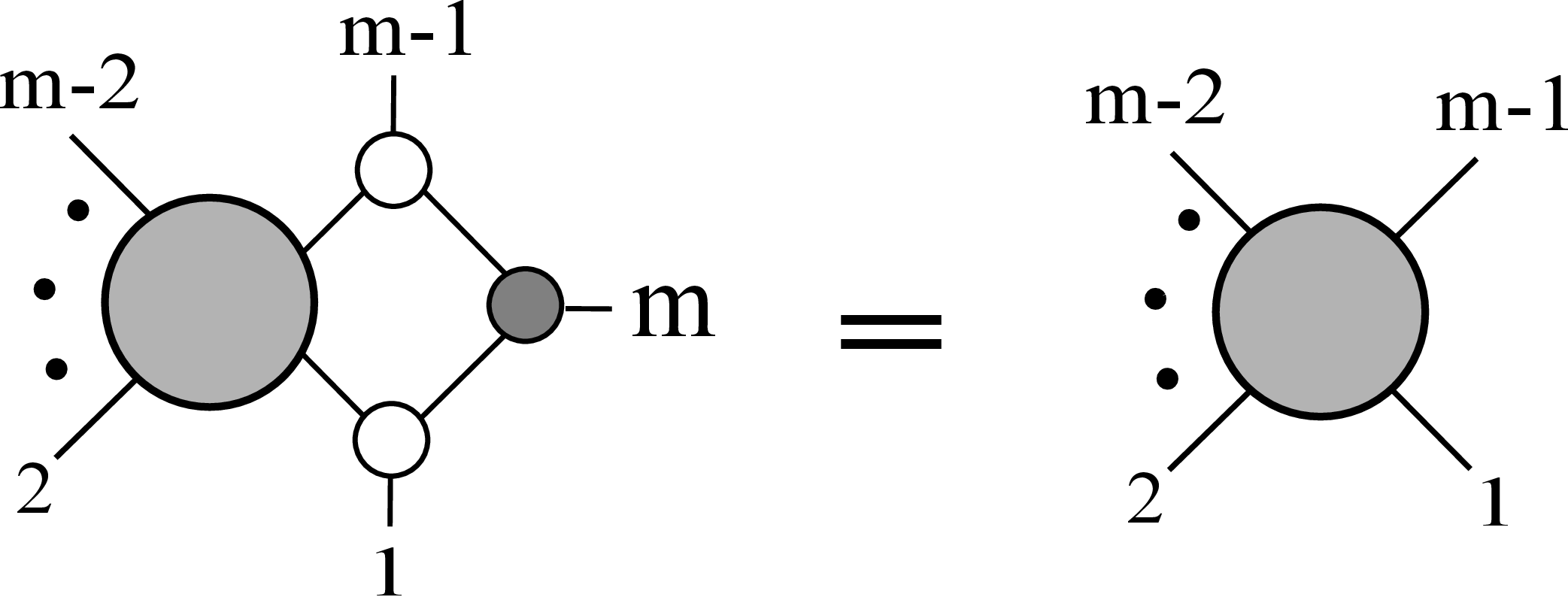}
\end{figure}

The planar part, in general, can be reduced to a very simple form (which we call a basic ``twin-box'') shown in
Fig.~\ref{fig:twinbox}.
\begin{figure}[ht!]
  \centering
 \includegraphics[width=0.83\textwidth]{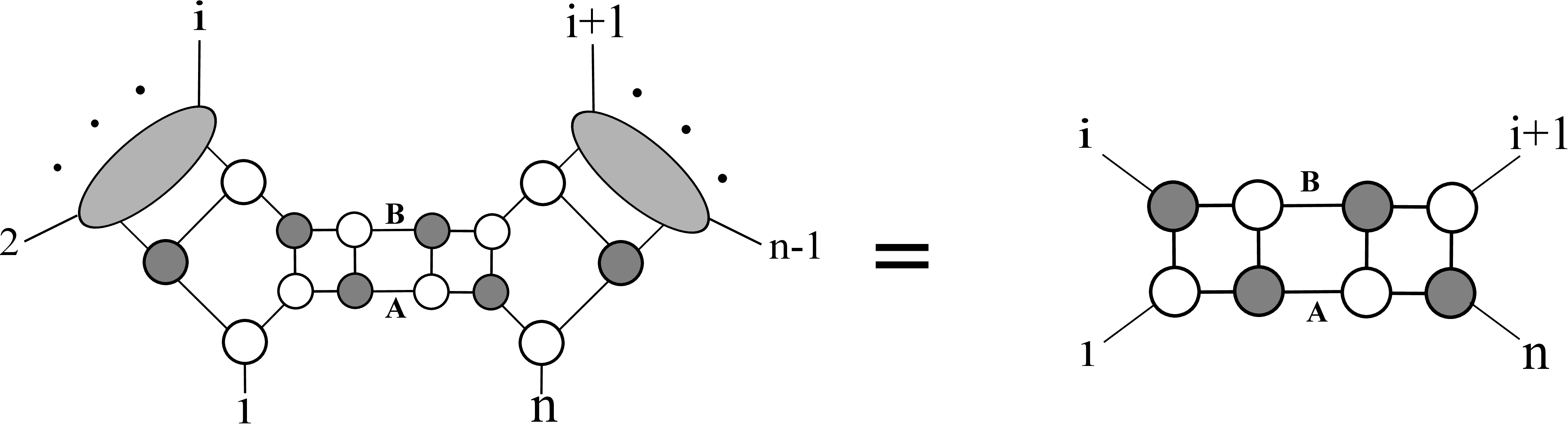}
 \caption{An n-point MHV one-loop planar amplitude after  unitarity cuts  can be converted to   a basic ``twin-box'' with only four external legs.}
 \label{fig:twinbox}
\end{figure}

We can therefore obtain the  MHV amplitude after unitarity cuts as
\begin{eqnarray}
A_c(1,2,\dots,i|i+1,\dots,n)=A_c(1,i|i+1,n),
\end{eqnarray}
where $A_c$ denotes the amplitude under each of the possible unitarity cuts. `$|$' denotes the cut line between these two legs and the other cut line between the first and last legs in the bracket. This relation shows that an n-point MHV amplitude is the same as a four-point amplitude under unitarity cuts in momentum twistor space.

\paragraph{Removing the unitarity cut constraints}~

The standard  way of reconstructing amplitudes from single cuts is by adding BCFW bridges across cut lines on a pair of external legs.
For amplitudes under unitarity cuts we do not necessarily have to add bridges on external legs.
In fact the unitarity-cut amplitudes in MHV case can be reduced to the basic ``twin-box'' in momentum twistor space (Fig.~\ref{fig:twinbox}).
They contain all the essential information of the whole amplitudes after unitarity cuts. This is also apparent  in momentum space.
After four internal integrals, all $\delta$-functions
from the three blocks vanish, reducing the number of external legs by one.
Recursively, the whole amplitude can be reduced  to a basic ``twin-box'' with four new external on-shell momenta.
This means that adding bridges to  the ``twin-box'' recover  the same result from a unitarity cut.  This new way of adding BCFW bridges will greatly simplify our subsequent computations.
In on-shell diagram it amounts to Fig.~\ref{fig:newBCFWbridging}.
\begin{figure}[ht!]
\centering
\includegraphics[width=0.83\textwidth]{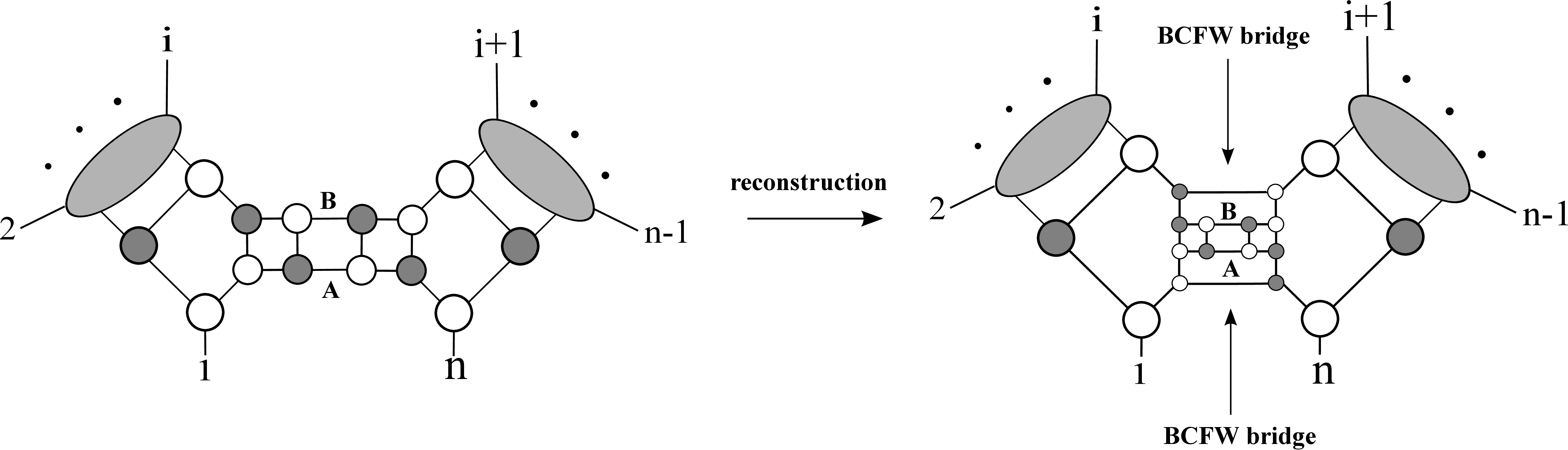}
\caption{A new way of adding BCFW bridges with simpler expressions of integrands.}
\label{fig:newBCFWbridging}
\end{figure}

The final result of this planar part is, thus,  nothing but  a one-loop four-point planar amplitude.
We denote the cut loop  momenta, $l$ and $\bar l$, by the variables $A$ and $B$ in momentum twistor space.~\footnote{We omit the terms related to integral variables $\int\left\langle ABd^2z_A\right\rangle\left\langle ABd^2z_B\right\rangle$ in this paper. Since we deal with the integrand of amplitudes, we leave this as a common factor of in the integrand.}
\begin{eqnarray}
\mathcal{A}_0(1,\dots,i|i+1,\dots,n)=\frac{\left\langle 1i i+1n\right\rangle\left\langle i i+1n1\right\rangle}{\left\langle AB1i\right\rangle\left\langle ABii+1\right\rangle\left\langle ABi+1n\right\rangle\left\langle ABn1\right\rangle},
\end{eqnarray}
where $\mathcal{A}_0$ denotes the amplitude after adding two BCFW bridge to the ``twin-box".

\paragraph{Conversion from unphysical poles to physical ones:}~~

Since the form of planar amplitudes after reconstruction is actually from a four-point $(Z_{i},Z_{i+1},Z_{1},Z_{n})$ (with the subscripts denoting the momenta of the four legs connected to the ``twin-box'' as in Fig.~\ref{fig:twinbox}) amplitude,
some propagators, say $\left\langle A\,B\,i+1\,n\right\rangle$ and $\left\langle A\,B\,i\, n\right\rangle$,  becomes unphysical poles inside  an n-point amplitude ($A$ and $B$ as before denote the loop momenta, $l$ and $\bar{l}$,  in the momentum twistor space.).

We at present  present a way to convert unphysical poles to physical ones using the unitarity cut condition.
We show by an example of an amplitude with color ordering $(1,2,\,\dots, \, n)$ and unitarity cut $\left\langle ABii+1\right\rangle\left\langle ABjj+1\right\rangle$ as shown in Fig.~\ref{fig:Pcut}.
In a particular order, this is just a planar diagram, we can  discuss the integrand.
\begin{figure}[ht!]
  \centering
 \includegraphics[width=0.33\textwidth]{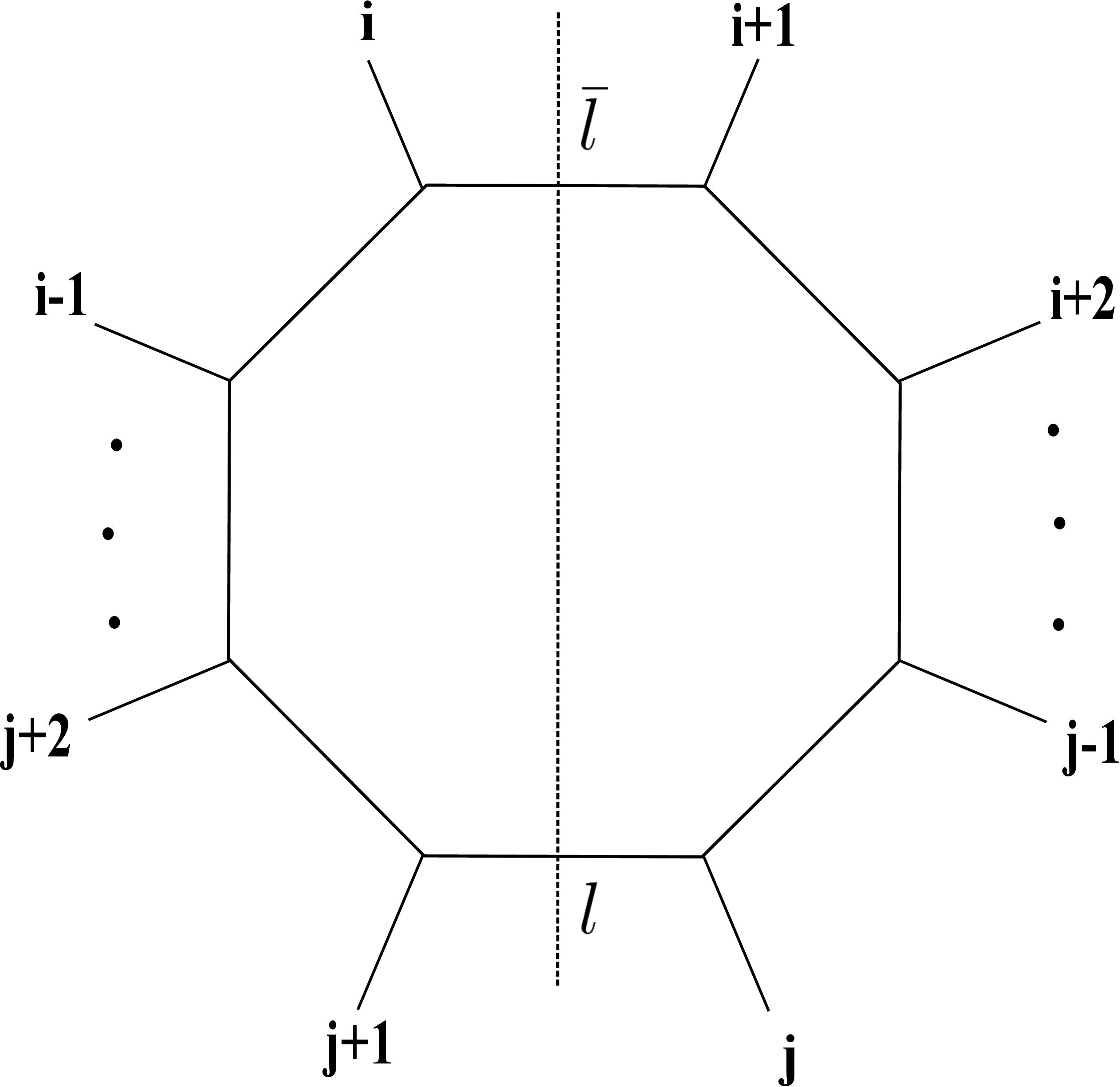}
 \caption{A unitarity cut n-point planar amplitude with color ordering $(1,2,\dots,n)$}
 \label{fig:Pcut}
\end{figure}
The unitarity cut condition is $\left\langle ABii+1\right\rangle=0,\left\langle ABjj+1\right\rangle=0$. And the poles $\left\langle ABi+1j\right\rangle$ and  $\left\langle ABj+1i\right\rangle$ become unphysical.

Imposing the unitarity cut condition,
\begin{eqnarray}\label{UCM}
\frac{\left\langle i i+1jj+1\right\rangle}{\left\langle AB i+1j\right\rangle}=\frac{\left\langle ABi+1i+2\right\rangle\left\langle ij-1jj+1\right\rangle+\left\langle ABii+2\right\rangle\left\langle j-1jj+1i+1\right\rangle}{\left\langle ABi+1i+2\right\rangle\left\langle ABj-1j\right\rangle},
\end{eqnarray}
where all of the poles in the denominator are physical.
To derive this equation, we can parameterize $Z_A, Z_B$ as follows:
\begin{eqnarray}
Z_A=c_1(Z_i +Z_{i+1})+Z_j+Z_{j+1}\nb\\
Z_B=Z_i +Z_{i+1}+c_2(Z_j+Z_{j+1}).\nb
\end{eqnarray}
Then (\ref{UCM}) is equivalent to
\begin{eqnarray}
&&\frac{\left\langle i i+1jj+1\right\rangle}{\left\langle Z_{i ,i+1}Z_{j,j+1} i+1j\right\rangle}\nb\\
&=&\frac{\left\langle Z_{i ,i+1}Z_{j,j+1}i+1i+2\right\rangle\left\langle ij-1jj+1\right\rangle+\left\langle Z_{i ,i+1}Z_{j,j+1}ii+2\right\rangle\left\langle j-1jj+1i+1\right\rangle}{\left\langle Z_{i ,i+1}Z_{j,j+1}i+1i+2\right\rangle\left\langle Z_{i ,i+1}Z_{j,j+1}j-1j\right\rangle},\nb
\end{eqnarray}
where $Z_{i ,i+1}\equiv Z_i +Z_{i+1}, Z_{j,j+1}\equiv Z_j+Z_{j+1}$. The right hand side of the equation above can be transformed as
\begin{eqnarray}
&&\frac{1}{\left\langle Z_{i ,i+1}Z_{j,j+1} i+1j\right\rangle \left\langle Z_{i ,i+1}Z_{j,j+1}i+1i+2\right\rangle\left\langle Z_{i ,i+1}Z_{j,j+1}j-1j\right\rangle}\nb\\
&\times&\left(\left\langle Z_{i ,i+1}Z_{j,j+1} i+1j\right\rangle \left\langle Z_{i ,i+1}Z_{j,j+1}i+1i+2\right\rangle\left\langle ij-1jj+1\right\rangle\right.\nb\\
&&+\left.\left\langle Z_{i ,i+1}Z_{j,j+1} i+1j\right\rangle \left\langle Z_{i ,i+1}Z_{j,j+1}ii+2\right\rangle\left\langle j-1jj+1i+1\right\rangle\right)\nb\\
&=&\frac{\left(\left\langle Z_{i ,i+1}Z_{j,j+1} i+1j\right\rangle \left\langle ij-1jj+1\right\rangle+ \left\langle Z_{i ,i+1}Z_{j,j+1}ij\right\rangle\left\langle j-1jj+1i+1\right\rangle\right)}{\left\langle Z_{i ,i+1}Z_{j,j+1} i+1j\right\rangle \left\langle Z_{i ,i+1}Z_{j,j+1}j-1j\right\rangle}\nb\\
&=&\frac{\left(\left\langle i j+1 i+1j\right\rangle \left\langle ij-1jj+1\right\rangle+ \left\langle i+1j+1ij\right\rangle\left\langle j-1jj+1i+1\right\rangle\right)}{\left\langle Z_{i ,i+1}Z_{j,j+1} i+1j\right\rangle \left\langle Z_{i ,i+1}Z_{j,j+1}j-1j\right\rangle}\nb\\
&=&\frac{\left\langle i i+1jj+1\right\rangle}{\left\langle Z_{i ,i+1}Z_{j,j+1} i+1j\right\rangle}\times \frac{\left( \left\langle ij-1jj+1\right\rangle- \left\langle j-1jj+1i+1\right\rangle\right)}{ \left\langle Z_{i ,i+1}Z_{j,j+1}j-1j\right\rangle}\nb\\
&=&\frac{\left\langle i i+1jj+1\right\rangle}{\left\langle Z_{i ,i+1}Z_{j,j+1} i+1j\right\rangle}. \nb
\end{eqnarray}
{And (\ref{UCM}) is therefore proven.}

This form of conversion is not unique: many forms can be constructed  equally up to a term  which vanishes under this unitary cut. A more general construction  is shown in
Lemma~\ref{cor:MHVFourTerm}. 

Here we convert it to another form related to the intersection of two planes (ii+1i+2) and (j-1jj+1), which is useful in the latter discussion.
\begin{eqnarray}\label{eq:cut}
\frac{\left\langle i i+1jj+1\right\rangle}{\left\langle AB i+1j\right\rangle}=\frac{\left\langle AB(ii+1i+2)\cap(j-1jj+1)\right\rangle}{\left\langle ABi+1i+2\right\rangle\left\langle ABj-1j\right\rangle}.
\end{eqnarray}

We use the same method to deal with pole $\left\langle ABj+1i\right\rangle$. According to (\ref{eq:cut}) and combining these two parts,
\begin{eqnarray}
\label{eq:CutPhys}
\mathcal{A}_{c_{j+1i}}&\equiv&\mathcal{A}_{c_{j+1i}}(j+1,\dots ,i|i+1,\dots, j)\\
&=&\frac{-\left\langle AB(i-1ii+1)\cap(j-1jj+1)\right\rangle\left\langle AB(jj+1j+2)\cap(i-1ii+1)\right\rangle}{\left\langle AB i-1i\right\rangle\left\langle AB ii+1\right\rangle\left\langle AB i+1i+2\right\rangle\left\langle AB j-1j\right\rangle\left\langle AB jj+1\right\rangle\left\langle AB j+1j+2\right\rangle}\nb.
\end{eqnarray}
In this paper $c_{ij}$ with $1\leqslant i<j\leqslant n-1$ denotes the cuts that divide the external lines into two groups $i\cdots j$ and $j+1\cdots i-1$. For convenience, we also use $C_{ij}$ to denote a cut set containing all the cuts $c_{12}\cdots c_{ij}$ in the union order. The corresponding $\mathcal{A}_{C_{ij}}$ denote the union of all the integrands from the cuts in $C_{ij}$. And $\mathcal{A}_{c_{ij}}$, $\mathcal{A}_{C_{ij}}$ may have the ambiguity of a rational function which will vanish under all the cuts in $C_{ij}$. We denote the corresponding arbitrary rational functions as $\mathcal{R}_{C_{ij}}$ or $\mathcal{R}_{c_{ij}}$.

\paragraph{Combining different unitarity cuts}~~

In the above paragraphs, we obtain an integrand $\mathcal{A}_c$ from each cut amplitudes. Such integrand contains only physical poles and is well-defined up to rational functions which will vanish under unitarity cuts. In this paragraph, we need to find a way to get the rational function   $\mathcal{A}_{C_{1n}}$ of the integrand such that  it is equal to  $\mathcal{A}_{c_{ij}}$ at the corresponding unitarity cuts. To this end, we define a operation $\cup$ to the constructed integrand, which means an g of two integrand from  different unitarity cuts.  After the union of all the possible unitarity cuts, we get the integrand $\mathcal{A}_{C_{1n}}$ automatically. According to the analysis in \cite{Bern94} on the unitary cut constructible for super-Yang-Mills theory, we can get $\mathcal{A}_{P}=\mathcal{A}_{C_{1n}}$.

Before the definition on $\cup$, we first define an union order for all the possible unitarity cuts as shown in Tab. \ref{tab:UnionOrder}, where the cut is label by one group of external legs in color order.
\begin{table}[htdp]
\caption{Union order of unitarity cuts}
\begin{center}
\begin{tabular}{cccccccccc}
$c$&$~$& &~& & & \\ \hline
12&$\rightarrow$&23&~&34&$\cdots$&n-3 n-2&&n-2 n-1\\
&&$\downarrow$&$\nearrow$&$\downarrow$ &$\cdots$&$\downarrow$&&$\downarrow$ \\
   &&123&&234&$\cdots$&n-4 n-3 n-2 && n-3 n-2 n-1\\
         &&&&$\downarrow$&$\cdots$& $\downarrow$&& $\downarrow$\\
      &&&&1234&$\cdots$&n-5 n-4 n-3 n-2&$\nearrow$& n-4 n-3 n-2 n-1\\
            &&&&&$\ddots$&$\vdots$&&$\vdots$\\
            &&&&~&&$\downarrow$& &$\downarrow$\\
            &&&&~&&1$\cdots$ n-2& &2$\cdots$n-1\\
\end{tabular}
\end{center}
\label{tab:UnionOrder}
\end{table}%

Now we define the operation $\cup$ on two rational functions in the function set group $\mathcal{A}_{c_{ij}}$ and $\mathcal{A}_{C_{ij}}$ as
\begin{eqnarray}
\mathcal{A}_{C_{ij}}\cup \mathcal{A}_{C_{i'j'}}.
\end{eqnarray}
The operation can be divided into two steps: First, choose a proper $\mathcal{R}_{C_{ij}}$ such that all the terms $T_{(ij),(i'j')}$, which have one cut in $C_{ij}$  and another cut  in $C_{i'j'}$, are the same in both $\mathcal{A}_{C_{ij}}$ and $\mathcal{A}_{C_{i'j'}}$, while other terms $T_{(ij)}$ or $T_{(i'j')}$ in $\mathcal{A}_{C_{ij}}$ and $\mathcal{A}_{C_{i'j'}}$  can only have one cut either in $C_{ij}$ or $C_{i'j'}$ respectively\footnote{In the following, we only verify this is possible for one-loop MHV amplitudes. And we will prove this for general ones in future work. Such procedure can also be generalized to other super Yang-Mills theory with lower super symmetry. };  Second,  add  all the terms of same formula once and all other terms. Hence we can get
\begin{eqnarray}
\mathcal{A}_{C_{ij}}\cup \mathcal{A}_{C_{i'j'}}=\sum T_{(ij),(i'j')}+\sum T_{(ij)}+\sum T_{(i'j')}+\mathcal{R}_{C_{ij}\cup C_{i'j'}}.
\end{eqnarray}
Such definition is similar for the union with $\mathcal{A}_{c_{ij}}$.

We unite the integrand from all the unitarity cuts in the union order one by one. We begin from the integrand $\mathcal{A}_{C_{12}}$ from the unitarity cut $c_{12}$, and then unite integrand $\mathcal{A}_{C_{34}}$ and $\mathcal{A}_{C_{12}}$ and then others in the order. Finally we can obtain $\mathcal{A}_{C_{1n}}$.  To this end, we first introduce some lemmas.
\begin{lemma}
\label{lem:UnitaryEq}
For a  pentagon integrand defined by five lines $\{\mathrm{L}_{i-1 i}, \mathrm{L_{ii+1}}$, $\mathrm{L_{j-1j}}$, $\mathrm{L_{jj+1}}$, $\mathrm{L_{j+1j+2}}\}$, where $\mathrm{L}_{i-1 i}=(Z_{i-1}Z_{i})$ and a plane $\mathcal{P}_j=(Z_{j-1}Z_{j}Z_{j+1})$ (or $\mathcal{P}_i=(Z_{i-1}Z_{i}Z_{i+1})$) and an arbitrary plane $\mathcal{P}_x=(Z_{x_1} Z_{x_2} Z_{x_3})$
\begin{eqnarray}
\label{eq:}
\frac{\bracf{AB\mathcal{P}_x\cap\mathcal{P}_j}}{\bracf{ABi-1i}\underline{\bracf{ABii+1}}\bracf{ABj-1j}\underline{\bracf{ABjj+1}}\bracf{ABj+1j+2}},\nb\\
\text{and} ~\frac{\bracf{AB\mathcal{P}_i\cap\mathcal{P}_x}}{\bracf{ABi-1i}\underline{\bracf{ABii+1}}\bracf{ABj-1j}\underline{\bracf{ABjj+1}}\bracf{ABj+1j+2}},\nb
\end{eqnarray}
$\exists~ \text{lines}~ \{\mathrm{Y}_2,\mathrm{Y}_3,\mathrm{Y}_4\}$ such that for any $Z_{m'}$ each integrand can be transformed to the following formulas
\begin{eqnarray}
\label{eq:FormRefLine}
&&{\bracf{AB\mathrm{Y}_2}\over \bracf{ABi-1i}\underline{\bracf{ABii+1}}\bracf{ABj-1j}\underline{\bracf{ABjj+1}}\bracf{AB m' m'+1}}\nb\\
&+&{\bracf{AB\mathrm{Y}_3}\over\bracf{ABi-1i}\underline{\bracf{ABii+1}}\underline{\bracf{ABjj+1}}\bracf{ABj+1j+2} \bracf{AB m' m'+1}}\nb\\
&+&{\bracf{AB\mathrm{Y}_4}\over \underline{\bracf{ABii+1}}\bracf{ABj-1j}\underline{\bracf{ABjj+1}}\bracf{ABj+1j+2} \bracf{AB m' m'+1}}+\mathcal{R}_{c_{i+1j}}\nb
\end{eqnarray}
under the unitarity cuts of the underlined propagators, where $\mathrm{Y}_i$ is proportional to a line which will keep the  scalar invariance of the  integrand.
\end{lemma}
\begin{proof}
We transform the integrand with numerator $\bracf{AB\mathcal{P}_x\cap\mathcal{P}_j}$ as following.
\begin{eqnarray}
\label{eq:}
&&\frac{\bracf{AB\mathcal{P}_x\cap\mathcal{P}_j}}{\bracf{ABi-1i}\underline{\bracf{ABii+1}}\bracf{ABj-1j}\underline{\bracf{ABjj+1}}\bracf{ABj+1j+2}}\\
&=&\frac{\bracf{AB m' m'+1}\bracf{AB\mathcal{P}_x\cap\mathcal{P}_j}}{\bracf{ABi-1i}\underline{\bracf{ABii+1}}\bracf{ABj-1j}\underline{\bracf{ABjj+1}}\bracf{ABj+1j+2}\bracf{AB m' m'+1}}\nb
\end{eqnarray}
Then we expand  the intersection of two plane  as
\begin{eqnarray}
\label{eq:}
\mathcal{P}_x\cap\mathcal{P}_j=a_1 \mathrm{L}_{j-1 j}+ a_2 \mathrm{L}_{j j+1}+a_3 \mathrm{L}_{j-1 j+1},
\end{eqnarray}
where the coefficients $a_i$ are constant which is related to the plane $\mathcal{P}_x$. It is obvious that only the  term with line $\mathrm{L}_{j-1 j+1}$ are  not obviously of the form in (\ref{eq:FormRefLine}). The numerator of this term  is
\begin{eqnarray}
\label{eq:Trans1}
\bracf{AB m' m'+1}\bracf{ABj-1 j+1}.
\end{eqnarray}
Since any point in twistor space of $\mathbb{CP}^3$ can be expand as four independent point. Here we choose four base points $\{Z_{j}, Z_{i-1}, Z_{i}, Z_{i+1}\}$ to expand points $\{Z_{m'} Z_{m'+1}\}$. Then we can expand bi-twistor
$(m' m'+1)$ based on $(j i-1)$, $(j i)$, $(j i+1)$, $(i-1 i)$, $(i i+1)$, $(i-1 i+1)$.  Then
\begin{equation}\label{eq:}
\bracf{AB m' m'+1}\bracf{ABj-1 j+1}\rightarrow\left\{ \begin{array}{c}
 \bracf{AB j i}\bracf{ABj-1 j+1}, ~\bracf{AB j i\pm1}\bracf{ABj-1 j+1} \\
\bracf{AB i i\pm1}\bracf{ABj-1 j+1}~~~ \\
\bracf{AB i-1 i+1}\bracf{ABj-1 j+1}
\end{array}\right. \nb
\end{equation}
According to the Schouten identity (\ref{ABId}), it is easy to see that only the terms with line $\mathrm{L}_{i-1 i+1}$ is not obvious to of the form (\ref{eq:FormRefLine})
\begin{eqnarray}
\label{eq:Trans2}
\bracf{AB i-1 i+1}\bracf{ABj-1 j+1}.
\end{eqnarray}
Then expanding point $Z_{j-1}$ as $\{Z_{j}, Z_{j+11}, Z_{j+2}, Z_{i}\}$,  (\ref{eq:Trans2}) can be transformed as
\begin{equation}\label{eq:}
\bracf{AB i-1 i+1}\bracf{ABj-1 j+1}\rightarrow\left\{ \begin{array}{c}
\bracf{AB i-1 i+1}\bracf{ABj j+1} \\
\bracf{AB i-1 i+1}\bracf{ABj+2 j+1}~~~ \\
\bracf{AB i-1 i+1}\bracf{ABi j+1}
\end{array}\right. \nb
\end{equation}
Finally all the terms are obviously of form (\ref{eq:FormRefLine}) according to
(\ref{ABId}). The integrand with numerator $\bracf{AB\mathcal{P}_i\cap\mathcal{P}_x}$ can be also proved similarly. \endofproof
\end{proof}
\begin{lemma}
\label{cor:MHVFourTerm}
For any unitarity cuts $c_{i+1j}$, we can choose a line $\mathrm{L}_{m'}=(m' m'+1)$ such that $\mathcal{A}_{c_{i+1j}}$ is
 \begin{eqnarray}
\mathcal{A}_{c_{i+1j}}&=&\mathcal I_{5}[\mathcal{P}_i\cap \mathcal{P}_j, \mathrm{L}_{m'}]+\mathcal I_{5}[\mathcal{P}_{i}\cap \mathcal{P}_{j+1}, \mathrm{L}_{m'}]\nn
&+&\mathcal I_{5}[\mathcal{P}_{i+1}\cap \mathcal{P}_{j}, \mathrm{L}_{m'}]+\mathcal I_{5}[\mathcal{P}_{i+1}\cap \mathcal{P}_{j+1}, \mathrm{L}_{m'}],
\end{eqnarray}
where $$\mathcal I_{5}[\mathcal{P}_i\cap \mathcal{P}_j, \mathrm{L}_{m'}]={\bracf{AB\mathcal{P}_i\cap \mathcal{P}_j} \bracf{i j \mathrm{L}_{m'}}\over \bracf{ABi-1 i}\bracf{AB i i+1}\bracf{ABj-1 j}\bracf{ABj j+1}\bracf{ABm' m'+1}}.$$
The geometry of the terms are shown in Fig.~\ref{GeneralCutGeometry}.
\end{lemma}
\begin{figure}[ht!]
  \centering
 \includegraphics[width=0.53\textwidth]{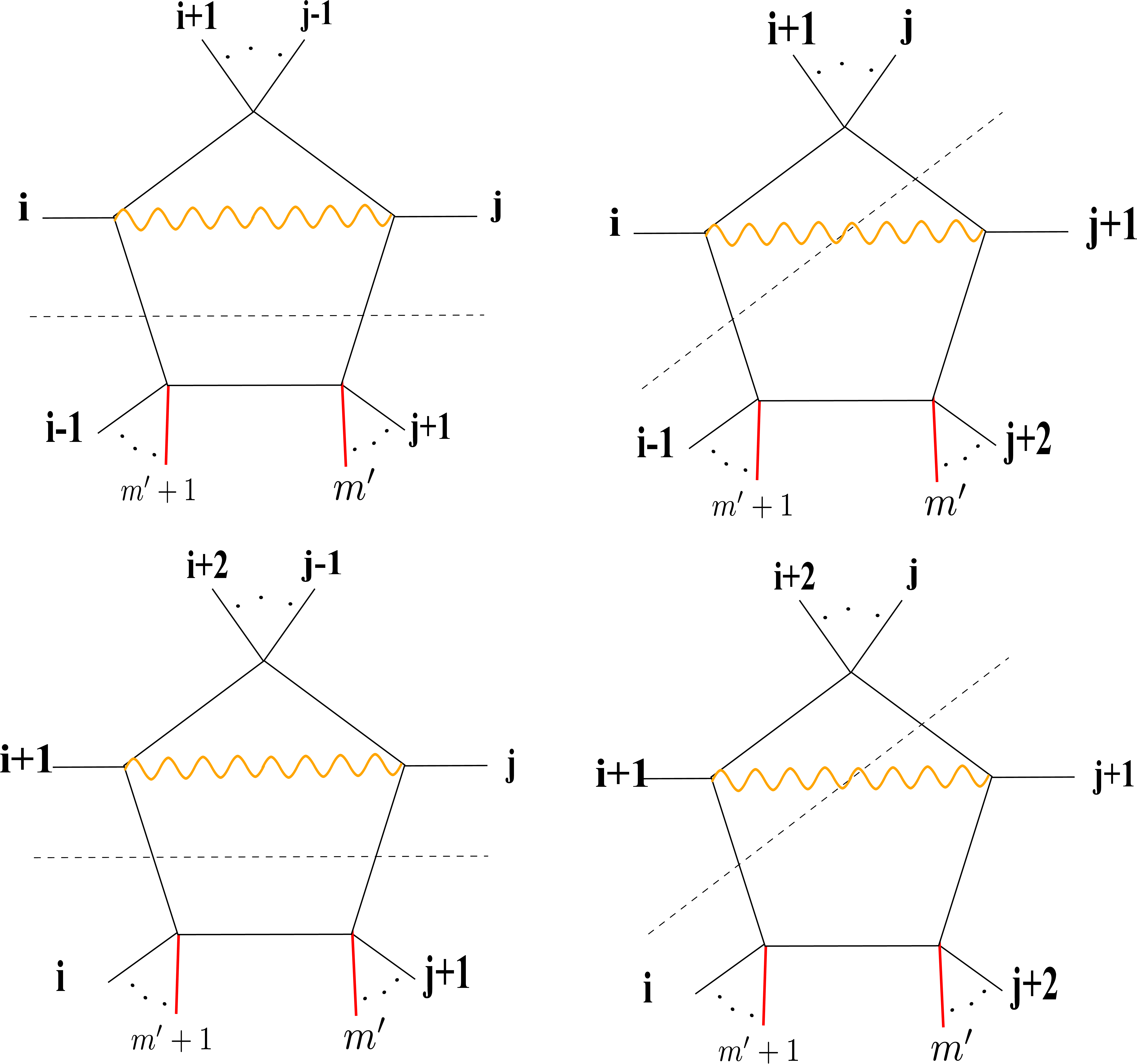}
 \caption{$\mathcal{A}_{c_{i+1j}}$ is the sum of these four terms without unphysical poles. The wavy line (ij) in this figure means the pentagon integrand $\mathcal I_5[\mathcal P_{i}\cap \mathcal P_j,(m'm'+1)]$ and the dash line denotes the unitary cuts. }
 \label{GeneralCutGeometry}
\end{figure}

\begin{proof}
According to Lemma~\ref{lem:UnitaryEq}, it is easy to see that the integrand (\ref{eq:CutPhys})  can be transformed to pentagons and boxes with one propagator $\bracf{ABm' m'+1}$. The explicit form can be calculated directly.
First, we can expand bi-twistor
$(Z_{m'}Z_{m'+1})$ based on six bi-twistors in the set $\mathcal{L}=\{(Z_{i}Z_{i+1})$, $(Z_{i}Z_{j+1})$, $(Z_{i+1}Z_{j})$, $(Z_{i}Z_{j})$, $(Z_{j}Z_{j+1})$, $(Z_{j+1}Z_{i+1})\}$.
\begin{eqnarray}
Z_{m'}Z_{m'+1}=\sum\limits_{(Z_mZ_n)\in \mathcal{L}}\frac{\left\langle klm'm'+1\right\rangle}{\left\langle ii+1jj+1\right\rangle}Z_mZ_n
\end{eqnarray}
where $(Z_{k}Z_{l})$ is the line in set $\mathcal{L}$ which do not have the common point with $(Z_mZ_n)$.
Then we add $(Z_A Z_{B})$ to get a $\left\langle\ ABm'm'+1\right\rangle$ as
 \begin{eqnarray}
\bracf{ABm'm'+1}&=&\frac{\left\langle ij+1m'm'+1\right\rangle}{\left\langle ii+1jj+1\right\rangle}\left\langle ABi+1j\right\rangle+\frac{\left\langle i+1jm'm'+1\right\rangle}{\left\langle ii+1jj+1\right\rangle}\left\langle ABij+1\right\rangle\\
&~&\frac{\left\langle j+1i+1m'm'+1\right\rangle}{\left\langle ii+1jj+1\right\rangle}\left\langle ABij\right\rangle+\frac{\left\langle ijm'm'+1\right\rangle}{\left\langle ii+1jj+1\right\rangle}\left\langle ABj+1i+1\right\rangle,\nb
\end{eqnarray}
where the unitarity cut condition $\left\langle ABii+1\right\rangle,\left\langle ABjj+1\right\rangle=0$ has been applied.

Then we add term $\left\langle ABm'm'+1\right\rangle$ to both denominator and numerator of $\mathcal{A}_{c_{i+1j}}$, which will not affect the final answer. We first deal with the term
$$
\mathcal{A}_{c_{i+1j}}\frac{\left\langle i+1jm'm'+1\right\rangle}{\left\langle ii+1jj+1\right\rangle}\frac{\left\langle ABij+1\right\rangle}{\left\langle ABm'm'+1\right\rangle}.
$$
According to (\ref{eq:cut}) and (\ref{eq:CutPhys}), we can get $$
\mathcal{A}_{c_{i+1j}}=\frac{-\left\langle AB(ii+1i+2)\cap(j-1jj+1)\right\rangle\left\langle i i+1jj+1\right\rangle}{\left\langle AB ii+1\right\rangle\left\langle AB jj+1\right\rangle\left\langle ABi+1i+2\right\rangle\left\langle ABj-1j\right\rangle\left\langle AB j+1i\right\rangle}.$$
Put them together, the term $\left\langle\ ABij+1\right\rangle$ vanishes, resulting
 \begin{eqnarray}
&~&\frac{\left\langle AB(ii+1i+2)\cap(j-1jj+1)\right\rangle\left\langle  i+1jm'm'+1\right\rangle}{\left\langle AB ii+1\right\rangle\left\langle ABi+1i+2\right\rangle\left\langle ABj-1j\right\rangle\left\langle AB jj+1\right\rangle\left\langle AB m'm'+1\right\rangle}\nn
&=&\mathcal I_{5}[\mathcal{P}_{i+1}\cap \mathcal{P}_{j},(m',m'+1)],
\end{eqnarray}
and this is exactly one term of the final answer.
Similarly, we can get $\mathcal I_{5}[\mathcal{P}_{i}\cap \mathcal{P}_{j+1},(m',m'+1)]$ from $\frac{\left\langle ABi+1j\right\rangle}{\left\langle ABm'm'+1\right\rangle}$.
In order to get the other two terms, we need to transform as
\begin{eqnarray}
\mathcal{A}_{c_{i+1j}}&=&-\frac{\left\langle i i+1jj+1\right\rangle^2}{\left\langle ABj+1i\right\rangle\left\langle ABii+1\right\rangle\left\langle ABi+1j\right\rangle\left\langle ABjj+1\right\rangle}\nn
&=&-\frac{\left\langle i i+1jj+1\right\rangle^2}{\left\langle ABii+1\right\rangle\left\langle ABjj+1\right\rangle\left\langle ABj+1i+1\right\rangle\left\langle ABij\right\rangle}
\end{eqnarray}
based on Schouten identity and unitarity cut condition.
Also, use the relation between dash line and wavy line, we can get
\begin{eqnarray}
\frac{\left\langle i i+1jj+1\right\rangle}{\left\langle AB ij\right\rangle}=-\frac{\left\langle AB(i-1ii+1)\cap(j-1jj+1)\right\rangle}{\left\langle ABi-1i\right\rangle\left\langle ABj-1j\right\rangle}
\end{eqnarray}
if we use this equation to replace $\left\langle AB ij\right\rangle$, while the remaining $\left\langle ABj+1i+1\right\rangle$ will be canceled by terms in $\left\langle ABm'm'+1\right\rangle$, and will become $\mathcal I_{5}[\mathcal{P}_{i}\cap \mathcal{P}_{j},(m',m'+1)]$
\endofproof
\end{proof}
\begin{lemma}
\label{lem:RefCutEqu}
Under the unitarity cut $c_{i+1n-1}$ we have
\begin{eqnarray}
\label{eq:Lemma2}
&&\mathcal I_{5}[\mathcal P_{j}\cap \mathcal P_{i},(n-1n)]+\mathcal I_{5}[\mathcal P_{j}\cap \mathcal P_{i+1},(n-1n)]\nn
&=&\mathcal H_{j-1j}-\mathcal H_{jj+1}+\mathcal{R}_{c_{i+1n-1}},
\end{eqnarray}
where
\begin{eqnarray}
\mathcal H_{j-1j}&=&{\bracf{j-1j(ABn-1)\cap\mathcal{P}_{i}}\bracf{ii+1i+2n}\over\bracf{ABj-1j}\bracf{ABi-1 i}\underline{\bracf{ABii+1}}\bracf{ABi+1i+2}\underline{\bracf{ABn-1n}}}\nb\\
&+&{\bracf{j-1j(ABn)\cap\mathcal{P}_{i}}\bracf{n-1ii+1i+2}\over\bracf{ABj-1j}\bracf{ABi-1 i}\underline{\bracf{ABii+1}}\bracf{ABi+1i+2}\underline{\bracf{ABn-1n}}}.
\end{eqnarray}
\end{lemma}
\begin{proof}
First we expand $\bracf{AB\mathcal{P}_j\cap\mathcal{P}_{i+1}}$ and add a term $\bracf{ABi-1i}$ both in the numerator and denominator. Since two parts have the same denominator, we only deal with the numerators. We first deal with $\bracf{AB\mathcal{P}_j\cap\mathcal{P}_{i+1}}\bracf{ABi-1i}$
\begin{eqnarray}
&&\bracf{AB\mathcal{P}_j\cap\mathcal{P}_{i+1}}\nn
&=&\bracf{ABj-1j}\bracf{j+1ii+1i+2}+\bracf{ABjj+1}\bracf{j-1ii+1i+2}+\bracf{ABj-1j+1}\bracf{ii+1i+2j}\nb.
\end{eqnarray}
We expand $Z_{i-1}$ in $\bracf{ABi-1i}$ based on Schouten identity in each term in previous equation and obtain six terms. The sum of three terms which have the same factor $\bracf{i-1ii+1i+2}$ is $\bracf{i-1ii+1i+2}(\bracf{ABj-1j}\bracf{ABj+1i}+\bracf{ABjj+1}\bracf{ABj-1i}-\bracf{ABj-1j+1}\bracf{ABji})$
Three terms in the bracket equals zero based on Schouten identity. And the sum of remaining terms forms $\bracf{AB\mathcal{P}_j\cap\mathcal{P}_{i}}\bracf{ABi+2i}$.
Now we can get an very important equation
\begin{eqnarray}
\bracf{AB\mathcal{P}_j\cap\mathcal{P}_{i+1}}\bracf{ABi-1i}=\bracf{AB\mathcal{P}_j\cap\mathcal{P}_{i}}\bracf{ABi+2i}
\end{eqnarray}
The sum of the two numerators in \ref{eq:Lemma2} is
\begin{eqnarray}\label{eq:Numerator}
\bracf{AB\mathcal{P}_j\cap\mathcal{P}_{i}}(\bracf{ABi+2i}\bracf{ji+1n-1n}+\bracf{ABi+1i+2}\bracf{jin-1n})
\end{eqnarray}
Now we deal with $(\bracf{ABi+2i}\bracf{ji+1n-1n}+\bracf{ABi+1i+2}\bracf{jin-1n})$
\begin{eqnarray}
&~&\bracf{ABi+2i}\bracf{ji+1n-1n}+\bracf{ABi+1i+2}\bracf{jin-1n}\\
&=&\bracf{ii+1(ABi+2)\cap(jn-1n)}\nn
&=&\bracf{ii+1Bi+2}\bracf{Ajn-1n}+\bracf{ii+1Ai+2}\bracf{jn-1nB}\nn
&=&\bracf{BA(ii+1i+2)\cap(jn-1n)}\nn
&=&\bracf{ABn-1j}\bracf{ii+1i+2n}+\bracf{ABnj}\bracf{n-1ii+1i+2}\nb,
\end{eqnarray}
where unitarity condition related to $\bracf{ABii+1}$ and $\bracf{ABn-1n}$ has been applied.
There will be three terms related to $\bracf{ABj-1j}$, $\bracf{ABjj+1}$ and $\bracf{ABj-1j+1}$ in $\bracf{AB\mathcal{P}_j\cap\mathcal{P}_{i}}$. However, according to $\bracf{ABj-1j+1}\bracf{ABkj}=\bracf{ABj-1k}\bracf{ABj+1j}+\bracf{ABj-1j}\bracf{ABkj+1}$ ($Z_k$ could be any twistor), (\ref{eq:Numerator}) only have two terms related to $\bracf{ABj-1j}$, $\bracf{ABjj+1}$.
The sum of terms contain $\bracf{ABj-1j}$ is
\begin{eqnarray}
&&\big(\bracf{ABn-1j} \bracf{ii+1i+2n} \bracf{j+1i-1ii+1}\nn
&+&\bracf{ABnj}\bracf{n-1ii+1i+2}\bracf{j+1i-1ii+1}\nn
&+&\bracf{ABn-1j+1}\bracf{ii+1i+2n}\bracf{i-1ii+1j}\nn
&+&\bracf{ABnj+1}\bracf{n-1ii+1i+2}\bracf{i-1ii+1j}\ \big) \bracf{ABj-1j}\nn
&=\big(& \bracf{j+1j(ABn-1)\cap(i-1ii+1)}\bracf{ii+1i+2n}\nn
&+&\bracf{j+1j(ABn)\cap(i-1ii+1)}\bracf{n-1ii+1i+2}\ \big)\bracf{ABj-1j}.
\end{eqnarray}
Similarly, we can get $\bracf{ABjj+1}(\bracf{j-1j(ABn-1)\cap(i-1ii+1)}\bracf{ii+1i+2n}+\bracf{j-1j(ABn)\cap(i-1ii+1)}\bracf{n-1ii+1i+2})$.
Combine numerator and denominator, we can get $\mathcal H_{j-1j}-\mathcal H_{jj+1}$.
\endofproof
\end{proof}
\begin{lemma}
\label{lem:RefCutR}
The integrand $\mathcal{A}_{c_{i+1n-1}}$ is equal to the integrand $\mathcal{A}_{C_{1n-2}}$ with reference line $\mathrm{L}_{n-1}$ under unitarity cut $c_{i+1n-1}$.
\end{lemma}
\begin{proof}
First we find out the terms in $\mathcal{A}_{C_{1n-2}}$ which contain both $\bracf{ABii+1}$ and $\bracf{ABn-1n}$.
We need to prove
\begin{eqnarray}
\mathcal{A}_{c_{i+1n-1}}&=&\sum\limits_{j=1}^{i-1}(\mathcal I_{5}[\mathcal P_{j}\cap \mathcal P_{i},(n-1n)]+\mathcal I_{5}[\mathcal P_{j}\cap \mathcal P_{i+1},(n-1n)])\nn
&+&\sum\limits_{j=i+2}^{n-2}(\mathcal I_{5}[\mathcal P_{j}\cap \mathcal P_{i},(n-1n)]+\mathcal I_{5}[\mathcal P_{j}\cap \mathcal P_{i+1},(n-1n)])\nn
&+&\mathcal I_{5}[\mathcal P_{i}\cap \mathcal P_{i+1},(n-1n)])
\end{eqnarray}
According to Lemma~\ref{lem:RefCutEqu},
the term $-\mathcal H_{jj+1}$ in
$$\mathcal I_{5}[\mathcal P_{j}\cap \mathcal P_{i},(n-1n)]+\mathcal I_{5}[\mathcal P_{j}\cap \mathcal P_{i+1},(n-1n)]$$
will be cancelled by the terms of  $\mathcal H_{jj+1}$ in
$$
\mathcal I_{5}[\mathcal P_{j+1}\cap \mathcal P_{i},(n-1n)]+\mathcal I_{5}[\mathcal P_{j+1}\cap \mathcal P_{i+1},(n-1n)]~.
$$

After summing  over all possible terms in $c_{i+1\,n-1}$,
we have only four terms left
\begin{eqnarray}
\mathcal{A}_{c_{i+1n-1}}=\mathcal H_{(n1)}-\mathcal H_{(i-1i)}+\mathcal H_{(i+1i+2)}-\mathcal H_{(n-2n-1)}+\mathcal I_{5}[\mathcal P_{i}\cap\mathcal P_{i+1},(n-1n)]~.\nn
\end{eqnarray}
For every term in the equation,
\begin{eqnarray}
\mathcal H_{n1}&=&\mathcal I_{5}[\mathcal P_{i+1}\cap\mathcal P_{n},(i-1i)]\nn
\mathcal H_{i-1i}&=&0\nn
\mathcal H_{i+1i+2}&=&-\mathcal I_{5}[\mathcal P_{i}\cap\mathcal P_{i+1},(n-1n)]\nn
\mathcal H_{n-2n-1}&=&-\mathcal I_{5}[\mathcal P_{i+1}\cap\mathcal P_{n-1},(i-1i)]~.
\end{eqnarray}
Therefore
\begin{eqnarray}
\mathcal{A}_{c_{i+1n-1}}=\mathcal I_{5}[\mathcal P_{i+1}\cap\mathcal P_{n},(i-1i)]+\mathcal I_{5}[\mathcal P_{i+1}\cap\mathcal P_{n-1},(i-1i)]
\end{eqnarray}
This result is the same as $\mathcal{A}_{c_{i+1n-1}}$ in
Lemma~\ref{cor:MHVFourTerm},
if we pick $(Z_{m'}Z_{m'+1})=(Z_{i-1}Z_{i})$.
\endofproof
\end{proof}

\begin{theorem}
\label{theo:MHVUnion}
The integrand of MHV one loop amplitudes can be constructed by the union over all the integrand from each unitarity cut. If we choose the reference line to be $\mathrm{L}_{n-1}$, the integrand of MHV one loop amplitudes is
\begin{eqnarray}
\label{eq:MHVAmpOneLoop}
\mathcal{A}_{C_{1n}}=\mathcal{A}_{C_{1n-2}}=\bigcup\limits_{1\leqslant i<j<n-1}\mathcal{A}_{c_{ij}}=\sum_{1\leqslant i<j<n-1} \mathcal I_{5}[\mathcal{P}_i\cap \mathcal{P}_j, \mathrm{L}_{n-1}]
\end{eqnarray}
\end{theorem}
\begin{proof}
As discussed above, for convenience, we combine the $\mathcal{A}_{c_{ij}}$ in a specific order.
In general, when combining integrands $\mathcal{A}_{c_{i-1 j}}$ and $\mathcal{A}_{C_{ij}}$, the only cut-related terms in $\mathcal{A}_{C_{ij}}$ which may influence the unitarity cuts of $c_{i-1 j}$ are $\{c_{i-1 j-1}, c_{i j}, c_{i j-1}\}$. On the other hand, $\mathcal{A}_{c_{i-1 j}}$ may affect the cuts $\{c_{i-1 j-1}, c_{i j}, c_{i j-1}\}$ in $C_{ij}$.
Hence we only need verify that the cut-related terms in $\mathcal{A}_{C_{ij}}$ are either same with the terms in $\mathcal{A}_{c_{i-1 j}}$ or cut un-related with each other. This is easy to prove according to Lemma~\ref{cor:MHVFourTerm}. Hence we can unite all the terms except the last column in Tab.~\ref{tab:UnionOrder}.
According the Lemma~\ref{lem:RefCutR}, we do not need to combine  the integrand from the last column and  the union of $C_{1n-2}$ is just the final integrand from the unitarity cuts up to a rational function on all the unitarity cuts.  We find such construction of integrand is equivalent to the integrand from single cuts in~\cite{NimaAllLoop, NimaLocalInt}.
\endofproof
\end{proof}

\subsection{MHV nonplanar amplitudes and unitarity cuts}
In this section, we will present  a general recipe for dealing with MHV nonplanar amplitudes $U(N)$ Yang-Mills theory.
 One-loop four-point and five-point amplitudes are carefully worked out as examples.
Our results  verify directly the $U(1)$ decoupling relation of one loop amplitudes.

\paragraph{Properties of MHV nonplanar amplitudes under unitarity cuts}~

We  consider  first  the situation with only one nonplanar leg. Based on the permutation relations of Yangian
Invariants (\ref{eq:4point}) this amplitude under a unitarity cut can be converted  to a planar one at a price of  a simple factor $f_{kin}$ in Fig.~\ref{fig:MHVgeneral}
\begin{figure}[ht!]
  \centering
 \includegraphics[width=0.53\textwidth]{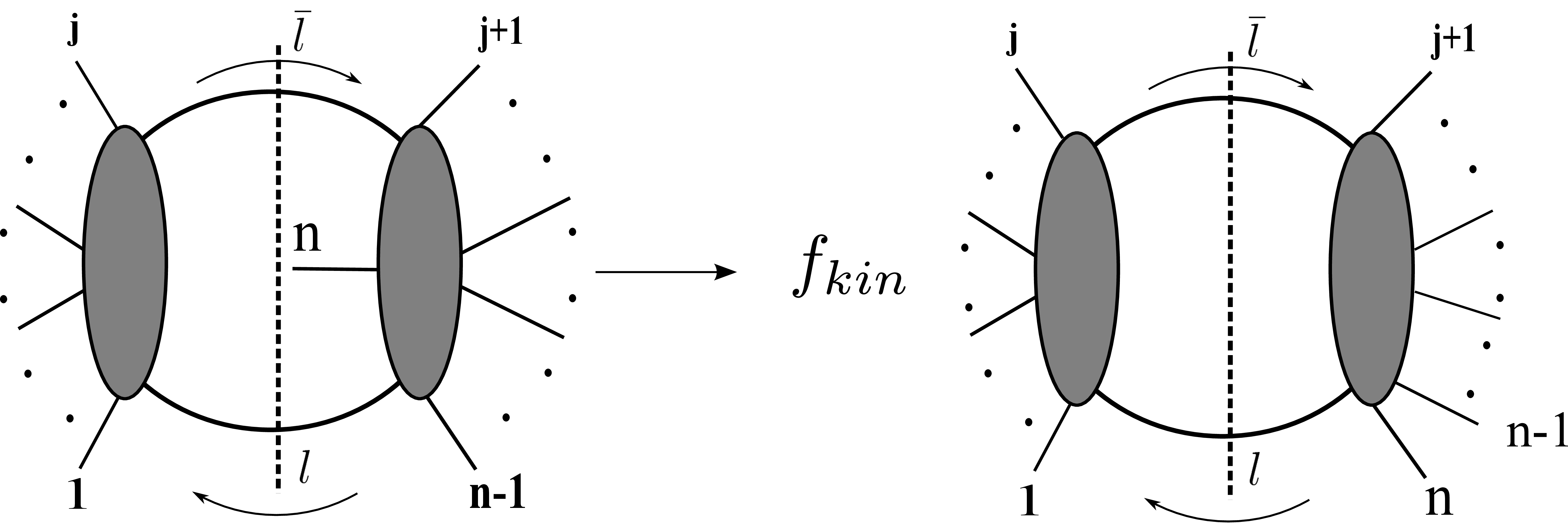}
 \caption{A nonplanar one-loop n-point  MHV amplitude, with one nonplanar leg marked ``$n$'',  can be converted  to a planar one at the price of a simple kinematic factor $f_{kin}$. }
 \label{fig:MHVgeneral}
\end{figure}
\begin{eqnarray}
f_{kin}\,  =\frac{\left\langle n-1n\right\rangle\left\langle \bar l l\right\rangle}{\left\langle n-1 l\right\rangle\left\langle n \bar l\right\rangle}.
\end{eqnarray}
This  step removes   one nonplanar leg.
For more than one nonplanar legs, say $m$,  we can repeat  this operation  $m$ times to arrive at a planar diagram.
Planar MHV amplitudes under unitarity cuts have been  discussed in Section~\ref{sec:Planar} we can therefore apply the planar results to obtain the nonplanar  MHV amplitude after a unitarity cut by
\begin{eqnarray}
A_c(1,2,\dots,j|j+1,\dots,n-1,\dot{n})=f_{kin}A_c(1,j|j+1,n),
\end{eqnarray}
where `$\dot i$' labels the nonplanar leg and $A$ is  planar  if there is no dotted leg in the arguments.
This equation clearly shows that  the difference between the planar result and the nonplanar one lies in  the kinematic factors.

\paragraph{Reconstructing the kinematic factors: }~

Amplitudes under a unitarity cut  only contains two variables related to the loop momentum.
Unitarity cut condition sets two cut propagators to zero, while fixing the other two variables.
With this in mind  we need to  reconstruct the amplitudes to a function of four variables.
A unitarity cut on nonplanar diagrams contains two pieces of information, the kinematic factor and the corresponding planar amplitudes. We need to discuss them separately.
As reconstruction of planar part is the same as above we deal with the kinematic factors here.

For one nonplanar leg (labelled n), we can read off from Fig.~\ref{fig:MHVgeneral} the kinematic factor as
\begin{eqnarray}
f_{kin}\,  =\frac{\left\langle n-1n\right\rangle\left\langle \bar l l\right\rangle}{\left\langle n-1 l\right\rangle\left\langle n \bar l\right\rangle}
\end{eqnarray}
Using Schouten identity, and expanding spinor
$\lambda_{n-1}$ based on $\lambda_{j+1}$ and $\lambda_{n}$ as
$$\lambda_{n-1} =
\frac{\left \langle n-1,\, n\right \rangle}
     {\left \langle j+1,\, n\right \rangle} \lambda_{j+1} +
\frac{\left \langle j+1, \, n-1\right \rangle}
     {\left \langle j+1,\, n \right \rangle}\lambda_n~, $$
$f_{kin}$  can be expanded as
\begin{eqnarray}
f_{kin}\,
= -1 + \frac{\left\langle ln\right\rangle\left\langle
    \bar{l}\,  j+1 \right\rangle}
  {\left\langle ln-1\right\rangle\left\langle \bar{l}\, n
      \right \rangle}
   \frac{\left \langle n-1, \, n\right \rangle}
     {\left \langle j+1, \, n\right \rangle}
 +\frac{\left\langle l\, n\right\rangle}
     {\left\langle l,\, n-1\right\rangle}
     \frac{\left \langle j+1, \, n-1\right\rangle}
     {\left \langle j+1, \, n\right \rangle}.
\end{eqnarray}
Under a unitarity cut, variables $l$ and $\bar l$ in momentum space are denoted by variables A and B in momentum twistor space.
So $\left\langle ln\right\rangle\left\langle \bar l j+1\right\rangle$
can be rewritten as
$\left\langle A\, n\right\rangle\left\langle B\, j+1\right\rangle$.
Using the equation
$$\left\langle A\, B\, i-1\, i\right\rangle
=\left\langle A\, B\right\rangle\left\langle i-1\, i\right\rangle(x-x_i)^2$$
we express $f_{kin}$ as
\begin{eqnarray}
\label{eq:Kfactor}
f_{kin}
= -1+ \frac{\left\langle A\, n\, B\, j+1\right\rangle^{o1}}
    {\left\langle A\, n-1\, B\, n\right\rangle^{o2}}
  \frac{\left \langle n-1\, n\right \rangle}
   {\left \langle j+1\, n\right \rangle}
 +\frac{\left\langle A\, n\, B\, j+1\right\rangle^{o1}}
   {\left\langle A\, n-1\, B\, j+1\right\rangle^{o3}}
   \frac{\left \langle j+1\, n-1\right \rangle}
      {\left \langle j+1\, n\right \rangle}
\end{eqnarray}
We take a  pause  here to clarify our notations.
The expansion of $\left\langle A\, B\, i-1\, i\right\rangle$ is valid if $A$ and $B$, $i-1$ and  $i$ are continuous in a certain color ordering  in twistor space.
Now $A$ and $n-1$,  and $B$ and $n$  are in fact not  continuous in the original color ordering.
We  can, nevertheless,  define a new color ordering,
labeled by  $^{o2}$ in equation, where in this ordering  $A$ and $n-1$, and $B$ and $n$  are both continuous.
In addition  if we choose
$$o1=(1,2,\ldots,j,j+1,\ldots,n)$$
$$o2=(1,2,\ldots,j,n,\ldots,n-1)$$
$$o3=(1,2,\ldots,j,j+1,\ldots,n,n-1)$$
then $(p_A\, +\, p_B)^2$ will appear both in the denominator and numerator and cancel each other.
Thus the  former equation (\ref{eq:Kfactor}) is naturally true.

The next step is to combine kinematic factors with the planar amplitudes.
Recall that  the planar amplitudes
\begin{eqnarray}
&&\mathcal A_0(1,\dots,j|j+1,\dots,n) \nn
&=&  \mathcal A_{n}^{tree_{MHV}}  \cdot
 \frac{-\left\langle j+1\, i\,  i+1\, j\right\rangle^2}
      {\left\langle A\, B\, j+1\, i\right\rangle
       \left\langle A\, B\, i\, i+1\right\rangle
       \left\langle A\, B\, i+1\, j\right\rangle
       \left\langle A\, B\, j\, j+1\right\rangle}
\end{eqnarray}
(Here we add the MHV tree amplitudes as a coefficient because, in nonplanar situation, this coefficient is important and can
 not be omitted).
It is then clear that the numerator
$\left\langle A\, n\, B\, j+1\right\rangle^{o1}$
will cancel the same term in the denominator of
$\mathcal A_0\, (1,2,\ldots,j|j+1,\ldots,n)$
by adding  a new term related to $o2$ or $o3$, and consequently changing  the color ordering.
Hence
\begin{eqnarray}
 &&  f_{kin}\,  \mathcal A_0\, (1,\dots,j|j+1,\dots,n) \nn
&=&-\, \mathcal A_0\, (1,\dots,j|j+1,\dots,n)
       -\, \mathcal A_0\, (1,\dots,j|n,\dots,n-1)   \nn
&&-\, \frac{\left\langle n\,n-2\right\rangle
         \left\langle n-1\,j+1\right\rangle}
      {\left\langle n-2\, n-1\right\rangle
        \left\langle n\, j+1\right\rangle}\,
     \cdot \mathcal A_0\, (1,\dots,j|j+1,\dots,n-1)
\end{eqnarray}
This is nothing but  $\mathcal A_0$ with unphysical poles, which needs  to be  converted to $\mathcal A_1$, which only contains physical poles, to arrive at the  final results. This can be achieved in a way completely analogous to the procedures presented in Section~\ref{sec:Planar}.

\paragraph{Comparison with $U(1)$ decoupling relation}~~

As alluded above a relation of reconstructed amplitudes
from permutation relation of Yangian Invariants has only these three terms  (or two terms in the case of four-point).
\begin{eqnarray}
 &&   \mathcal A_1(1,2,\dots,j|j+1,\dots,n-1,\dot{n})  \nn
&=&  f_{kin}\, \mathcal A_1(1,2,\ldots,j|j+1,\ldots,n) \nn
&=&  -\, \mathcal A_1(1,2,\ldots,j|j+1,\ldots,n)
    -\, \mathcal A_1(1,2,\ldots,j|n,\ldots,n-1)  \nn
&& -\, \frac{\left\langle n\,n-2\right\rangle
          \left\langle n-1\, j+1\right\rangle}
        {\left\langle n-2\, n-1\right\rangle
          \left\langle n\, j+1\right\rangle}\,
      \cdot \mathcal A_1\, (1,2,\ldots,j|j+1,\ldots,n,n-1)\nb
\end{eqnarray}
When considering the $U(1)$ decoupling relation
\begin{eqnarray}
&& \mathcal A_1\, (1,2,\ldots,j|j+1,\ldots,n-1,\dot{n}) \nn
&=&-\mathcal A_1\, (1,2,\ldots,j|j+1,\ldots,n)
   -\mathcal A_1(1,2,\ldots,j|n,\ldots,n-1)  \nn
&~&-\sum\limits_{i}
   \mathcal A_1(1,2,\ldots,j|j+1,\ldots,i,n,i+1,\ldots,n-1),
\end{eqnarray}
more terms will appear.
However, due to the special property of MHV amplitudes, the exchange of external legs which are not the four legs connecting the ``basic twin-box''  will not affect the final result of this cut.
That is to say, all of the possible cases in this set have the same cut amplitudes.
The only difference is the pre-factor, MHV tree amplitude.
If we expand the factor
$$
\frac{\left\langle n\,  n-2\right\rangle
        \left\langle n-1\, j+1\right\rangle}
       {\left\langle n-2\, n-1\right\rangle
        \left\langle n\, j+1\right\rangle}
$$
we can find cuts with all  possible color ordering,
which are the same as obtained from $U(1)$ decoupling.
This  form is obviously more compact than $U(1)$ decoupling, since it combines,  in one integral,  terms related to planar amplitudes of different orders but having same result under a given  unitarity cut.

\paragraph{Final results of MHV nonplanar amplitudes:}~

Final results of MHV nonplanar amplitudes are the union of all possible $f_{kin}\, \mathcal A_1$
\begin{eqnarray}
\mathcal A_{NP}=\bigcup\limits_{i}f_{kin}^i\mathcal A_{1i}
\end{eqnarray}
where $\mathcal A_{NP}$ stands for  the final results of nonplanar amplitudes and $\mathcal A_{P}$ is the final results of planar counterparts.

This formula just give a procedure, which is the same as the planar situation, to get the final result. However, to get the general formula of integrand, a lemma needs to be proved. When two terms in two different cut results are cut related (one term has propagators of the other unitarity cuts), one of them can convert to terms which are the same as the other one and others that are cut un-related. We will prove it in future work. In the following sections, we will show this method is valid in some particular examples of four-point and five-point situation.

\subsection{One-loop four-point MHV nonplanar amplitudes}

In this section we describe a nonplanar four point amplitude with one nonplanar leg (we call ``(3+1)'' case) as an explicit example of our construction.

We start with the cut result of this case, which is shown in Fig.~\ref{fig:stu-cuts}. First, we convert this cut amplitude to planar one with a kinematic factor.
in Fig.~\ref{fig:NPtoP}.
\begin{figure}[ht!]
  \centering
 \includegraphics[width=0.83\textwidth]{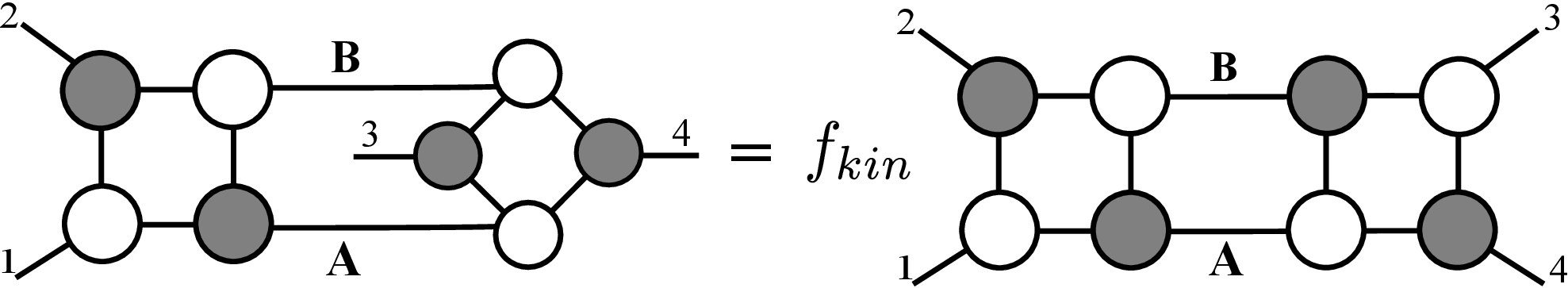}
 \caption{Convert the nonplanar cut amplitude to planar one with a kinematic factor}
  \label{fig:NPtoP}
\end{figure}

The way to reconstruct the planar diagram is simply adding two BCFW bridges across legs (2 3) and (1 4). The result is exactly the one loop four-point planar amplitude (\ref{fig:4pointloop1}).
\begin{figure}[ht!]
  \centering
 \includegraphics[width=0.83\textwidth]{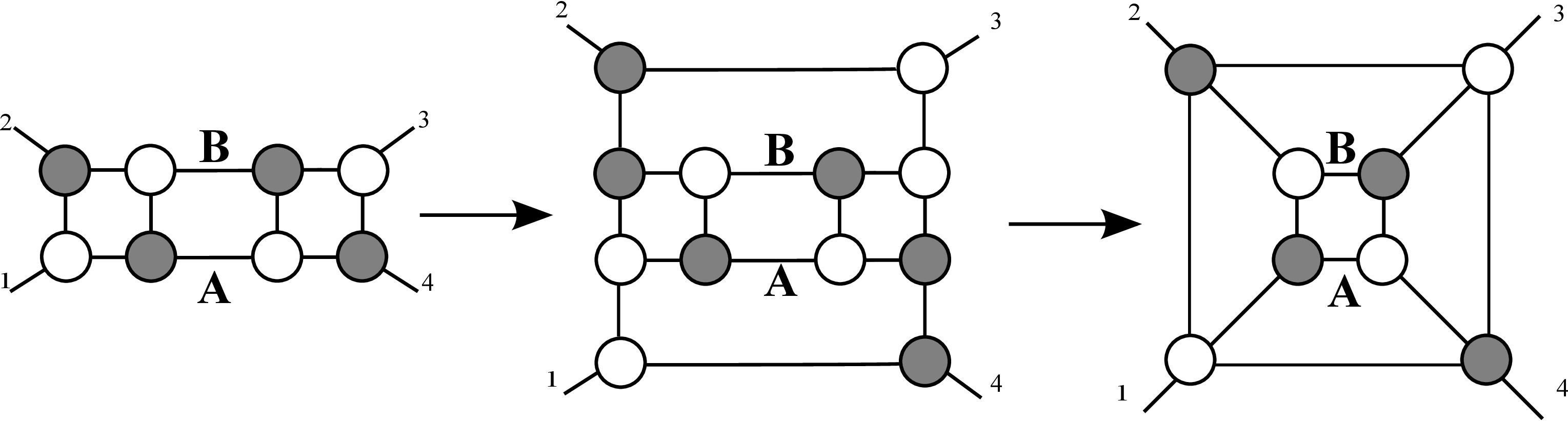}
 \caption{Way of reconstruct the planar diagram. First step is adding two BCFW bridges. Using ``square moves'' and ``merges'' to represent the diagram as shown in~\cite{NimaGrass}. }
 \label{fig:4pointloop1}
\end{figure}
We can simply write down the result of the diagram above in the form of loop integrand.
\begin{eqnarray}
\mathcal{A}_0(1,2,3,4)=\frac{-\left\langle 1234\right\rangle^2}{\left\langle AB12\right\rangle\left\langle AB23\right\rangle\left\langle AB34\right\rangle\left\langle AB41\right\rangle}
\end{eqnarray}
Here $\mathcal A_0=\mathcal A_1$ because there are only four points, we need not to remove unphysical poles.

Using the unitarity cut condition$\left\langle AB23\right\rangle=0,\left\langle AB41\right\rangle=0$,we can rewrite the kinematic factor as
$$f_{kin}=\frac{s_{ll_2}}{s_{l3}}=-1-\frac{\left\langle AB34\right\rangle^{o1}}{\left\langle AB43\right\rangle^{o2}}$$
Here we define two color orders, $o^1$ means color order (1,2,3,4), and $o^2$ means (1,2,4,3).
Combine the factor and planar result, we can get
\begin{eqnarray}
&& \mathcal{A}_1(1,2|3,\dot 4) \nn
&=& \mathcal A_{4}^{MHVtree}(1,2,3,4)\,
 (1+\frac{\left\langle A\, B\, 3\, 4\right\rangle^{o1}}
         {\left\langle A\, B\, 4\, 3\right\rangle^{o2}})
 \frac{\left\langle 1\, 2\, 3\, 4\right\rangle^2}
      {\left\langle A\,B\,1\,2\right\rangle
         \left\langle A\, B\, 2\, 3\right\rangle
         \left\langle A\, B\, 3\, 4\right\rangle
         \left\langle A\, B\,4\,1\right\rangle}   \nn
&=&
 \mathcal A_{4}^{MHVtree}(1,2,3,4)\,
 \frac{\left\langle 1\, 2\, 3\,4\right\rangle^2}
    {\left\langle A\,B\,1\,2\right\rangle^{o1}
     \left\langle A\,B\,2\,3\right\rangle^{o1}
      \left\langle A\,B\,3\,4\right\rangle^{o1}
      \left\langle A\,B\,4\,1\right\rangle^{o1}}\nn
&~&
+\mathcal A_{4}^{MHVtree}(1,2,4,3)
 \frac{\left\langle 1\,2\,4\,3\right\rangle^2}
   {\left\langle A\,B\,1\,2\right\rangle^{o2}
    \left\langle A\,B\,2\,4\right\rangle^{o2}
    \left\langle A\,B\,4\,3\right\rangle^{o2}
    \left\langle A\,B\,3\,1\right\rangle^{o2}}
\end{eqnarray}
Other possible unitarity cuts should be taken into consideration. The steps to deal with all possible unitarity cuts are as described above.
 Since each term is actually the planar four-point one-loop amplitude with a certain color ordering, we can simplify the expression as
\begin{eqnarray}
\mathcal{A}_1(1,2|3,\dot 4)=\mathcal{A}_P(1,2,3,4)+\mathcal{A}_P(1,2,4,3)\nonumber\\
\mathcal{A}_1(2,3|1,\dot 4)=\mathcal{A}_P(2,3,1,4)+\mathcal{A}_P(2,3,4,1)\nonumber\\
\mathcal{A}_1(3,1|2,\dot 4)=\mathcal{A}_P(3,1,2,4)+\mathcal{A}_P(3,1,4,2)\nonumber
\end{eqnarray}

Now we need to unite $\mathcal{A}_1(2,3|1,\dot 4)$ and $\mathcal{A}_1(1,2|3,\dot 4)$. Obviously, the terms with $o1$ are the same in these two. While it is not obvious to judge whether other terms are cut related or not. However, we can set one of the cut propagator in all terms as $l^2$ in momentum space and write down the denominators of these three terms.
\begin{eqnarray}
o1(1,2,3,4)&:&\ \underline{l^2}(l-p_1)^2\underline{(l-p_1-p_2)^2}(l+p_4)^2\nn
o2(1,2,4,3)&:&\ \underline{l^2}(l-p_1)^2\underline{(l-p_1-p_2)^2}(l+p_3)^2\nn
o3(2,3,1,4)&:&\ \underline{l^2}(l-p_1)^2\underline{(l-p_1-p_4)^2}(l+p_3)^2
\end{eqnarray}
Now we can clearly find out  that the terms of $o3$ in $\mathcal{A}_1(2,3|1,\dot 4)$ do not have common unitarity cuts with the terms in $\mathcal{A}_1(1,2|3,\dot 4)$ and vice versa. So, according to the definition of operation union,
\begin{eqnarray}
\mathcal{A}_1(1,2|3,\dot 4)\cup \mathcal{A}_1(2,3|1,\dot 4) =\mathcal{A}_P(1,2,3,4)+\mathcal{A}_P(1,2,4,3)+\mathcal{A}_P(3,1,2,4).
\end{eqnarray}
So we can obtain the final result of nonplanar amplitude as
 \begin{eqnarray}
\mathcal{A}_{NP}(1,2,3,\dot 4)&=&\bigcup\limits_{i=1}^{3}(\mathcal{A}_{NP}(i,i+1|i+2,\dot 4)\nn
&=&\mathcal{A}_P(1,2,3,4)+\mathcal{A}_P(1,2,4,3)+\mathcal{A}_P(1,4,2,3),
\end{eqnarray}
which is the familiar result of nonplanar ``(3+1)'' case~\cite{Bern94}
for $U(N)$ Yang-Mills theory. This equation imply that the $U(1)$ gauge field will decouple from the $SU(N)$ part of the $U(N)$ gauge fields.

\subsection{One-loop five-point MHV nonplanar amplitudes}
We first obtain all the unitarity cuts of five point nonplanar amplitudes $A_c(1,2,3,4,\dot{5})$ as follows,
\begin{itemize}
\item[Cut Type I] \begin{eqnarray}\label{cut5I}
A_{c}(1, \dot{5} | 2,3,4), A_{c}(2, \dot{5} | 3,4,1), A_{c}(3, \dot{5} | 2,4,1),A_{c}(4, \dot{5} | 3,2,1) \nb
\end{eqnarray}
\begin{figure}[ht!] \label{fig:4pointloop2}
  \centering
 \includegraphics[width=0.83\textwidth]{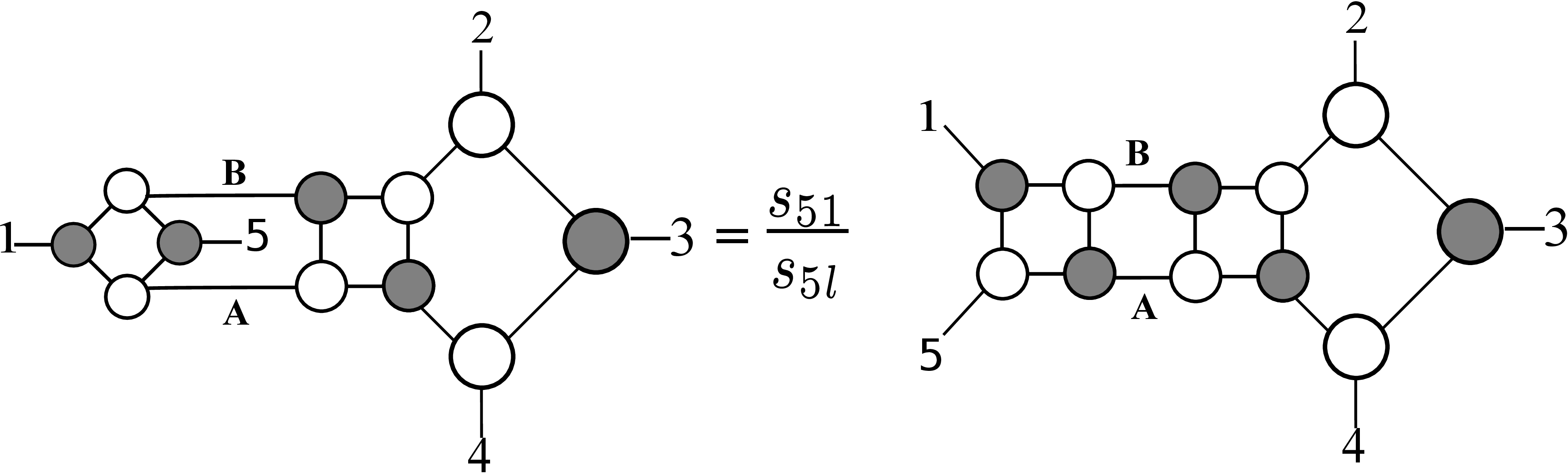}
  \caption{On-shell diagram of cut amplitude $\mathcal{A}_{c}(1, \dot{5} | 2,3,4)$}
\end{figure}
\item[Cut Type II]  \begin{eqnarray}\label{cut5II}
A_{c}(1,2 | 3,4, \dot{5}), A_{c}(1,2, \dot{5} | 3,4),  A_{c}(2,3 | 4, 1,\dot{5}), A_{c}(2,3,\dot{5} | 4, 1).   \nb
\end{eqnarray}
\begin{figure}[ht!] \label{fig:4pointloop3}
  \centering
 \includegraphics[width=0.53\textwidth]{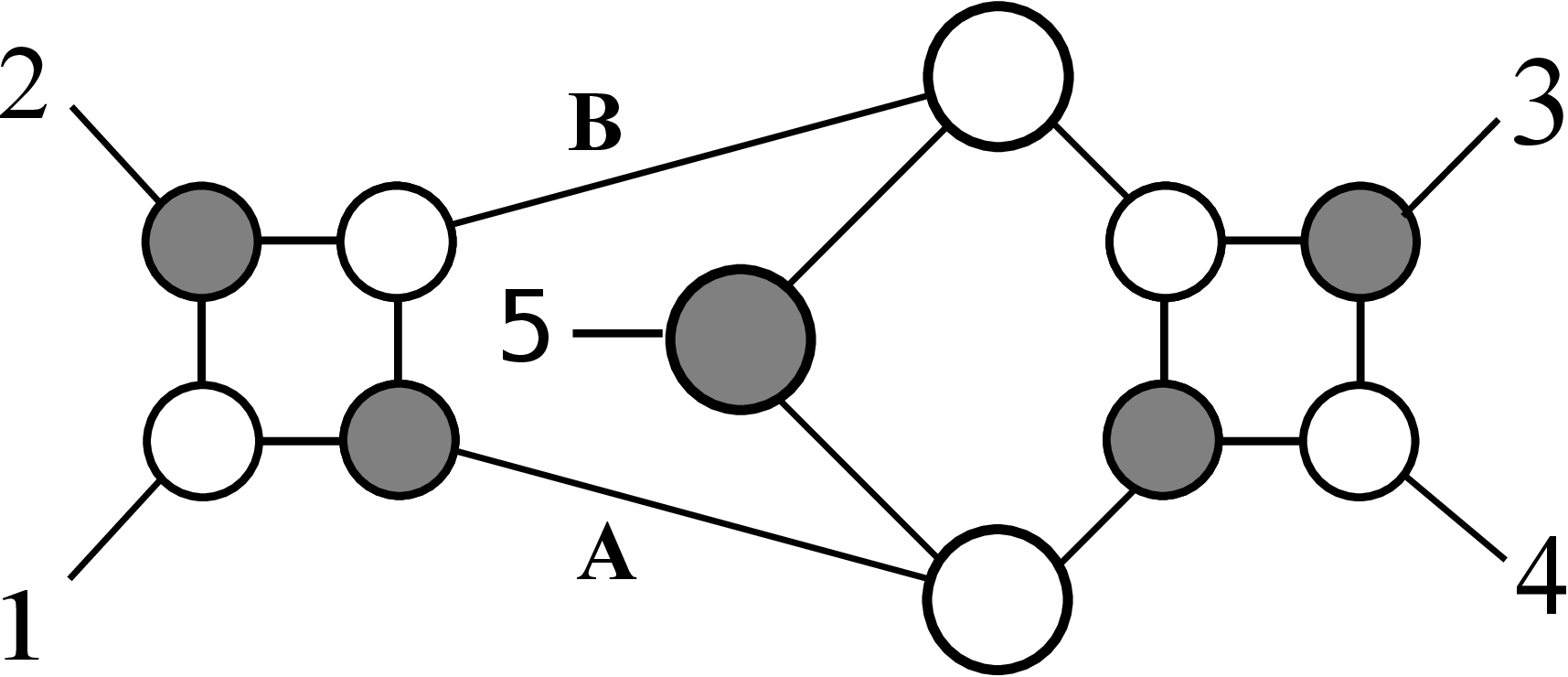}
 \caption{On-shell diagram of cut amplitude $\mathcal{A}_{c}(1,2 | 3,4, \dot{5})$}
\end{figure}
\end{itemize}
For the cut of Type I, the tree level part of nonplanar leg is four point amplitude. While for the Cut Type II, the tree level part of nonplanar leg is five point amplitude. In both cases, all the tree part under a  unitarity cut are MHV amplitudes.

Since  Case~I  is  actually the same as the four-point case, we can simply write down the result:
 \begin{eqnarray}\label{cut5I}
A_{c}(1, \dot{5} | 2,3,4)&=&{s_{15}\over s_{l5}}A_c(51|234)\nb\\
\mathcal A_1(1, \dot{5} | 2,3,4)&=&\mathcal A_1(1, 5 | 2,3,4)+\mathcal A_1(5,1| 2,3,4)
\end{eqnarray}

In  Case~II, we take $\mathcal{A}_{c}(1,2 | 3,4, \dot{5})$ as an example and obtain the relation in bipartite on-shell diagram as Fig.~\ref{fig:5pointloop}
\begin{figure}[ht!]
  \centering
 \includegraphics[width=0.83\textwidth]{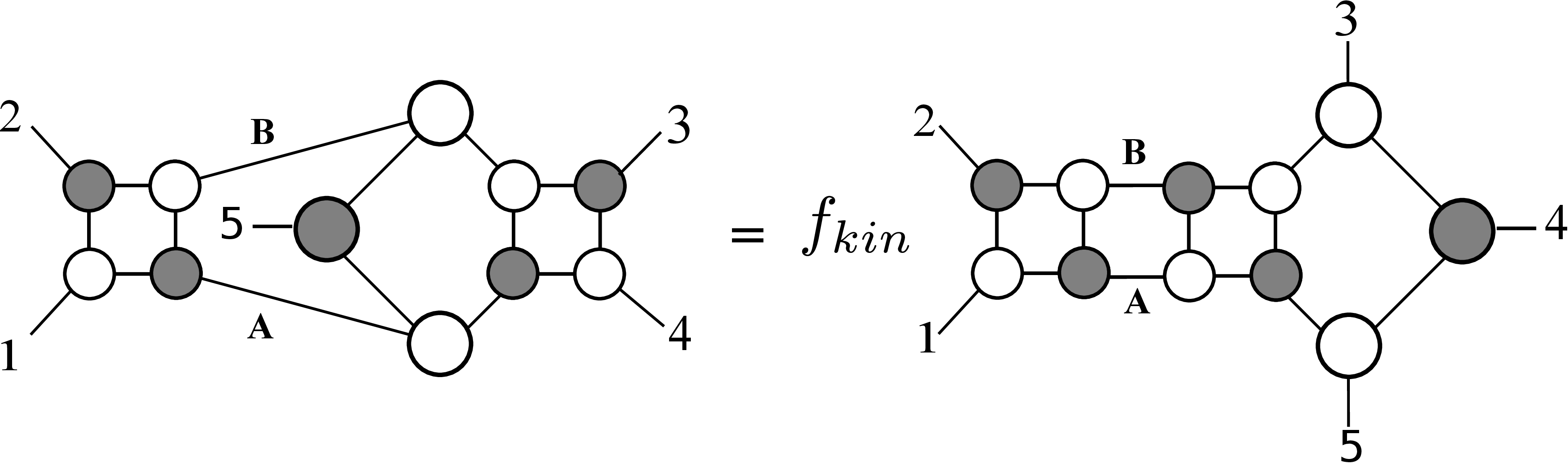}
 \caption{Permutation relation in $\mathcal{A}_{c}(1,2 | 3,4, \dot{5})$ }
 \label{fig:5pointloop}
\end{figure}

The planar part can be done with the same strategy above, and covert to box integrand. Here we first deal with the kinematic factor $f_{kin}\, =\frac{\left \langle 45\right \rangle\left \langle l\bar{l}\right \rangle}{\left \langle{l4}\right \rangle\left \langle\bar{l}5\right \rangle}$

Considering different color order, we can get
\begin{eqnarray}
f_{kin}\,  =-1+\frac{\left\langle A5B3\right\rangle^{o1}}{\left\langle A4B5\right\rangle^{o2}}\frac{\left \langle 45\right \rangle}{\left \langle 35\right \rangle}+\frac{\left\langle A5B3\right\rangle^{o1}}{\left\langle A4B3\right\rangle^{o3}}\frac{\left \langle 34\right \rangle}{\left \langle 35\right \rangle}
\end{eqnarray}
Then we can simply write
$$
\mathcal A_0(1,2|3,4,5)
=\mathcal A_{5}^{tree_{MHV}}
\frac{-\left\langle1235\right\rangle^2}
    {\left\langle AB12\right\rangle
      \left\langle AB23\right\rangle
      \left\langle AB35\right\rangle
      \left\langle AB51\right\rangle}.
$$
Combining $f_{kin}$ and $\mathcal{A}_0(1,2|3,4,5)$ we  get
\begin{eqnarray}
\mathcal{A}_0(1,2| 3,4,\dot{5})
&=&f_{kin}\, \mathcal{A}_0(1,2|3,4,5)\nn
&=&-\,(\mathcal{A}_0(1,2|3,4,5)
       +\mathcal{A}_0(1,2|3,5,4)
       +\mathcal{A}_0(1,2|5,3,4))
\end{eqnarray}
Upon eliminating unphysical poles $\mathcal A_1$ becomes
\begin{equation}
\mathcal{A}_1(1,2|3,4,5)
=\frac{\left\langle AB(234)\cap(451)\right\rangle
        \left\langle 3512\right\rangle}
      {\left\langle AB12\right\rangle
        \left\langle AB23\right\rangle
        \left\langle AB34\right\rangle
        \left\langle AB45\right\rangle
        \left\langle AB51\right\rangle}
\end{equation}

Uniting $\mathcal{A}_1(1,2|3,4,\dot 5)$ and $\mathcal{A}_1(2,3|4,1,\dot 5)$ based on the same method in four-point case, we can find that every term is only cut-related with the results of its own color order,  while not affect those of other orders under unitary cuts.  The final non-planar integrand is the union of all five cuts,
\begin{eqnarray}
\mathcal{A}_{NP}(1,2,3,4,\dot{5})=\bigcup\limits_{i}{f_{kin}}^i \mathcal{A}_1(i,i+1|i+2,i+3,\dot5).
\end{eqnarray}

This result contains all of the possible results of unitary cuts in all possible color orders $(i,i+1,i+2,i+3,5),i=1,2,3,4$ of planar amplitudes. For instance, we can find all unitary cuts $\mathcal{A}_1(i,i+1|i+2,i+3,i+4),i=1,2,3,4,5$ of order (1,2,3,4,5).  According to the discussion in planar MHV amplitudes, the union of these cuts can get $\mathcal A_P(1,2,3,4,5)$. In the same way, we can get $\mathcal A_P(1,2,3,5,4)$, $\mathcal A_P(1,2,5,3,4)$, $\mathcal A_P(1,5,2,3,4)$. Since results from different orders do not affect each other under unitary cuts, the union of all non-planar results equals to the sum of the unions of every order
\begin{eqnarray}
\mathcal A_{NP}(1234\dot5)=\mathcal A_P(12345)+\mathcal A_P(12354)+\mathcal A_P(12534)+\mathcal A_P(15234).
\end{eqnarray}
This equation is just the $U(1)$ decoupling relation for  amplitudes in the $U(N)$ Yang-Mills theory.

This method can also be applied to one-loop MHV non-planar amplitudes with $k$($k>1$) non-planar legs. The conversion from non-planar diagram to planar one using permutation relation of Yangian invariants will be applied successively $k$ times to arrive at  the final results.
For instance,  six  four-point planar amplitudes  with different orders arise  in the case of  $\mathcal A_{NP}(1,2,\dot3,\dot4)$ after unitarity cut while twelve planar   amplitudes in  the case of $\mathcal A_{NP}(1,2,3,\dot4,\dot5)$

\section{NMHV nonplanar amplitude from  generalized  unitarity cuts}
\label{sec:NMHV}

Although the  general  procedures presented in Section~\ref{subset:UCut} above  can be applied to NMHV nonplanar amplitudes.
However more involved procedures are called for to cancel the non-physical poles for NMHV amplitudes. Other interesting physics may arise in the process, which  we will leave  to a future investigation.
We  choose, instead,  generalized unitarity cuts~\cite{2011JPhA...44S4003B} to tackle the problem of NMHV amplitudes because of the absence of non-physical poles.

In this section we will present results of a 6-point one loop NMHV nonplanar amplitude, by generalized  unitarity cuts (quadruple cuts), in the invariant top form.
It is convenient to see the geometric structures of the amplitudes from this invariant top form~\cite{NimaGrass}.
In order to write the total amplitude in the top form   the newly found  permutation relation of the Yangian Invariants again comes in handy.

Before investigating  specific  examples
we  propose the general procedures of constructing total amplitudes for one loop nonplanar Feynman diagrams.  At one loop level the planar diagrams corresponds to the single-trace partial amplitudes in color-order decomposition. The planar on-shell diagrams are associated with the $(k\times n)$ Grassmannian Matrices~\cite{NimaGrass}, $C$,
\begin{equation}\label{eq:}
C=\left( \begin{array}{cccc}
 c_{11}&c_{12}&\cdots& c_{1n} \\
 c_{21}&c_{22}&\cdots& c_{2n} \\
 \vdots&\vdots&\ddots& \vdots \\
 c_{k1}&c_{k2}&\cdots& c_{kn}
\end{array}.\right)
\end{equation}
 It is of convenience to view the Grassmannian cell $C$ as a collection of  $k$-dimensional columns $\{\vec{c}_1,\vec{c}_2\cdots \vec{c}_n\}$. There are $k\times (n-k)$ parameter for a generic Grassmanian matrix $C$, of which $2n-4$ parameter are determined by  the $\delta$-functions $\delta(C\cdot\vec{\tilde\lambda})$ and $\delta(\vec{\lambda}\cdot C^{\bot})$ in (\ref{eq:IntForm}) and others are determined by  the linear-structures of the on-shell reduced diagrams.
 As discussed in~\cite{NimaGrass}, for planar diagrams, such linear-structures of linear-dependencies among consecutive chains of columns  is known as \textrm{positroid stratification}~\cite{2006math......9764P, 2011arXiv1111.3660K}.  The top form of correct singularities should be
 \begin{eqnarray}
\label{eq:}
\Omega={d^{k\times n}C\over vol(GL(k))}{1\over (1\cdots k)\cdots (n\cdots k-1)},
\end{eqnarray}
where $(1\cdots k)$ is the minor of matrix $\{\vec{c}_1 \cdots \vec{c}_k\}$. Hence  the top forms of planar diagrams are characterized by consecutive minors of the Grassmannian matrix.

Corresponding to double-trace partial  amplitudes, every on-shell bipartite diagram of the leading singularity in a double-trace partial amplitudes is also associated with Grassmannian cell $C^{k\times n}=\{\vec{c}_1,\cdots, \vec{c}_j, \vec{c}_{\dot{j+1}} \cdots \vec{c}_{\dot{n}}\}$. And $C$ are determined by the $\delta$-function in (\ref{eq:IntForm}) and the linear-structures of the on-shell reduced diagrams. Since such nonplanar diagrams are endowed with two cyclic-orderings for $\{\vec{c}_1,\cdots, \vec{c}_j\}$ and $\{\vec{c}_{\dot{j+1}} \cdots \vec{c}_{\dot{n}}\}$ respectively.  By hunch, the linear-structures is the linear-dependencies among the chains of columns with consecutive legs with respect to each cyclic-ordering, and hence is a stratification of $G(k,n)$.  Since the linear-dependencies are characterized by the minors.  The external leg indexes of the columns in each minor, if they are in the same trace,  should be consecutive.  Hence the proper minors are
\begin{eqnarray}
\label{eq:Minors}
\mathcal{M}&=&\{(1\cdots k) \cdots (j\,\cdots\,  k-1), (\dot{j+1}\,\cdots \,\dot{j+k})\cdots (\dot{n}\,\cdots\, \dot{j+k-1})\\
&&\bigcup_{k_1+k_2=k}(1\, \cdots\, k_1\, \dot{j+1}\, \cdots \, \dot{j+k_2})\cdots \bigcup_{k_1+k_2=k}(j\, \cdots\, k_1-1\, \dot{n}\, \cdots \, \dot{j+k_2-1})\}.\nb
\end{eqnarray}

The linear-dependencies of an on-shell diagram  correspond to $k\times (n-k)-(2n-4)$ minors $m_i$ in $\mathcal{M}$ vanishing.   The set of these $k\times (n-k)-(2n-4)$ minors is denoted as $M_I$.   As conclusion, the on-shell bipartite diagram of the leading singularities of the loop amplitudes correspond to the Grassmannian cell $C^{k\times n}$ whose values  can  be completely  fixed  by the $\delta$-function in (\ref{eq:IntForm})
and the constraints $m_i=0$ for arbitrary $m_i\in M_I$.
In this way terms with the same Grassmannian geometry
are collected together.  We shall be using these properties in a crucial way to study the singularity structures of nonplanar NMHV amplitudes in this section. We shall as before use a specific example, in this case a 6-point NMHV nonplanar amplitude, to assist a general discussion whenever appropriate.

For the on-shell diagrams with same linear-structures characterized by  $M_I$, we define a function
\begin{eqnarray}
\label{eq:NMHVTopForm}
F^{g}_{M_I}=\oint\limits_{C\in \bar G_{m_i=0, \forall m_i\in M_I }}{d^{k\times n}C\over vol(GL(3))}{1\over \mathcal{P}_g}\delta^{k\times 4}(C\cdot \tilde\eta) \delta^{k\times 2}(C\cdot \tilde\lambda)\delta^{2\times k}(\lambda\cdot C^{\bot}).
\end{eqnarray}
where $\bar G_{m_i=0}$ is a subset in $G(k,n)$ with $m_i=0$, and   $\oint\limits_{C\in \bar G_{m_i=0}}$
picks up  the residue on one minor $m_i=0$ upon an integration
 along  the contour, $C$ in $\bar G_{m_i=0}$.
In order  to form an  invariant top form and to include all  existent poles,
 $\mathcal{P}_g\equiv \prod_{i=1}^n m_{g_i}$ and
 $\mathcal{P}_g$ should scale uniformly as
 $\mathcal{P}_g(t C)=t^{k\times n} \mathcal{P}_g(C)$ and contain  all the factors  $m_i\in M_I$.

We therefore propose a general formula for nonplanar one-loop diagram in  the invariant top form
\begin{eqnarray}
\label{eq:GSumTopForm}
\mathcal{A}_n^{k}(1\cdots j, \dot{j+1}\cdots \dot{n})=\sum_{M_I\subset\mathcal{M} } \mathcal{F}_{M_I}(1\cdots j, \dot{j+1}\cdots \dot{n}),
\end{eqnarray}
where $\mathcal{F}_{M_I}(1\cdots j, \dot{j+1}\cdots \dot{n}) \equiv \sum_{g} N_g F^g_{M_I}$.
The sum  runs  over all the top forms with poles on the hypersurfaces, defined by $m_i=0$,   in $G(k,n)$.
The coefficients, $N_g$, do not depend on $C$.

We consider a 6-point one-loop amplitude $\mathcal{A}(1,2,3,4,5,\dot 6)$ with ``$6$''  being the nonplanar leg.
More  general one-loop amplitudes and higher-loop amplitudes will be left to a future publication.
The  set of minors  is
\begin{eqnarray}
\label{eq:Minors}
\mathcal{M}=\{(123), (234), (345), (451), (512), (612), (623), (634), (645), (651)\}.
\end{eqnarray}
with $M_I$ containing  only one element in  $\mathcal{M}$. Then
\begin{eqnarray}
\label{eq:NMHVTopForm}
F^j_{M_I}=\oint\limits_{C\in \bar G_{m_i=0, \forall m_i\in M_I }}{d^{3\times 6}C\over vol(GL(3))}{1\over \mathcal{P}_j}\delta^{3\times 4}(C\cdot \tilde\eta) \delta^{3\times 2}(C\cdot \tilde\lambda)\delta^{2\times 3}(\lambda\cdot C^{\bot}),
\end{eqnarray}

All possible products of minors are listed,
\begin{equation}\label{eq:minor-products}
\left\{
\begin{array}{ccccccc}
 \mathcal{P}_1&=&(123) (234) (345) (645) (651) (612)& &
 \mathcal{P}_6&=&(123) (234) (451) (623) (645) (651) \\
\mathcal{P}_2&=& (234) (345) (451)(623) (651) (612)&&
\mathcal{P}_7&=& (234) (345) (512) (634) (651) (612)\\
\mathcal{P}_3&=& (345) (451) (512) (623) (612) (634)&&
\mathcal{P}_8&=& (345) (451) (123) (645) (623) (612) \\
\mathcal{P}_4&=& (451) (512) (123) (623) (645) (634)&&
\mathcal{P}_9&=& (451) (512) (234) (651) (623)  (634)\\
\mathcal{P}_5&=& (512) (123) (234) (645) (651) (634) &&
\mathcal{P}_{10}&=& (512) (123) (345) (612) (645) (634).
\end{array}
\right.
\end{equation}

We need to verify that the amplitude is in the form
\begin{eqnarray}
\label{eq:SumTopForm}
\mathcal{A}(12345\dot 6)=\sum_{m_i\in \mathcal{M}} \mathcal{F}_{M_I}(12345\dot 6)\equiv\sum_{M_I\in \mathcal{M}} \sum_{j}^{10} N_j F^j_{M_I}.
\end{eqnarray}
And we also need to determine the $N_j$.
To this end  we classify all the leading singularities
into three type:
\begin{itemize}
\item[Type I:]
The nonplanar leg belongs to a 3-point amplitude after  a quadruple  cut, as shown in Fig.~\ref{fig:NMHVLead1}, which  are the same as a planar diagram up to a  minus sign.
\item[Type II:]
The nonplanar leg belongs to a 4-point amplitude after a quadruple  cut, as shown in Fig.~\ref{fig:NMHVLead2}, which can be transformed into a planar diagram up to a  kinematic factor.
\item[Type III:]
The nonplanar leg lies  in a 5-point amplitude after a
quadruple  cut, as shown in Fig.~\ref{fig:NMHVLead3},   which  can also be transformed into  a planar diagram up to an overall coefficient.
\end{itemize}

These three types of singularities are presented in
Fig.~\ref{fig:NMHVLead1}, Fig.~\ref{fig:NMHVLead2},
and Fig.~\ref{fig:NMHVLead3}, respectively.
\begin{figure}[ht!]
  \centering
 \includegraphics[width=0.83\textwidth]{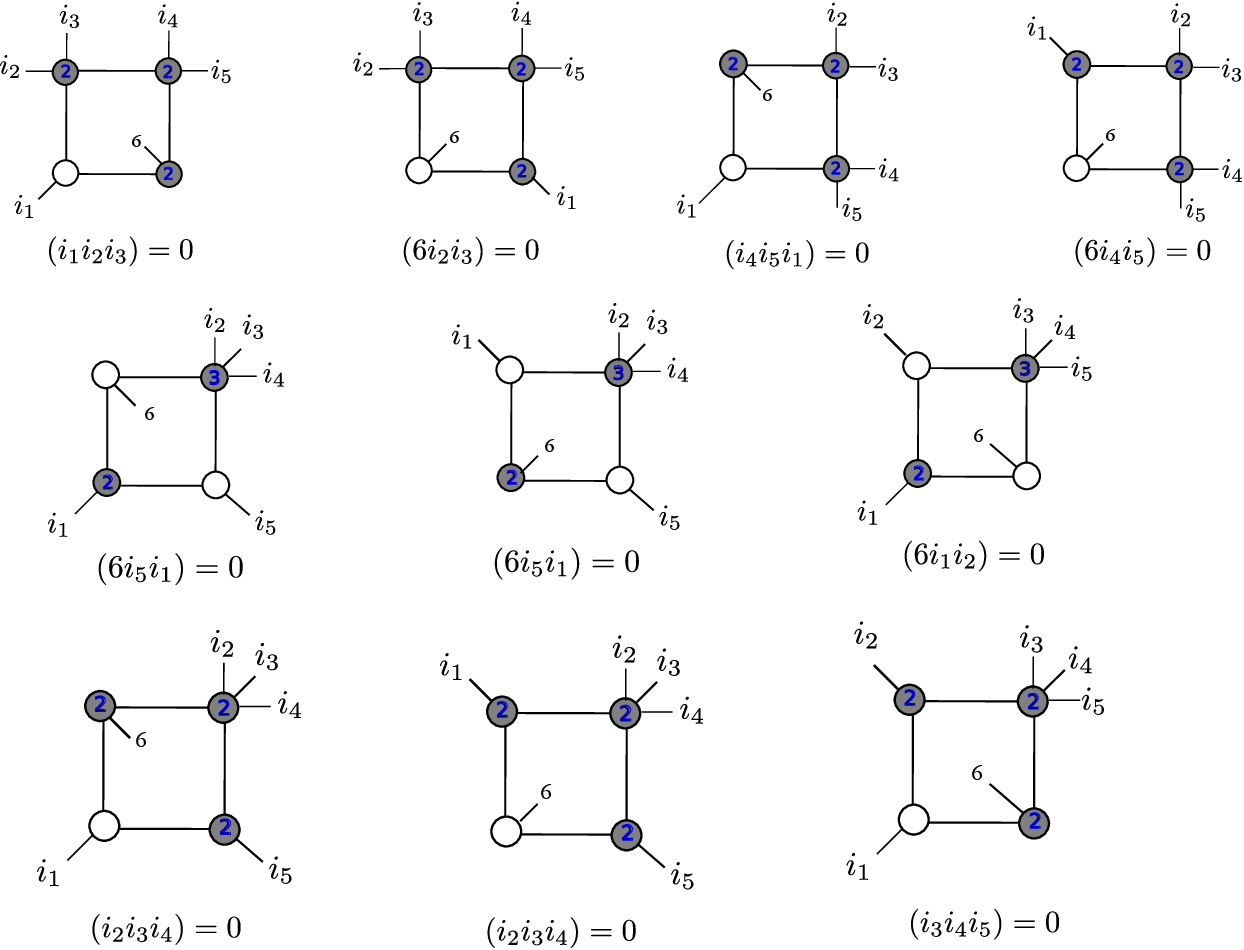}
\caption{Leading singularity Type I}
\label{fig:NMHVLead1}
\end{figure}
\begin{figure}[ht!]
  \centering
 \includegraphics[width=0.63\textwidth]{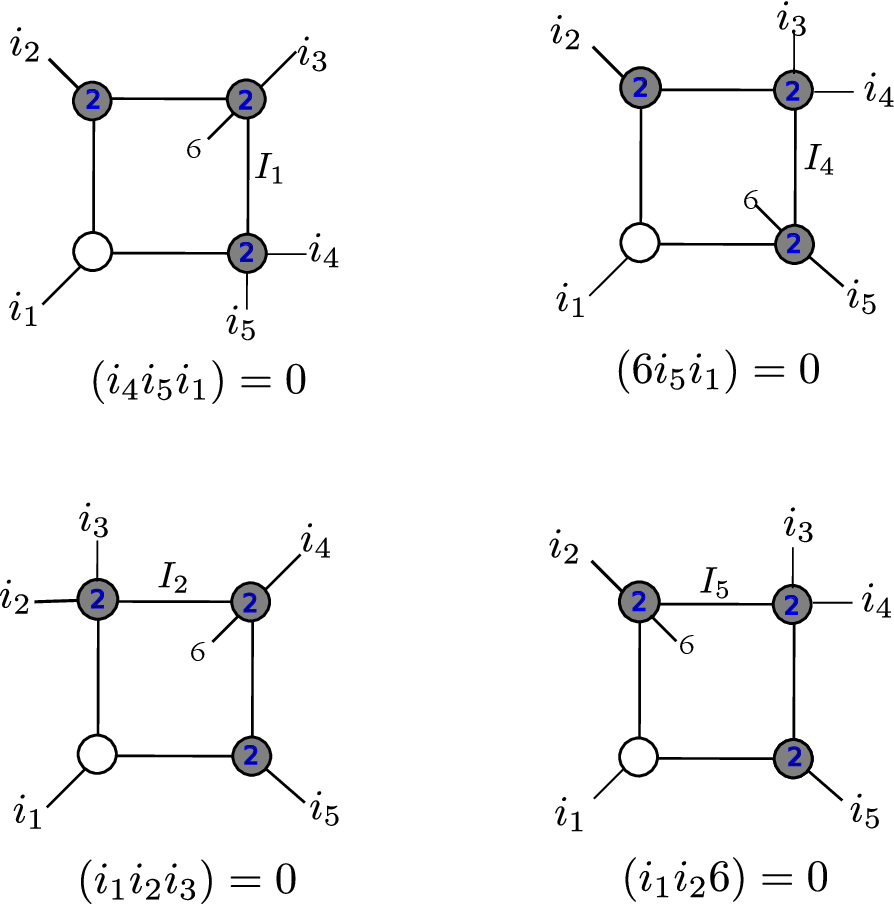}
\caption{Leading singularity Type II}
\label{fig:NMHVLead2}
\end{figure}
\begin{figure}[ht!]
  \centering
 \includegraphics[width=0.63\textwidth]{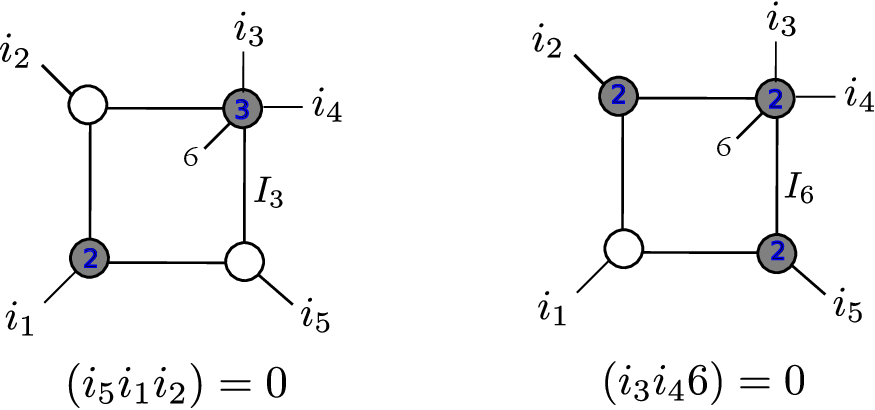}
\caption{Leading singularity Type III}
\label{fig:NMHVLead3}
\end{figure}

To arrive at the top form (\ref{eq:SumTopForm}) we need to group  terms with the same vanishing minor together.
For example, if we consider $(123)=0$,
then according to Fig.~\ref{fig:NMHVLead1}, Fig.~\ref{fig:NMHVLead2}, and Fig.~\ref{fig:NMHVLead3},
we  get
\begin{equation}\label{eq:M123}
\mathcal{F}_{(123)}(12345|6)=\left\{ \begin{array}{lc}
 -(\mathcal{I}_{\{1,23,45,6\}}+\mathcal{I}_{\{5,6,123,4\} })\mathcal{T}^{(123)=0}_{\{123456\}}-\mathcal{I}_{\{6,5,123,4\}}\mathcal{T}^{(123)=0}_{\{123465\}}&~~\\
 -(\mathcal{I}_{\{3,6,45,12\}}+\mathcal{I}_{\{4,5,123,6\}})\mathcal{T}^{(123)=0}_{\{123645\}}& ~~~\textrm{Type I} \\
+{s_{56}\over s_{5I_1}}\mathcal{I}_{\{3,4,56,12\}}\mathcal{T}^{(123)=0}_{\{123456\}} +{s_{46}\over s_{4I_2}}\mathcal{I}_{\{1,23,46,5\}}\mathcal{T}^{(123)=0}_{\{123465\}} &~~~ \textrm{Type II} \\
+{s_{\hat 56}\over s_{\hat 5I_3}}\mathcal{I}_{\{2,3,456,1\}}\mathcal{T}^{(123)=0}_{\{123456\}}& ~~~ \textrm{Type III},
\end{array}\right.
\end{equation}
where $$p_{I_1}={(p_1+p_2+p_3)|\lambda_3\rangle\bigotimes [\tilde\lambda_4|(p_1+p_2+p_3)\over  [\tilde\lambda_4|(p_1+p_2+p_3)|\lambda_3\rangle},$$
$$p_{I_2}={(p_1+p_2+p_3)|\lambda_1\rangle\bigotimes [\tilde\lambda_5|(p_1+p_2+p_3)\over  [\tilde\lambda_5|(p_1+p_2+p_3)|\lambda_1\rangle},$$
$$p_{I_3}={(p_1+p_2)|\lambda_3\rangle\bigotimes [\tilde\lambda_2|(p_1+p_2)\over  [\tilde\lambda_2|(p_1+p_2)|\lambda_3\rangle},$$
$$p_{\hat 5}={(p_4+p_5)|\lambda_4\rangle\bigotimes \langle\lambda_1|(p_2+p_3)(p_4+p_5)\over \langle\lambda_1|(p_2+p_3)(p_4+p_5)|\lambda_4\rangle}.$$
Here $\mathcal{I}_{\{1,23,46,5\}}$  is the scalar integration
\begin{eqnarray}
\label{eq:Scalar}
\mathcal{I}_{\{1,23,46,5\}}=\int{(p_1+p_2+p_3)^2 (p_1+p_4+p_6)^2\over l^2 (l+p_1)^2 (l+p_1+p_2+p_3)^2 (l-p_5)^2}
\end{eqnarray}
and $\mathcal{T}^{(123)=0}_{\{123645\}}$ is the cyclic integration around the pole $(123)=0$ of the
top-form~\cite{NimaGrass} of the tree amplitudes with color ordering
$\{123645\}$
\begin{eqnarray}
\label{eq:TreeTopForm}
\mathcal{T}^{(123)=0}_{\{123645\}}=\oint\limits_{C\in \bar G_{(123)=0 }}{d^{3\times 6}C\over vol(GL(3))}{\delta^{3\times 4}(C\cdot \tilde\eta) \delta^{3\times 2}(C\cdot \tilde\lambda)\delta^{2\times 3}(\lambda\cdot C^{\bot})\over (123)(236)(364)(645)(451)(512)},
\end{eqnarray}
and it works similarly  for others.
For Type II and Type III,
the coefficients in (\ref{eq:M123}) are obtained by the permutation relation of a ``box.''  As it is  explained in detail  in Section~\ref{sec:PRBox}, the permutation relation do not change the geometry of the Grassmannian cell.

Similarly,  according to Fig.~\ref{fig:NMHVLead1},
Fig.~\ref{fig:NMHVLead2}, and Fig.~\ref{fig:NMHVLead3},
a sum  of the terms with Grassmannian geometry $(612)=0$ is
\begin{equation}
\label{eq:M612}
\mathcal{F}_{(612)}(12345\dot 6)=\left\{ \begin{array}{lc}
 -(\mathcal{I}_{\{6,12,34,5\}}+\mathcal{I}_{\{1,2,345,6\} })\mathcal{T}^{(612)=0}_{\{123456\}}-\mathcal{I}_{\{6,2,345,1\}}\mathcal{T}^{(612)=0}_{\{162345\}}&~~\\
 -(\mathcal{I}_{\{6,3,45,12\}}+\mathcal{I}_{\{6,345,1,2\}})\mathcal{T}^{(612)=0}_{\{126345\}}& ~~~\textrm{Type I} \\
+{s_{16}\over s_{1I_4}}\mathcal{I}_{\{2,3,45,61\}}\mathcal{T}^{(612)=0}_{\{123456\}} +{s_{26}\over s_{2I_5}}\mathcal{I}_{\{1,26,34,5\}}\mathcal{T}^{(612)=0}_{\{126345\}} &~~~ \textrm{Type II} \\
+{s_{\hat 26}\over s_{\hat 2I_6}}\mathcal{I}_{\{4,5,126,3\}}\mathcal{T}^{(612)=0}_{\{126345\}}& ~~~ \textrm{Type III},
\end{array}\right.
\end{equation}
where
$$p_{I_4}={(p_1+p_2+p_6)|\lambda_2\rangle\bigotimes [\tilde\lambda_3|(p_1+p_2+p_6)\over  [\tilde\lambda_3|(p_1+p_2+p_6)|\lambda_2\rangle},$$
$$p_{I_5}={(p_1+p_2+p_6)|\lambda_1\rangle\bigotimes [\tilde\lambda_5|(p_1+p_2+p_6)\over  [\tilde\lambda_5|(p_1+p_2+p_6)|\lambda_1\rangle},$$
$$p_{I_6}={(p_3+p_4)|\lambda_4\rangle\bigotimes [\tilde\lambda_5|(p_3+p_4)\over  [\tilde\lambda_5|(p_3+p_4)|\lambda_4\rangle},$$
$$p_{\hat 2}={(p_1+p_2)(p_4+p_5)|\tilde\lambda_3]\bigotimes [\tilde\lambda_1|(p_1+p_2)\over [\tilde\lambda_1|(p_1+p_2)(p_4+p_5)|\tilde\lambda_3]}.$$
All other terms can be generated by cyclic permutations
$Z_5$ of $\{12345\}$.
And the  total amplitude can be written as
\begin{eqnarray}
\label{eq:SumSimpleForm}
\mathcal{A}(12345\dot 6)&=&\sum_{\sigma\in Z_5} \mathcal{F}_{\left(\sigma(1) \sigma(2) \sigma(3)\right)}(\sigma(1)\sigma(2)\sigma(3)\sigma(4)\sigma(5)|6)\nb\\&&+\mathcal{F}_{\left(6 \sigma(1) \sigma(2)\right)}(\sigma(1)\sigma(2)\sigma(3)\sigma(4)\sigma(5)|6).
\end{eqnarray}
The coefficients $N_j$ are  obtained  by comparing (\ref{eq:SumTopForm}), (\ref{eq:M123}), (\ref{eq:M612})
with (\ref{eq:SumSimpleForm}).
An interesting observation is that  all the coefficients of the top-form for $\mathcal{P}_6\cdots \mathcal{P}_{10}$ vanish,  which, in turn, serves as a direct verification of our proposition (\ref{eq:SumTopForm}).

\section{Conclusion and Outlook}
\label{sec:Conclusion}

In this paper we present  a new and useful permutation relation of Yangian Invariants. Different from  KK and BCJ relations working  at the level of amplitudes, it unveils a relation between two Yangian Invariants with two consecutive legs exchanged.
Interesting properties governing the permutations of Yangian Invariants
can be uncovered in the bipartite on-shell diagram.
For instance all Yangian Invariants have at least one ``box'' connecting to two external legs.
When these two legs are exchanged  the Grassmannian matrix does not  change but maintain the same geometric property.
The two Yangians are related by a simple kinematic factor which can be calculated recursively by BCFW method.

However, it is not always obvious to find the ``box'' due to the equivalence
of bipartite on-shell diagrams.
To this end we  give a simple criterion from the associated permutation to check whether a given pair of consecutive legs are connected to a ``box''.
Because all consecutive legs in MHV amplitudes\footnote{MHV amplitudes have only one Yangian Invariant which is itself} lie in a ``box,'' we can exchange any two legs at the expense of  the kinematic factor.
Most  importantly, for a general diagram, if we exchange two lines--either internal or external--connecting to a ``box'' the geometry of the underlying Grassmannian  will not be affected.
This property  can be interpreted  as a new generator of new kind of equivalence relation in bipartite on-shell diagram--other than the square moves and mergers already observed in~\cite{NimaGrass}.

In the  case of NMHV amplitudes there will be a special case--but only one case--that cannot be molded  into  a ``box.'' There arises a second basic building block in bipartite diagram,
a ``bridged twin-box'' (Fig.~\ref{fig:NMHVBiBox}), the permutation relation of which is discussed in Section~\ref{sec:BridgeBiBox}.
With these  two permutation relations  we can resolve all permutations in NMHV amplitudes in the process of constructing their total on-shell integrals.

In this paper we also present a systematic way to deal with the integrands of scattering amplitudes using  unitarity cuts.
Momentum twistor space is a natural language to reconstruct integrand without unphysical propagators.  We discover a new way to add BCFW bridges  and a new operation called ``union'' is introduced to combine results from different cuts to arrive at the total integrands. For one-loop planar MHV amplitudes our results coincide with those obtained from single cuts. The advantage of our proposal is its easy extension to NMHV and higher loops.

For  nonplanar loop amplitudes
we apply unitarity cuts to fix the loop momenta endowing them  with a reasonable definition in the loop integrand.
A crucial relation between planar and nonplanar elements has been discovered which, in turn, enable us to turn nonplanar components into planar ones at the expense of a simple kinematic factor.
With  on-shell diagrams we present detailed and  systematic constructions of the total integrands  for  four- and  five-point one-loop nonplanar MHV amplitudes.
The kinematic factors as well as the corresponding planar amplitudes are separately dealt with using unitarity conditions. Final results are the ``union'' of all results reconstructed from all possible unitarity cuts.

Generalized unitarity cuts are used to address NMHV amplitudes. With six-point one-loop nonplanar  as an explicit example, the amplitude after quadruple   cuts--with all loop momenta being fixed by the cut constraints--is a leading singularity without any variables.
Interesting geometric properties, nevertheless, can be found in the nonplanar leading singularities: it is the result of top-forms integrating around  different poles.

There is an abundance of interesting open questions generated from  these ideas.
In the next paper we will present findings on the leading singularities in bipartite on-shell diagrams  as well as a systematic way of building these  diagrams in the twistor space.
This way of dealing with leading singularities  lends itself straightforward applications to higher loops.
Furthermore, according to the geometric  properties  of the Yangian Invariants, say,  collinearity or coplanarity of several points,
we can further classify the permutation relations;  and we  will probably find permutation relations of non-adjacent legs.
Moreover, interesting  geometric shapes, such as knots, will appear  in two loops.
Ideas and methods in topology are called for to deal  with higher-loop  nonplanar  amplitudes.
Last but not the least we will apply our  methodology to
$\mathcal N<4$ SYM or  gauge theories in other dimensions.

\acknowledgments
Useful discussions with Nima Arkani-Hamed, Bo Feng,  Yijian Du, Jens Fjelstad,  and Konstantin Savvidy are gratefully acknowledged.
We would also  like to thank Antonio Amariti, Andreas Brandhuber,  Livia Ferro, Song He,  Jan Plefka, and  Ellis Yuan  for helpful communications.
Peizhi Du would like to thank Nima Arkani-Hamed for encouragement.

This research project has been supported in parts by the Jiangsu Ministry of Science and Technology under contract~BK20131264 and by the Swedish Research Links programme of the Swedish Research Council (Vetenskapsradets generella villkor) under contract~348-2008-6049.

We also acknowledge
985 Grants from the Ministry of Education, and the
Priority Academic Program Development for Jiangsu Higher Education Institutions (PAPD).

\appendix
\section{The momentum twistor space}
\label{app:twistor}

The introduction  of momentum twistor space are discussed in~\cite{NimaAllLoop, NimaLocalInt}.  Here we summarize the  basic concepts and some useful identities in momentum twistor space for completeness.
In momentum space, the spinor form~\cite{Berends:1981uq, Berends:1981rb, Kleiss:1985yh, Gunion:1985vca, Xu:1986xb} of  on-shell momentum is $p_{\alpha \dot \alpha}=p_{\mu}\sigma^{\mu}_{\alpha \dot \alpha}=\lambda_{\alpha}\tilde{\lambda}_{ \dot \alpha}$ , satisfying  the constraint $p^2=0$ by construction.
The momentum conservation, $\sum\limits_{i=1}^n p_{i}=0$,  however,   needs to be enforced by $\delta$-functions $\delta(\sum\lambda_{i}\tilde{\lambda}_{i})$  in the scattering amplitudes.
One often  uses the  dual coordinates $x_i$~\cite{Drummond:2006rz}, where $p_i=x_i-x_{i-1}$, in which the momentum conservation $\sum\limits_{i=1}^n p_{i}=0$ is naturally satisfied,
at the expense of  $p_i^2=0$ being  obscured.
These two constraints are, however,  both manifest in momentum twistor space,
with twistor $Z=(\lambda ,\mu)$ satisfying  $\mu_{\dot \alpha}=x_{\alpha \dot \alpha}\lambda^{\alpha}$.

Any $x_i$ in $\mathbb{C}^4$ corresponds to a projective line $(Z_i,Z_{i+1})$ in $\mathbb{CP}^3$.
Two lines $(Z_i-1,Z_i)$ and $(Z_i,Z_i+1)$ intersect at the point $Z_i$ and  the momentum $p^2=(x_i-x_{i-1})^2=0$ is a null vector.
When twistors  are used to build momenta,  the corresponding twistor space is called momentum twistor space~\cite{Penrose:1967wn, Hodges:2010kq}.

$\left\langle Z_i\,  Z_j\, Z_k\, Z_l \right\rangle$ denotes  the determinant of four twistors.
If line $(Z_i Z_j)$ and $(Z_k Z_l)$ corresponds to the spacetime points $x$ and $y$, the determinant is simply
\begin{eqnarray}
\left\langle Z_i\,  Z_j\, Z_k\, Z_l\, \right\rangle=\left\langle \lambda_i \lambda_j\right\rangle\left\langle \lambda_k \lambda_l\right\rangle(x-y)^2,
\end{eqnarray}
where $\left\langle \lambda_i \lambda_j\right\rangle=\epsilon_{\alpha\beta}\lambda^{\alpha}_i\lambda^{\beta}_j$.
In particular  if two lines intersect,  $(x-y)^2=0$,  then the determinant vanishes.
It implies that  these four points are coplanar.

$(abc)$ denotes  the plane spanned by the three points $Z_a$, $Z_b$, $Z_c$,
while  $(ab)\cap(cde)$ denotes  a  point in twistor space where the line, $(ab)$, intersects with  the plane, $(cde)$, and
\begin{eqnarray}
(ab)\cap(cde)=Z_a\left\langle bcde \right\rangle+Z_b\left\langle cdea \right\rangle=-(Z_c\left\langle deab \right\rangle+Z_d\left\langle eabc \right\rangle+Z_e\left\langle abcd \right\rangle)~.
\end{eqnarray}
With this definition  we deduce  that  $(ab)\cap(cde)\, =\, -(cde)\cap(ab)$.

Likewise the line,  $(abc)\cap(def)$,  is the intersection of  two planes $(abc)$ and $(def)$ 
\begin{eqnarray}
(abc)\cap(def)&=&Z_aZ_b\left\langle cdef \right\rangle+Z_bZ_c\left\langle adef \right\rangle+Z_cZ_a\left\langle bdef \right\rangle\nn
&=&\left\langle abcd \right\rangle Z_eZ_f+\left\langle abcf \right\rangle Z_dZ_e+\left\langle abce \right\rangle Z_fZ_d~.
\end{eqnarray}

Here we also give several very useful identities for momentum twistor space called \textit{Schouten identity}. The familiar Schouten identity based on spinors is
\begin{eqnarray}\label{SpinorId}
\bracf{ac}\bracf{bd}=\bracf{ab}\bracf{cd}+\bracf{ad}\bracf{bc}.
\end{eqnarray}
In momentum twistor space, any arbitrary set of five twistors $\{Z_a,Z_b,Z_c,Z_d,Z_e\}$ will satisfy the following identity,
\begin{eqnarray}\label{ABId}
Z_a\bracf{bcde}+Z_b\bracf{cdea}+Z_c\bracf{deab}+Z_d\bracf{eabc}+Z_e\bracf{abcd}=0.
\end{eqnarray}
According to this, we could obtain the 5-term identity also called a Schouten identity:
\begin{eqnarray}\label{GeneralId}
\bracf{fgha}\bracf{bcde}+\bracf{fghb}\bracf{cdea}+\bracf{fghc}\bracf{deab}+\bracf{fghd}\bracf{eabc}+\bracf{fghe}\bracf{abcd}=0.\nn
\end{eqnarray}
We will show another frequently used identity related to A and B, which is very analogous to (\ref{SpinorId}),
\begin{eqnarray}\label{ABId}
\bracf{AB13}\bracf{AB24}=\bracf{AB12}\bracf{AB34}+\bracf{AB14}\bracf{AB23}.
\end{eqnarray}

\bibliographystyle{jhep}
\input{nonplanar_jhep.bbl}
\end{document}

%% file: nonplanar_jhep.bbl
\providecommand{\href}[2]{#2}\begingroup\raggedright\endgroup